\documentclass[twocolumn,english,reprint]{revtex4-1}
\usepackage[T1]{fontenc}
\usepackage[latin9]{inputenc}
\usepackage{babel}
\usepackage{enumitem}
\usepackage{bm}
\usepackage{amsmath}
\usepackage{amsthm}
\usepackage{amssymb}
\usepackage{stmaryrd}
\usepackage{graphicx}
\usepackage[unicode=true,
bookmarks=true,bookmarksnumbered=false,bookmarksopen=false,
breaklinks=false,pdfborder={0 0 1},backref=false,colorlinks=false]
{hyperref}
\hypersetup{
hidelinks}

\makeatletter
\newlength{\lyxlabelwidth} 
\theoremstyle{plain}
\newtheorem{thm}{\protect\theoremname}
\theoremstyle{plain}
\newtheorem{prop}[thm]{\protect\propositionname}
\theoremstyle{plain}
\newtheorem{lem}[thm]{\protect\lemmaname}

\usepackage{enumitem}
\usepackage{xargs}[2008/03/08]
\usepackage{newtxtext}
\usepackage{color}
\usepackage{xpatch}
\usepackage{appendix}
\usepackage{color}
\RequirePackage[dvipsnames,usenames]{xcolor}
\usepackage{chngcntr}
\usepackage{apptools}

\usepackage{bm}

\usepackage[capitalize]{cleveref}
\AtBeginDocument{
\let\oldref\ref
\renewcommand{\ref}{\cref}

\newcommand\mycitep[2][usedefault, addprefix=\global]{\citep[#1,~][]{#2}}%

}

\let\oldbibliography\bibliography
\renewcommand{\bibliography}[1]{
\renewcommand{\ref}{\oldref}
\oldbibliography{#1}
\renewcommand{\ref}{\cref}
}

\newcommand{\myfootref}[1]{Eq.~(\oldref{#1})}

\def\cleartheorem#1{%
\expandafter\let\csname#1\endcsname\relax
\expandafter\let\csname c@#1\endcsname\relax
}

\newif\ifnotlyx
\notlyxfalse
\AtAppendix{\counterwithin{thm}{section}}
\AtAppendix{\counterwithin{lem}{section}}
\AtAppendix{\counterwithin{prop}{section}}
\AtAppendix{\counterwithin{condition}{section}}

\ifnotlyx

\newlength{\lyxlabelwidth} %
\theoremstyle{plain}

\else
\cleartheorem{thm}
\cleartheorem{prop}
\cleartheorem{lem}
\cleartheorem{cor}
\cleartheorem{condition}

\fi
\newtheorem{thm}{\protect\theoremname}\theoremstyle{plain}
\newtheorem{prop}{\protect\propositionname}\theoremstyle{plain}
\newtheorem{lem}{\protect\lemmaname}\theoremstyle{plain}
\theoremstyle{plain}
\newtheorem{condition}{Condition}

\crefname{condition}{Condition}{Conditions}
\crefname{cor}{Corollary}{Corollaries}
\crefname{lem}{Lemma}{Lemmas}
\crefname{prop}{Proposition}{Propositions}
\crefname{thm}{Theorem}{Theorems}

\providecommand{\corollaryname}{Corollary}
\providecommand{\lemmaname}{Lemma}
\providecommand{\propositionname}{Proposition}
\providecommand{\theoremname}{Theorem}

\crefname{condition}{Condition}{Conditions}

\providecommand{\cc}[1]{}

\makeatother

\providecommand{\lemmaname}{Lemma}
\providecommand{\propositionname}{Proposition}
\providecommand{\theoremname}{Theorem}

\begin{document}
\title{Dependence of integrated, instantaneous, and fluctuating entropy
production on the initial state in quantum and classical processes}
\author{Artemy Kolchinsky}
\email{artemyk@gmail.com}

\affiliation{Santa Fe Institute, 1399 Hyde Park Road, Santa Fe, NM 87501, USA}
\author{David H. Wolpert}
\altaffiliation{Complexity Science Hub, Vienna; Arizona State University, Tempe, AZ}

\affiliation{Santa Fe Institute, 1399 Hyde Park Road, Santa Fe, NM 87501, USA}
\begin{abstract}
We consider the additional entropy production (EP) incurred by a fixed
quantum or classical process on some initial state $\rho$, above
the minimum EP incurred by the same process on any initial state.
We show that this additional EP, which we term the ``mismatch cost
of $\rho$'', has a universal information-theoretic form: it is given
by the contraction of the relative entropy between $\rho$ and the
least-dissipative initial state $\varphi$ over time. We derive versions
of this result for integrated EP incurred over the course of a process,
for trajectory-level fluctuating EP, and for instantaneous EP rate.
We also show that mismatch cost for fluctuating EP obeys an integral
fluctuation theorem. Our results demonstrate a fundamental relationship
between \emph{thermodynamic irreversibility} (generation of EP) and
\emph{logical irreversibility} (inability to know the initial state
corresponding to a given final state). We use this relationship to
derive quantitative bounds on the thermodynamics of quantum error
correction and to propose a thermodynamically-operationalized measure
of the logical irreversibility of a quantum channel. Our results hold
for both finite and infinite dimensional systems, and generalize beyond
EP to many other thermodynamic costs, including nonadiabatic EP, free
energy loss, and entropy gain.
\end{abstract}
\maketitle

\section{Introduction}
The second law of thermodynamics states that the total entropy of
a system and any coupled reservoirs cannot decrease during a physical
process. For this reason, the overall amount of entropy production
(EP) is the fundamental measure of the irreversibility of the process
in both classical and quantum thermodynamics~\citep{seifert2012stochastic,deffnerQuantumThermodynamicsIntroduction2019}.

Consider a quantum system coupled to one or more thermodynamic reservoirs.
Suppose the system starts in some initial state $\rho$ and evolves
for a time interval $t\in[0,\tau]$, and that the evolution of the
system's state can be formalized in terms of a quantum channel $\Phi$
that takes initial states to final states, $\rho\mapsto\Phi(\rho)$.
The {integrated EP} incurred during this process can be written
as a function of the initial state $\rho$ as~\citep{esposito2010entropy,deffnerNonequilibriumEntropyProduction2011a,landi2020irreversible}
\begin{equation}
\Sigma(\rho)=S(\Phi(\rho))-S(\rho)+Q(\rho),\label{eq:EP0-ent}
\end{equation}
where $S(\cdot)$ is von Neumann entropy and $Q(\rho)$ is the \emph{entropy
flow}, i.e., the increase of the thermodynamic entropy of the coupled
reservoirs. The precise form of the entropy flow term $Q$ is determined
by the number and characteristics of the coupled reservoirs (for instance,
for a single heat bath at inverse temperature $\beta$, $Q$ is equal
to $\beta$ times the generated heat).

Deriving expressions and bounds for EP has important implications
for understanding the thermodynamic efficiency of various artificial
and biological devices, and it serves as a major focus of research
in nonequilibrium statistical physics~\citep{seifert2012stochastic, landi2020irreversible, van2013stochastic, jarzynski_equalities_2011}.
Some of this research derives exact expressions for EP given a fully
specified protocol and a fixed initial state~\citep{esposito2010entropy,deffnerNonequilibriumEntropyProduction2011a}.
Other research derives bounds on EP in terms of general properties
of the dynamics (e.g., the fluctuations of observables, as in ``thermodynamic
uncertainty relations''~\citep{gingrich2016dissipation, gingrich2017fundamental}).
A third approach considers bounds on EP in terms of various properties
of the driving protocol, such as the driving speed~\citep{sivak2012thermodynamic, esposito2010finite, shiraishi_speed_2018}
or constraints on the available generators~\citep{wilming_second_2016, kolchinsky2020entropy}.

In this paper, we consider the complementary issue, and analyze how the EP incurred
during a fixed physical process depends on the initial state $\rho$.
This question is relevant whenever there is a fixed process
that may be carried out with different initial states. For example,
one can imagine a fixed biological process whose initial state can depend on a fluctuating environment,
and wish to know how its thermodynamic efficiency depends on the state of the environment~\citep{kolchinsky2017maximizing}.
As another example, one can imagine
a fixed computational device whose input distribution can be set by
different users \citep{kolchinsky2017maximizing,kolchinsky2016dependence},
and wish to know how its thermodynamic efficiency depends on the
variability among the users. In a similar vein, one can imagine a feedback-control apparatus that extracts thermodynamic
work from a system, in which there is uncertainty about the
initial statistical state of the observed system. In these cases,
as well as many others, it is useful to know how the amount of EP
changes as the initial state is varied.

The dependence of EP on the initial state is well-understood in some
special cases. In particular, for a free relaxation toward an equilibrium
Gibbs state $\pi$, the EP incurred by initial state $\rho$ is the
drop of the relative entropy between $\rho$ and $\pi$ over time~\citep{breuer2002theory,deffnerNonequilibriumEntropyProduction2011a,landi2020irreversible},
\begin{equation}
\Sigma(\rho)=S(\rho\Vert\pi)-S(\Phi(\rho)\Vert\pi).\label{eq:relax}
\end{equation}
Note that if there are multiple equilibrium states, any one can be
equivalently chosen as the reference equilibrium state $\pi$ in \ref{eq:relax}
(see \footnote{The fact that any equilibrium state can be chosen as the reference
state follows immediately from our results as stated later in the
paper, such as \myfootref{eq:equalityBasis}. Consider any two equilibrium
states \unexpanded{$\pi,\pi'$} and EP \unexpanded{$\Sigma$}
defined relative to reference equilibrium state $\pi$, as in \myfootref{eq:relax}.
Since \unexpanded{$\pi'$} is also an equilibrium state, it must (1) be a minimizer
of \unexpanded{$\Sigma$}, (2) achieve \unexpanded{$\Sigma(\pi')=0$}, and (3) satisfy \unexpanded{$\Phi(\pi')=\pi'$}.
Then, as long as \unexpanded{$S(\rho\Vert\pi')<\infty$}, \myfootref{eq:equalityBasis}
gives \unexpanded{$\Sigma(\rho)=S(\rho\Vert\pi')-S(\Phi(\rho)\Vert\pi')$}, which
means that EP defined relative to reference equilibrium state \unexpanded{$\pi$}
(LHS) is equal to EP defined relative to reference equilibrium state
\unexpanded{$\pi'$} (RHS).}).

In fact, \ref{eq:relax} can be generalized beyond simple relaxations,
to processes with arbitrary driving and/or multiple reservoirs (such
that no equilibrium state exists). In previous work \citep{kolchinsky2016dependence,kolchinsky2017maximizing,wolpert2020thermodynamic}\footnote{See also \citep{kolchinsky2020thermodynamic} for a derivation of
\myfootref{eq:classicalMismatch} for a classical system with a countably
infinite state space but deterministic dynamics.}, we analyzed the \emph{mismatch cost of $\rho$} for a finite-state
classical process, which we defined as the extra integrated EP incurred
by the process on initial distribution $\rho$, in addition to the
EP incurred by the process on the optimal initial distribution that
minimizes EP, $\varphi\in\mathop{\arg\min}_{\omega}\Sigma(\omega)$.
We showed that as long as $\mathrm{supp}\,\rho\subseteq\mathrm{supp}\,\varphi$,
mismatch cost can be expressed as the contraction of relative entropy
between $\rho$ and $\varphi$,
\begin{equation}
\Sigma(\rho)-\Sigma(\varphi)=S(\rho\Vert\varphi)-S(\Phi(\rho)\Vert\Phi(\varphi)).\label{eq:classicalMismatch}
\end{equation}
The right hand side is non-negative by the monotonicity
of relative entropy \citep{muller2017monotonicity} and vanishes if
$\rho=\varphi$. \ref{eq:relax} is a special case of \ref{eq:classicalMismatch},
since in a free relaxation $\varphi$ is the Gibbs equilibrium state
$\pi$, which has full support and obeys $\Sigma(\pi)=0$, $\Phi(\pi)=\pi$.
This relationship is visualized in \ref{fig:1}. \ref{eq:classicalMismatch}
was recently generalized to finite-dimensional quantum processes by
Riechers and Gu~\citep{riechers2020initial,riechersImpossibilityLandauerBound2021}\footnote{Although Ref.~\citep{riechers2020initial} never explicitly states the assumption
of a finite-dimensional Hilbert space, it is implicit in the derivations of that paper.
For example, in infinite dimensional spaces, it cannot be assumed
that the directional derivative can be written in terms of the gradient
(as in the derivation of Theorem 1 in \citep{riechers2020initial}),
that the directional derivative at the optimizer with full support
vanishes (as in Eq. 10 in \citep{riechers2020initial}), or that \unexpanded{$S(\rho\Vert\varphi)<\infty$}
whenever \unexpanded{$\mathrm{supp}\,\rho\subseteq\mathrm{supp}\,\varphi$}.}.

\begin{figure}
\includegraphics[width=0.8\columnwidth]{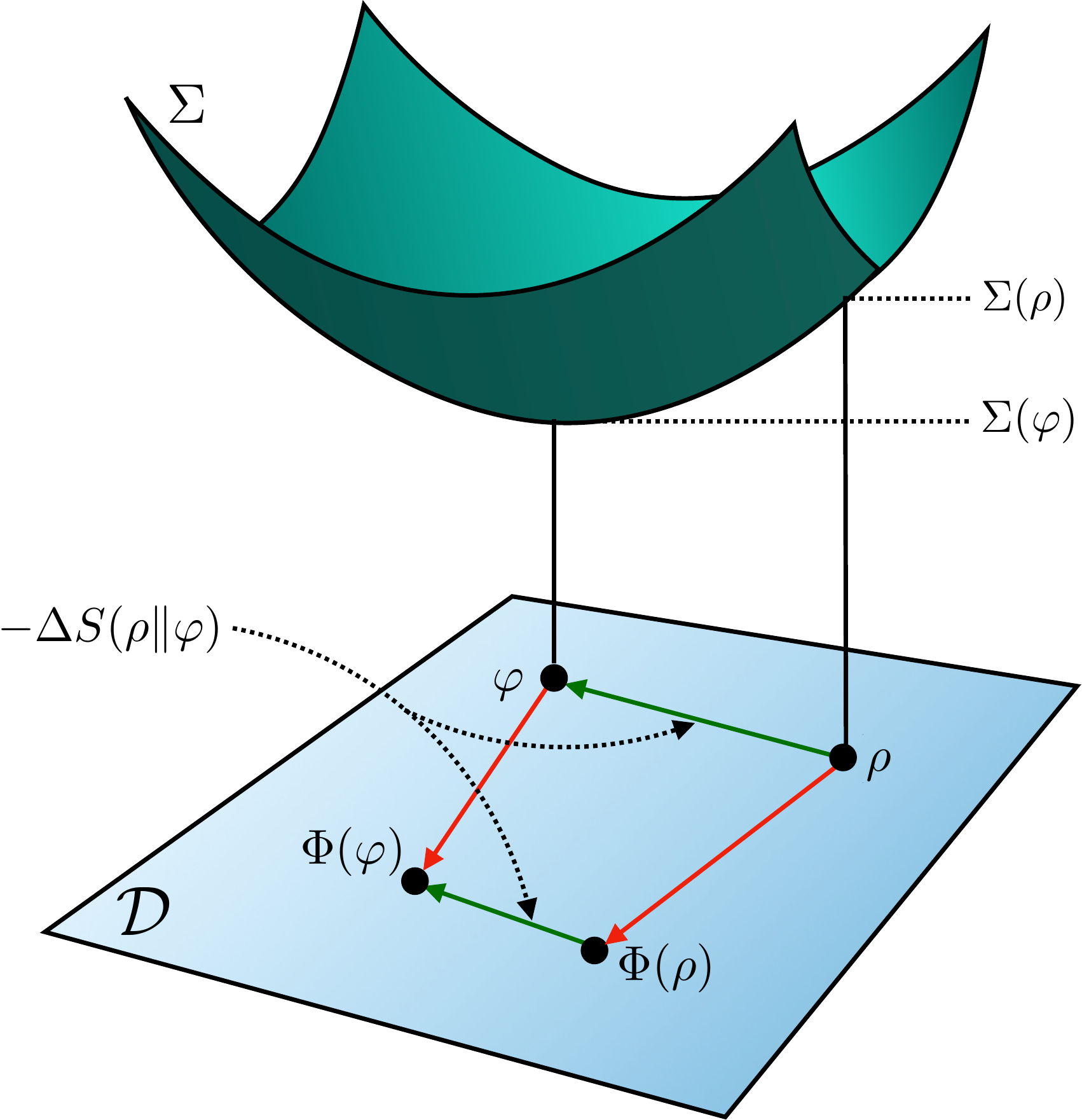}\caption{\label{fig:1}\textbf{Information-theoretic form of mismatch cost.
}The top surface represents the entropy production (EP) $\Sigma$
as a function of the initial state $\rho$, for a physical process
whose dynamics are described by the quantum channel $\Phi$ (red arrows).
The bottom surface represents the set of states $\mathcal{D}$. \ref{eq:classicalMismatch}
says that the extra EP incurred by some initial state $\rho$, additional
to the EP incurred by the optimal initial state $\varphi$ which minimizes
EP, is equal to the decrease of relative entropy between $\rho$ and
$\varphi$ over time (contraction of green arrows).}
\end{figure}

In this paper, we extend these earlier results in several ways:
\begin{itemize}[wide,labelindent=0pt,labelwidth=!]
\item In \ref{sec:integratedEP}, we show that the expression for mismatch
cost in \ref{eq:classicalMismatch} holds for arbitrary quantum systems,
both finite and infinite dimensional, and coupled to any number of
idealized or non-idealized reservoirs. We also show that this expression
applies not only when $\varphi$ is the globally optimal initial state,
but also when $\varphi$ is the optimal incoherent state
(relative to a given set of projection operators), which can
be used to decompose mismatch cost into separate quantum and classical
contributions. Finally, we derive simple sufficient conditions that
guarantee that the optimal initial state $\varphi$ has full support,
which allows \ref{eq:classicalMismatch} to be applied to arbitrary
$\rho$ (since \ref{eq:classicalMismatch} holds only when the support
of $\rho$ falls within the support of $\varphi$).
\item In \ref{sec:Fluctuating-EP}, we analyze mismatch cost for the \emph{fluctuating
EP}, that is the trajectory-level EP generated when a physical process
undergoes stochastically sampled realizations~\citep{campisi2011colloquium}.
We derive an expression for trajectory-level fluctuating mismatch
cost, which can be seen as the trajectory-level version of \ref{eq:classicalMismatch}. We also demonstrate that this expression obeys an integral fluctuation
theorem.
\item In \ref{sec:eprate}, we analyze mismatch cost for the \emph{instantaneous}
\emph{EP rate} incurred at a given instant in time. We show that,
similarly to the case of integrated EP and fluctuating EP, mismatch
cost for EP rate can be expressed in terms of the instantaneous rate
of the contraction of relative entropy between the actual initial
state $\rho$ and the optimal initial state $\varphi$ which minimizes
the EP rate.
\item In \ref{sec:classical}, we discuss our results in the context of
classical systems. In particular, we demonstrate that all of our results
apply to discrete-state and continuous-state classical systems,
where they describe the dependence of classical EP on the choice of
the initial probability distribution.
\end{itemize}
After deriving the above results, in \ref{sec:logicalvsthermo} we
discuss them within the context of thermodynamics of information processing.
In particular, we show that our expressions for mismatch cost imply
a fundamental relationship between \emph{thermodynamic irreversibility}
(generation of EP) and \emph{logical irreversibility} (inability to
know the initial state corresponding to a given final state). We use
this relationship to derive quantitative bounds on the thermodynamics
of quantum error correction, and to propose an operational measure
of the logical irreversibility of a quantum channel $\Phi$, which
provides a lower bound on the worst-case EP incurred by any physical
process that implements $\Phi$.

In \ref{sec:EP-Type-functionals} we show that our results for mismatch
cost apply not only to EP (which is the main focus of this paper)
but in fact to any function that can be written in the general form of
\ref{eq:EP0-ent}, as the increase of system entropy plus some linear
term. Examples of such functions include many thermodynamic costs
of interest beyond EP, including nonadiabatic EP \citep{horowitz2013entropy,horowitz2014equivalent,esposito2010three,manzanoQuantumFluctuationTheorems2018},
free energy loss \citep{kolchinsky2017maximizing,faist2019thermodynamic},
and entropy gain \citep{plastino1995fisher,holevo2011entropy,holevo2011entropyB}.
For any such thermodynamic cost, the extra cost incurred by initial
state $\rho$, additional to that incurred by the optimal initial
state $\varphi$ which minimizes that cost, is given by the contraction
of relative entropy between $\rho$ and $\varphi$ over time.

Before proceeding, we briefly review some relevant prior literature
and introduce some necessary notation. We finish with a brief discussion
in \ref{sec:Discussion}.

\subsection{Relevant prior literature}

In our own prior work~\citep{kolchinsky2016dependence,kolchinsky2017maximizing,wolpert2020thermodynamic},
we derived an expression of mismatch cost for the integrated EP incurred
by a finite-state classical system. In addition, in this earlier work we
showed that mismatch cost has important implications for understanding
the thermodynamics of classical information processing, including
computation with digital circuits \citep{wolpert2020thermodynamic}
and deterministic classical Turing machines \citep{kolchinsky2020thermodynamic}.
Finally, we also used mismatch cost to study the thermodynamics of
free-energy harvesting systems, both in classical and quantum systems
\citep{kolchinsky2017maximizing}.

Riechers and Gu analyzed mismatch cost for integrated EP incurred
by finite-dimensional quantum systems. They used these results to
analyze the thermodynamics of information erasure in finite-dimensional
quantum systems, as well as the ``thermodynamic cost of modularity''
\citep{riechers2020initial,riechersImpossibilityLandauerBound2021}.

An important precursor of mismatch cost appeared in \citep{maroney2009generalizing}.
This paper considered one specific quantum process that carries out
information processing over a set of classical logical states. It
was pointed out that if the protocol is thermodynamically reversible
for some initial distribution $\varphi$ over logical states, then
for any other initial distribution $\rho$ over the logical states,
$\Sigma(\rho)=S({\rho}\Vert{\varphi})-S(\Phi(\rho)\Vert\Phi(\varphi))$
\mycitep[Eq.~168]{maroney2009generalizing}. This can be seen as a
special case of classical mismatch cost, where the optimal state $\varphi$
is thermodynamically reversible (so $\Sigma(\varphi)=0$). A similar
result was derived for a specific classical process in \citep{wolpert_arxiv_beyond_bit_erasure_2015}.
Some related ideas were also discussed in Turgut \citep{turgut_relations_2009}.

\subsection{Notational preliminaries}

\label{sec:Background}

We use $\mathcal{D}$ to indicate the set of all states (i.e., density
operators) over the system's Hilbert space $\mathcal{H}$, which may
be finite or infinite dimensional. For any orthogonal set of projection
operators $P=\{\Pi_{1},\Pi_{2},\dots\}$, we define
\begin{equation}
\mathcal{D}_{P}:=\{\rho\in\mathcal{D}:\rho=\sum_{\Pi \in P}\Pi \rho\Pi \}\label{eq:densBdef}
\end{equation}
as the set of states that are incoherent relative to projectors in
$P$. Note that the set of projection operators $P$ may be complete
($\sum_{\Pi \in P}\Pi =I$) or incomplete ($\sum_{\Pi\in P}\Pi \ne I$).
Special cases of $\mathcal{D}_{P}$ include the set of all states
$\mathcal{D}$ ($P=\{I\}$), the set of states with support limited
to some subspace $\mathcal{H}'\subset\mathcal{H}$ ($P=\{\Pi\}$
such that $\Pi\mathcal{H}=\mathcal{H}'$), and the set of states diagonal
in some orthonormal basis $\{\vert i\rangle\}_{i}$ ($P=\{\vert i\rangle\langle i\vert\}_{i}$).
We write
\begin{equation}
\mathcal{H}_{P}=\mathcal{H}\sum_{\Pi\in P}\Pi\label{eq:hilbB}
\end{equation}
to indicate the Hilbert subspace spanned by the projection operators
in $P$.

{We use the von Neumann entropy of state $\rho\in\mathcal{D}$,
\[
S(\rho):=-\mathrm{tr}\{\rho\ln\rho\}.
\]
We also use the (quantum) relative entropy, defined for any pair
of states $\rho,\varphi\in\mathcal{D}$ as
\begin{align}
S(\rho\Vert\varphi):= & \begin{cases}
\mathrm{tr}\{\rho(\ln\rho-\ln\varphi)\} & \text{if \ensuremath{\mathrm{supp}\,\rho\subseteq\mathrm{supp}\,\varphi}}\\
\infty & \text{otherwise}
\end{cases}\label{eq:relentDef}
\end{align}
For notational convenience, we often write the change of relative
entropy under some quantum channel $\Phi$ as
\begin{align}
\Delta S({\rho}\Vert{\varphi}) & :=S(\Phi(\rho)\Vert\Phi(\varphi))-S(\rho\Vert\varphi).\label{eq:ddDef}
\end{align}
}

Finally, given some quantum channel $\Phi$ and some reference state
$\varphi\in\mathcal{D}$, the Petz \emph{recovery map} is defined
as \mycitep[Sec.~12.3]{wildeQuantumInformationTheory2017}\footnote{The definition in \myfootref{eq:recov} holds for finite dimensional
spaces and \unexpanded{$\rho$} such that \unexpanded{$\mathrm{supp}\,\rho\subseteq\mathrm{supp}\,\varphi$}.
For a more general definition, see \citep{petz1988sufficiency,jungeUniversalRecoveryMaps2018}.}
\begin{equation}
\mathcal{R}_{\Phi}^{\varphi}(\rho):=\varphi^{1/2}\Phi^{\dagger}(\Phi(\varphi)^{-1/2}\rho\Phi(\varphi)^{-1/2})\varphi^{1/2}.\label{eq:recov}
\end{equation}
The recovery map undoes the effect of $\Phi$ on the reference state,
so that $\mathcal{R}_{\Phi}^{\varphi}(\Phi(\varphi))=\varphi$. It
can be seen as a generalization of the Bayesian inverse to quantum
channels \citep{leiferFormulationQuantumTheory2013}.

\section{Mismatch Cost for Integrated EP}

\label{sec:integratedEP}

In our first set of results, we consider the state dependence of integrated
EP, in terms of the additional integrated EP incurred by some
initial state $\rho$ rather than the optimal initial state $\varphi$.

Our results apply to $\Sigma(\rho)$ as defined in \ref{eq:EP0-ent}
in terms of the increase of system entropy plus the entropy flow,
where $\Phi$ is some positive and trace-preserving map and the entropy
flow $Q$ is some linear function (which we assume is lower-semicontinuous).
Our results also apply when $\Sigma(\rho)$ is defined in terms of
an explicitly-modeled system+environment that jointly evolve in a
unitary manner as $\rho\otimes\omega\to U(\rho\otimes\omega)U^{\dagger}$.
In this case, the quantum channel can be expressed in the Stinespring
form as $\Phi(\rho)=\mathrm{tr}_{Y}\{U(\rho\otimes\omega)U^{\dagger}\}$
(where $\mathrm{tr}_{Y}$ indicates a partial trace over the environment),
and EP can be written as
\begin{equation}
\Sigma(\rho)=S(U(\rho\otimes\omega)U^{\dagger}\Vert\Phi(\rho)\otimes\omega).\label{eq:functype2-maintext2}
\end{equation}
This expression for EP often appears in recent work on quantum thermodynamics
\citep{esposito2010entropy,ptaszynskiEntropyProductionOpen2019,landi2020irreversible}.

These two formulations of EP, \ref{eq:EP0-ent} and \ref{eq:functype2-maintext2},
have different advantages and disadvantages. {\ref{eq:EP0-ent}
can be more experimentally accessible since --- unlike \ref{eq:functype2-maintext2}
-- it does not require knowledge of the exact state and evolution
of the environment, only the total amount of entropy flow (e.g., as could be
measured by a calorimeter). For the same reason, \ref{eq:EP0-ent}
is also more appropriate for studying EP for a system coupled to ``idealized''
baths (which have infinite size and instantaneous self-equilibration
\citep{breuer2002theory}). On the other hand, \ref{eq:functype2-maintext2}
is more appropriate for studying EP for a system coupled to more realistic
``non-idealized'' baths (which have finite size and possibly slow
relaxation times).} From a purely mathematical perspective,
the two forms are equivalent for any $\rho$ with finite entropy:
\ref{eq:functype2-maintext2} can be rewritten in the form of \ref{eq:EP0-ent}
and vice versa (see \ref{prop:formequiv} in the appendix).

Now consider the set of states $\mathcal{D}_{P}$, defined as in \ref{eq:densBdef}
in terms of a set of projection operators $P$, as well as any state
$\rho\in\mathcal{D}_{P}$. {As mentioned below, common choices of $\mathcal{D}_{P}$ include the set of all states
(corresponding to $P=\{I\}$) and the set of states that are incoherent
relative to some basis (corresponding to $P=\{\vert i\rangle\langle i\vert\}_{i}$
for some basis $\{\vert i\rangle\}$).} We analyze the mismatch cost
of $\rho$, defined as the additional integrated EP incurred by $\rho$
relative to an optimal initial state within $\mathcal{D}_{P}$, $\varphi_{P}\in\mathop{\arg\min}_{\omega\in\mathcal{D}_{P}}\Sigma(\omega)$.
Our first result is that as long as $S(\rho\Vert\varphi_{P})<\infty$,
the mismatch cost is equal to the drop in relative entropy between
$\rho$ and $\varphi_{P}$ during the process,
\begin{equation}
\Sigma(\rho)-\Sigma(\varphi_{P})=-\Delta S({\rho}\Vert{\varphi_{P}}).\label{eq:equalityBasis}
\end{equation}
A sketch of the proof of this result is provided at the end of this
section, with details left for \ref{app:proofs-integrated}.

\ref{eq:equalityBasis} is a generalization of \ref{eq:classicalMismatch},
which holds for both finite and infinite dimensional systems, as well
as for optimizers $\varphi$ within arbitrary sets $\mathcal{D}_{P}$.
In the special case when $\mathcal{D}_{P}=\mathcal{D}$ (as induced
by $P=\{I\}$), \ref{eq:equalityBasis} expresses the ``global''
mismatch cost, the additional integrated EP incurred by the initial
state $\rho$ relative to a global optimizer $\varphi_{\mathcal{D}}\in\mathop{\arg\min}_{\omega\in\mathcal{D}}\Sigma(\omega)$.

\begin{figure}
\includegraphics[width=0.8\columnwidth]{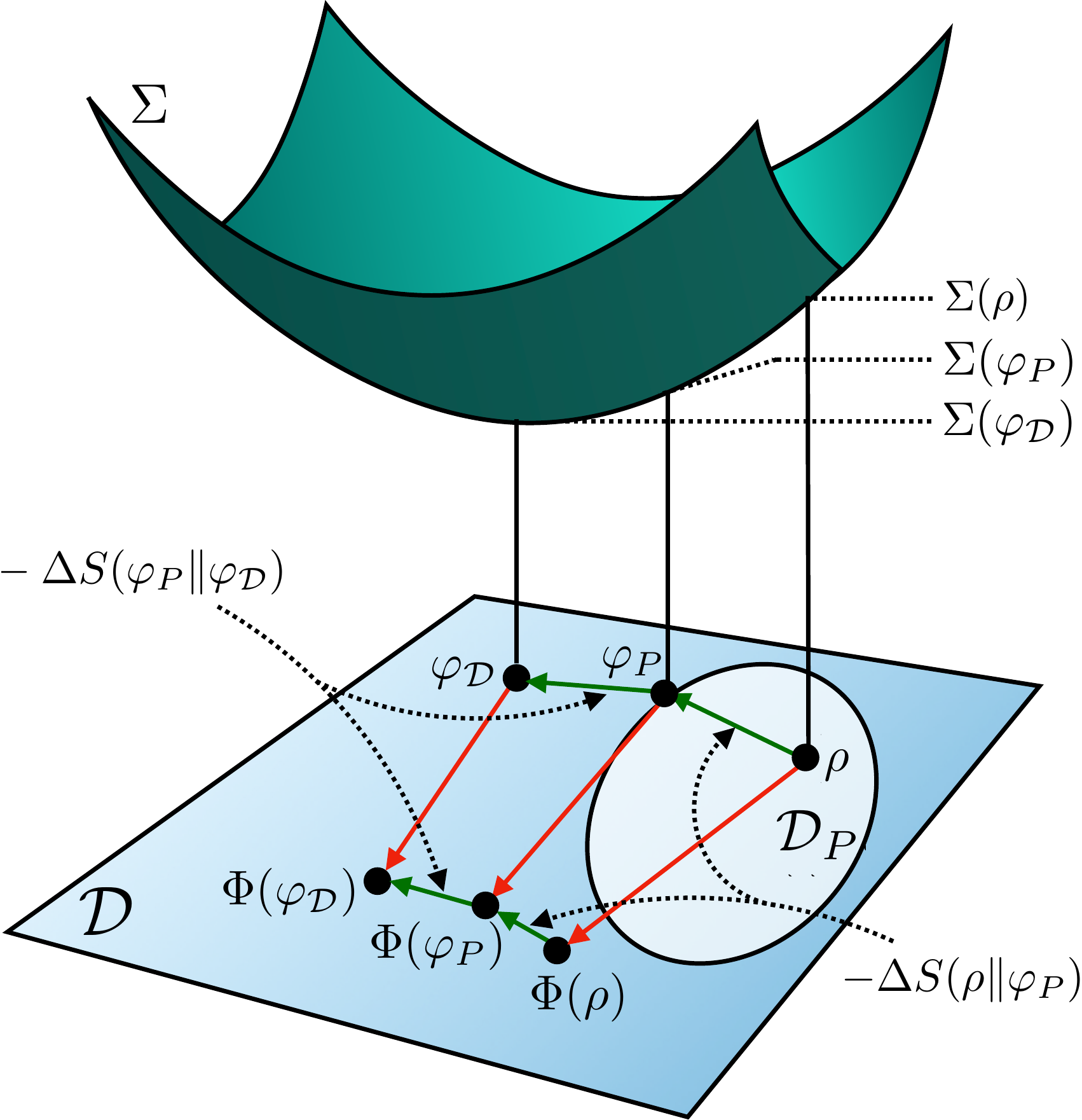}\caption{\label{fig:decomp}The mismatch cost of $\rho$ relative to the global
optimizer $\varphi_{\mathcal{D}}$ can be decomposed into a sum of
a \textquotedblleft classical\textquotedblright{} and \textquotedblleft quantum\textquotedblright{}
components, \ref{eq:mdecomp}. The classical component is given by
contraction of relative entropy between $\rho$ and $\varphi_{P}$,
the optimal state in the set of states diagonal in the same basis
as $\rho$ ($\mathcal{D}_{P}$, shown as a light oval). The quantum
component is given by the contraction of relative entropy between
$\varphi_{P}$ and $\varphi_{\mathcal{D}}$. (Compare to \ref{fig:1}.)}
\end{figure}

We can derive various useful decompositions of mismatch cost by applying
\ref{eq:equalityBasis} in an iterative manner. For example, consider
an orthonormal basis $\{\vert i\rangle\}_{i}$ that diagonalizes $\rho$.
Let $P=\{\vert i\rangle\langle i\vert\}_{i}$ so that $\mathcal{D}_{P}$
is the set of states diagonal in that basis, which in particular contains
$\rho$. Also let $\varphi_{P}\in\mathop{\arg\min}_{\omega\in\mathcal{D}_{P}}\Sigma(\omega)$
be an optimal initial state within $\mathcal{D}_{P}$, and let $\varphi_{\mathcal{D}}\in\mathop{\arg\min}_{\omega\in\mathcal{D}}\Sigma(\omega)$
be a global optimizer. In general, $\varphi_{\mathcal{D}}$ will not
be diagonal in the same basis as $\rho$, and so will not belong to
$\mathcal{D}_{P}$. We can then write
\[
\Sigma(\rho)-\Sigma(\varphi_{\mathcal{D}})=(\Sigma(\rho)-\Sigma(\varphi_{P}))+(\Sigma(\varphi_{P})-\Sigma(\varphi_{\mathcal{D}})),
\]
and --- assuming that $S(\rho\Vert\varphi_{P})$ and $S(\varphi_{P}\Vert\varphi_{\mathcal{D}})$
are finite --- apply \ref{eq:equalityBasis} to the two terms on the RHS. This
leads to the following decomposition of the global mismatch cost of
$\rho$ into two non-negative terms, which is visualized in \ref{fig:decomp}:
\begin{equation}
\Sigma(\rho)-\Sigma(\varphi_{\mathcal{D}})=-\Delta S({\rho}\Vert{\varphi_{P}})-\Delta S({\varphi_{P}}\Vert{\varphi_{\mathcal{D}}}).\label{eq:mdecomp}
\end{equation}
The first term, $-\Delta S({\rho}\Vert{\varphi_{P}})$, reflects the
mismatch cost between $\rho$ and $\varphi_{P}$. Since these two
states are diagonal in the same basis, it can be seen as the classical
contribution to mismatch cost. The second term, $-\Delta S({\varphi_{P}}\Vert{\varphi_{\mathcal{D}}})$,
is the purely quantum contribution to mismatch cost, which vanishes
when $\rho$ and $\varphi_{\mathcal{D}}$ can be diagonalized in the
same basis (since then $\Sigma(\varphi_{P})-\Sigma(\varphi_{\mathcal{D}})=0$).

Note that \ref{eq:mdecomp} is different from the decomposition of
mismatch cost into coherent and classical components previously derived
in \mycitep[Eq.~14]{riechers2020initial}. First, in our decomposition
both the classical and quantum are always non-negative (which is not
necessarily the case in \citep{riechers2020initial}). {Another
difference is that our decomposition does not include terms explicitly
related to the ``relative entropy of coherence'' \citep{baumgratz_quantifying_2014},
which appear in \mycitep[Eq.~14]{riechers2020initial} (as well as
in other classical-vs-quantum decompositions derived for EP in relaxation
processes \citep{santos_role_2019,francica_role_2019} and for quantum
work extraction \citep{francica2020quantum}).}

We now state our most generally applicable result for integrated EP
mismatch cost. Let $\mathcal{S}\subseteq\mathcal{D}$ be \emph{any}
convex subset of states, which may or may not have the form defined
in \ref{eq:densBdef}. Then, for any state $\rho\in\mathcal{S}$ and
a minimizer $\varphi_{\mathcal{S}}\in\mathop{\arg\min}_{\omega\in\mathcal{S}}\Sigma(\omega)$,
as long as $S(\rho\Vert\varphi_{\mathcal{S}})<\infty$,
\begin{equation}
\Sigma(\rho)-\Sigma(\varphi_{\mathcal{S}})\ge-\Delta S({\rho}\Vert{\varphi_{\mathcal{S}}}).\label{eq:ineqConvex}
\end{equation}
Equality holds if $(1-\lambda)\varphi_{\mathcal{S}}+\lambda\rho\in\mathcal{S}$
for some $\lambda<0$.

Since $\Sigma(\varphi_{\mathcal{S}})\ge0$ by the second law, \ref{eq:ineqConvex}
implies
\begin{equation}
\Sigma(\rho)\ge-\Delta S({\rho}\Vert{\varphi_{\mathcal{S}}}).\label{eq:genBound}
\end{equation}
The RHS of this bound is non-negative by the monotonicity of relative
entropy~\citep{muller2017monotonicity}. Thus, \ref{eq:genBound}
gives a tighter bound on EP than the second law, $\Sigma(\rho)\ge0$.
This tighter bound reflects the additional EP due to a suboptimal
choice of the initial state within any convex set of states $\mathcal{S}\ni\rho$.

We now briefly sketch the derivation of \ref{eq:equalityBasis,eq:ineqConvex},
leaving formal proofs for \ref{app:proofs-integrated}. A central
idea behind our derivations is that EP is a convex function whose ``amount
of convexity'' has a simple information-theoretic expression. Specifically,
using some simple algebra, it can be shown that for any convex mixture
$\varphi(\lambda)=(1-\lambda)\varphi+\lambda\rho$ of two states $\rho$
and $\varphi$,
\begin{multline}
(1-\lambda)\Sigma(\varphi)+\lambda\Sigma(\rho)-\Sigma(\varphi(\lambda))=\\
-\lambda\Delta S({\rho}\Vert{\varphi(\lambda)})-(1-\lambda)\Delta S({\varphi}\Vert{\varphi(\lambda)}),\label{eq:conv2-1}
\end{multline}
The quantity on the right hand side of \ref{eq:conv1} has been called
\emph{entropic disturbance} in quantum information theory \citep{shirokovLowerSemicontinuityEntropic2017a,buscemiUnifiedApproachInformationDisturbance2009,buscemiApproximateReversibilityContext2016}.
It is non-negative by monotonicity of relative entropy \citep{muller2017monotonicity},
which proves that $\Sigma$ is convex. Next, we consider the directional
derivatives of $\Sigma$ at $\varphi$ in the direction of $\rho$,
\[
{\textstyle {\textstyle \partial_{\lambda}^{+}}}\Sigma(\varphi(\lambda))\vert_{\lambda=0}=\lim_{{\lambda\to0^{+}}}\frac{\Sigma(\varphi(\lambda))-\Sigma(\varphi)}{\lambda}.
\]
In \ref{thm:genf} in the appendix, we rearrange \ref{eq:conv2-1}
and compute the appropriate limits to show that the directional derivative
can be evaluated as
\begin{equation}
{\textstyle {\textstyle \partial_{\lambda}^{+}}}\Sigma(\varphi(\lambda))\vert_{\lambda=0}=\Sigma(\rho)-\Sigma(\varphi)+\Delta S({\rho}\Vert{\varphi}).\label{eq:dd0}
\end{equation}

\ref{eq:ineqConvex} follows from \ref{eq:dd0} and the fact that
the directional derivative toward at the minimizer must be non-negative
(otherwise one could decrease the value of EP by moving slightly from
$\varphi$ to $\rho$, contradicting the fact that $\varphi$ is a
minimizer). To derive \ref{eq:equalityBasis}, suppose that $\varphi$
is a minimizer of EP within a set of states $\mathcal{D}_{P}$ defined
as in \ref{eq:densBdef}. If $\rho\ge\alpha\varphi$ for some $\alpha>0$,
then the directional derivative in \ref{eq:dd0} vanishes (since $\lambda=0$
is the minimizer of the function $\lambda\mapsto\Sigma(\varphi(\lambda))$
in the open set $(-\alpha,1)$), which in combination with \ref{eq:dd0}
implies \ref{eq:equalityBasis}. If $\rho\not\ge\alpha\varphi$ for
all $\alpha>0$, then \ref{eq:equalityBasis} can be derived by considering
a sequence of finite-rank projections of $\rho$ onto the top $n$
eigenvectors of $\varphi$, and then using continuity properties of
EP and relative entropy.

Note that our expression for mismatch cost, $-\Delta S({\rho}\Vert{\varphi})$,
depends both on the quantum channel $\Phi$ and the optimal state
$\varphi\in\mathop{\arg\min}_{\omega}\Sigma(\omega)$. The optimal
state $\varphi$ in turn depends on $\Phi$ and the entropy flow function
$Q$, which will encode various details of the physical process under
consideration (such as the precise trajectory of the driving Hamiltonians,
etc.). In general, the same channel $\Phi$ can be implemented with
different physical process, which will have different entropy flow
functions $Q$ and optimizers $\varphi$. For this reason, different
implementations of the same channel $\Phi$ can lead to different
values of mismatch cost for the same initial state $\rho$.

We also note that in order to evaluate some of our results numerically,
one must find an optimal state $\varphi\in\mathop{\arg\min}_{\omega}\Sigma(\omega)$.
In some special cases, $\varphi$ can be found in closed form. One
such case is considered below, in our analysis of protocols that obey
a symmetry group. Another example occurs when $\varphi\in\mathop{\arg\min}_{\omega\in\mathcal{D}_{P}}$
is a minimizer within some set of states $\mathcal{D}_{P}$
and $\Phi$ is input-independent (there is some $\rho'$ such that
$\Phi(\rho)=\rho'$ for all $\rho$). Then, writing the entropy flow
term in trace form as $Q(\rho)=\mathrm{tr}\{\rho A\}$, it is straightforward
to show that the minimizer must have the following form \footnote{This follows by writing \unexpanded{$\Sigma(\rho)=S(\Phi(\rho))-S(\rho)+\mathrm{tr}\{\Pi A\Pi \rho\}$}\unexpanded{$=S(\rho\Vert\varphi)+\text{const}$},
where \unexpanded{$\varphi$} is defined as in \myfootref{eq:optsol}.}:
\begin{equation}
\varphi=e^{-\sum_{\Pi\in P}\Pi A\Pi}/\mathrm{tr}\{e^{-\sum_{\Pi\in P}\Pi A\Pi}\}.\label{eq:optsol}
\end{equation}
More generally, $\varphi$ can be found using numerical techniques.
Because $\Sigma$ is a convex function, this optimization can be performed
efficiently (some appropriate algorithms are discussed in \citep{ramakrishnan2019non}).

\subsection{Support conditions}

\label{subsec:Support-conditions-integrated}

Our result for mismatch cost, \ref{eq:equalityBasis}, only apply
when $S(\rho\Vert\varphi_{P})<\infty$, for which it is necessary
that
\begin{align}
\mathrm{supp}\,\rho\subseteq\mathrm{supp}\,\varphi_{P}.\label{eq:suppcond-1}
\end{align}
(In finite dimensions, \ref{eq:suppcond-1} is both necessary and
sufficient for $S(\rho\Vert\varphi_{P})<\infty$; in infinite dimensions,
it is necessary but not sufficient). Here, we show that \ref{eq:suppcond-1}
is satisfied in many cases of interest.

To begin, we consider some set of states $\mathcal{D}_{P}$, while
making the weak assumption that the physical process is such that
$\Sigma(\rho)$ is finite for all pure states in $\mathcal{D}_{P}$.
Then, \ref{thm:connectedmeansfullsupport} in the appendix shows that
the support of the optimizer $\varphi_{P}\in\mathop{\arg\min}_{\omega\in\mathcal{D}_{P}}\Sigma(\rho)$
and its orthogonal complement must be non-interacting subspaces under
the action of $\Phi$,
\begin{equation}
\Phi(\varphi_{P})\perp\Phi(\omega)\qquad\forall\omega\in\mathcal{D}_{P}:\omega\perp\varphi.\label{eq:connres-2}
\end{equation}
Now, suppose that $\Phi$ is ``irreducible'' (over $P$) in the sense that pairs
of states which jointly span $\mathcal{H}_{P}$ always incur some
overlap,
\begin{equation}
\Phi(\omega)\not\perp\Phi(\omega')\quad\forall\omega,\omega'\in\mathcal{D}_{P}:\mathrm{supp}\,(\omega+\omega')=\mathcal{H}_{P},\label{eq:fsp0-1}
\end{equation}
where $\mathcal{H}_{P}$ is defined as in \ref{eq:hilbB}. Then, it
must be that $\mathrm{supp}\,\varphi_{P}=\mathcal{H}_{P}$, since
otherwise there would be some state $\omega\in\mathcal{D}_{P}$ that
leads to a contradiction between \ref{eq:connres-2,eq:fsp0-1}.

To summarize, our results show that if $\Phi$ is irreducible in sense
of \ref{eq:fsp0-1}, then the support condition in \ref{eq:suppcond-1}
must hold. Note that \ref{eq:fsp0-1} is satisfied when the support
of all output states is equal,
\begin{equation}
\mathrm{supp}\,\Phi(\rho)=\mathrm{supp}\,\Phi(\omega)\qquad\forall\rho,\omega,\label{eq:m234}
\end{equation}
such as the common situation when $\Phi(\rho)>0$ for all $\rho$.

Conversely, if $\Phi$ is \emph{not} irreducible in the sense of \ref{eq:fsp0-1},
then one can decompose the $\mathcal{H}_{P}$ into a set of orthogonal
subspaces $\mathcal{H}_{1},\mathcal{H}_{2},\dots$ such that \ref{eq:fsp0-1}
holds in each subspace \footnote{In general, this decomposition will not be unique: imagine
the trivial case where, in \ref{eq:EP0-ent}, \unexpanded{$\Phi=\mathrm{Id}$}
and \unexpanded{$Q(\rho)=0$}; then, \unexpanded{$\Sigma(\rho)=0$}
for all \unexpanded{$\rho$}, and any complete basis \unexpanded{$\{\vert i\rangle\}$}
can be used to define a basin decomposition.}. Such orthogonal subspaces have been previously called ``basins''
in the quantum context \citep{riechers2020initial} and ``islands''
in the classical context \citep{wolpert2020thermodynamic}. Using
the arguments above, it can be shown that the optimal state within
each basin $\mathcal{H}_{i}$ will have support equal to $\mathcal{H}_{i}$;
from \ref{eq:connres-2}, it also follows that optimal states within
different basins will not interact under the action of $\Phi$. This
resolves a conjecture in \citep{riechers2020initial} and justifies
the decomposition of $\Sigma$ developed in that paper into a sum
of mismatch costs incurred within each basin, plus an ``inter-basin
coherence'' term (for details, see Appendix~E in \citep{riechers2020initial}).

\subsection{Example}

\label{subsec:Example-symm}

To illustrate our results with a concrete example, we analyze the
EP incurred by a process that obeys a symmetry group. (For related
analyses for classical systems see \citep{kolchinsky2020entropy},
and for quantum systems see \citep{janzing_quantum_2006,marvianHowQuantifyCoherence2016,vaccaro_tradeoff_2008}).

To begin, consider a physical process whose dynamics $\Phi$ commute
with some unitary $U$,
\begin{equation}
\Phi(U\rho U^{\dagger})=U\Phi(\rho)U^{\dagger}\qquad\forall\rho,\label{eq:covar0}
\end{equation}
implying that the dynamics are ``covariant'' under $U$ \citep{holevoNoteCovariantDynamical1993}
. Furthermore, suppose that the entropy flow function $Q$ associated
with the process is invariant under the action of the same unitary,
\begin{equation}
Q(\rho)=Q(U\rho U^{\dagger})\qquad\forall\rho.\label{eq:var0}
\end{equation}
\ref{eq:covar0} says that in terms of dynamics, it does not matter
when one first applies $U$ to the initial $\rho$ and then evolves
the system under $\Phi$, or first evolves the system under $\Phi$
and then applies the unitary $U$. \ref{eq:var0} says that in terms
of thermodynamics, the entropy flow doesn't change when one transforms
$\rho$ by $U$.

For simplicity, we will first assume that $U$ is some involution
($UU=I$). For concreteness, one can imagine that $U$ involves flipping
the state of a qubit in a quantum circuit, which does not interact
with the other qubits nor change state during the operation of the
circuit (it can be verified that \ref{eq:covar0} and \ref{eq:var0}
will hold under these assumptions).

\begin{figure}
\includegraphics[width=0.8\columnwidth]{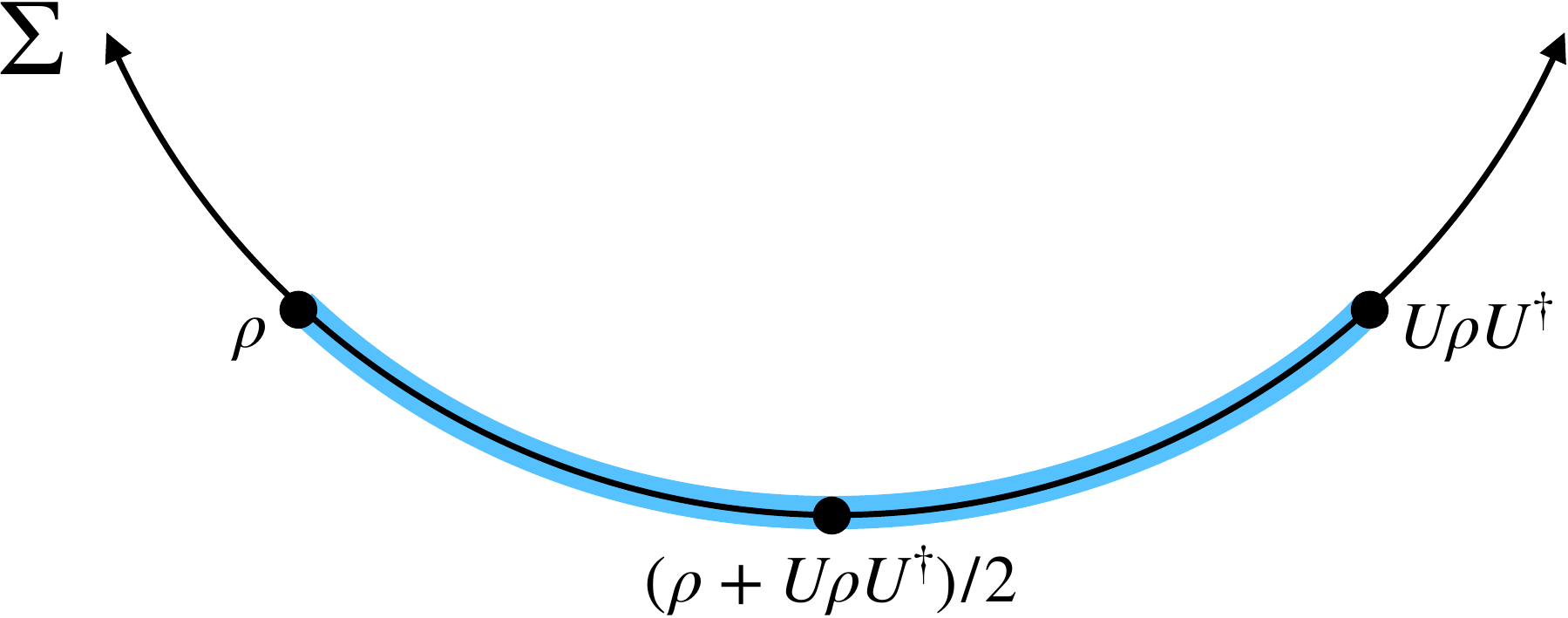}\caption{\label{fig:symm}As an example, we consider a physical process in
which the EP is invariant under some unitary involution, $\Sigma(\rho)=\Sigma(U\rho U^{\dagger})$
and $UU=I$. For any $\rho$, the uniform mixture $(\rho+U\rho U^{\dagger})/2$
achieves minimum EP within the set of convex combinations of $\rho$
and $U\rho U^{\dagger}$, $\mathcal{S}=\{\lambda\rho+(1-\lambda)U\rho U^{\dagger}:\lambda\in[0,1]\}$
(thick blue line). This leads to the lower bound on EP incurred by
state $\rho$, \ref{eq:g1}.}
\end{figure}

Plugging \ref{eq:covar0} and \ref{eq:var0} into \ref{eq:EP0-ent},
and using the fact that von Neumann entropy is invariant under unitary
transformations, we see that the EP incurred by the process is invariant
under $U$:
\begin{equation}
\Sigma(\rho)=\Sigma(U\rho U^{\dagger})\qquad\forall\rho.\label{eq:var0-1}
\end{equation}
We can now use the results derived above to bound the EP incurred
by any initial state $\rho$. To guide intuition, in \ref{fig:symm}
we plot the EP incurred by states in the set $\mathcal{S}$ consisting
of convex combinations of $\rho$ and $U\rho U^{\dagger}$. Observe
that for any such convex combination $\omega=\lambda\rho+(1-\lambda)U\rho U^{\dagger}\in\mathcal{S}$,
\begin{align}
\Sigma(\omega) & =(\Sigma(\omega)+\Sigma(U\omega U^{\dagger}))/2\nonumber \\
& \ge\Sigma((\omega+U\omega U^{\dagger})/2)\nonumber \\
& =\Sigma((\rho+U\rho U^{\dagger})/2),\label{eq:m7}
\end{align}
where we first used \ref{eq:var0-1}, then the convexity of $\Sigma$,
and finally that $(\omega+U\omega U^{\dagger})/2=(\rho+U\rho U^{\dagger})/2$
(which follows from some simple algebra and the fact that $U$ is involution). \ref{eq:m7}
implies that minimizer of EP in $\mathcal{S}$ is $(\rho+U\rho U^{\dagger})/2$.
Next, for convenience, define the linear operator $\Psi(\rho)=(\rho+U\rho U^{\dagger})/2$.
\ref{eq:genBound} then gives the following EP bound:
\begin{align}
\Sigma(\rho) & \ge S(\rho\Vert\Psi(\rho))-S(\Phi(\rho)\Vert\Phi(\Psi(\rho)))\nonumber \\
& =S(\rho\Vert\Psi(\rho))-S(\Phi(\rho)\Vert\Psi(\Phi(\rho))),\label{eq:g1}
\end{align}
where in the second line we used that $\Phi$ and $\Psi$ commute
(due to linearity of $\Phi$ and \ref{eq:covar0}).

It is straightforward to generalize this result from simple involutions
to more general symmetry groups. Let $G$ be a finite group that acts
on $\mathcal{H}$ via a set of unitaries $\{U_{g}:g\in G\}$ (the
involution example above corresponds to the $S_{2}$ group which acts
on $\mathcal{H}$ via $\{I,U\}$). Suppose that \ref{eq:covar0} and
\ref{eq:var0} (and hence \ref{eq:var0-1}) hold for each $U_{g}$
individually. Using \ref{eq:genBound} and a similar derivation as
above, one can show that \ref{eq:g1} still holds, as long as the
operator $\Psi$ is defined as a uniform average over all elements
of the group, $\Psi(\rho):=\frac{1}{|G|}\sum U_{g}\rho U_{g}^{\dagger}$.

In the quantum information literature, the linear operator $\Psi$
is called a ``twirling'' operator \citep{vaccaro_tradeoff_2008}.
Moreover, the quantity $S(\rho\Vert\Psi(\rho))$ in \ref{eq:g1} is
known as \emph{relative entropy of asymmetry}, and it measures the
amount of asymmetry in state $\rho$ relative to the group $G$ \citep{marvianHowQuantifyCoherence2016,vaccaro_tradeoff_2008}.
Thus, \ref{eq:g1} shows that for any process that is invariant under
the action of a symmetry group, in the sense that \ref{eq:covar0} and \ref{eq:var0}
are obeyed, the EP involved in transforming $\rho\to\Phi(\rho)$
is lower bounded by the decrease of asymmetry during that transformation.
Said somewhat differently, any process that obeys a symmetry group
must dissipate asymmetry as EP.

\section{Mismatch Cost for Fluctuating EP}

\label{sec:Fluctuating-EP}

In our second set of results, we analyze EP and mismatch cost at the
level of individual stochastic realizations of the physical process.
To begin, we briefly review the definitions of fluctuating EP as used
in quantum stochastic thermodynamics.

Consider a system that evolves according to the channel $\Phi$ from
some initial mixed state $\rho=\sum_{i}p_{i}\vert i\rangle\langle i\vert$
to some final mixed state $\Phi(\rho)=\sum_{\phi}p'_{\phi}\vert\phi\rangle\langle\phi\vert$.
Suppose that this stochastic process is carried out multiple times,
resulting in a set of randomly sampled realizations. Each realization
can be characterized by the associated initial pure state $\vert i\rangle\langle i\vert$,
the final pure state $\vert\phi\rangle\langle\phi\vert$, and the
associated entropy flow $q\in\mathbb{R}$ (i.e., the increase of the
thermodynamic entropy of the reservoirs that occurs during that realization).
The fluctuating EP of realization $(i\!\shortrightarrow\!\phi,q)$
is then given by~\citep{esposito2006fluctuation,esposito2009nonequilibrium,campisi2011colloquium}
\begin{equation}
\sigma_{\rho}(i\!\shortrightarrow\!\phi,q):=(-\ln p'_{\phi}+\ln p_{i})+q,\label{eq:fluctEP}
\end{equation}
while the probability of realization $(i\!\shortrightarrow\!\phi,q)$
is given by
\begin{align}
p_{\rho}(i,\phi,q) & =p_{\rho}(i,\phi)p(q\vert i,\phi)\label{eq:m2}\\
& =p_{i}T_{\Phi}(\phi\vert i)p(q\vert i,\phi)\\
& =p_{i}\mathrm{tr}\{\Phi(\vert i\rangle\langle i\vert)\vert\phi\rangle\langle\phi\vert\}p(q\vert i,\phi).\label{eq:m3}
\end{align}
In \ref{eq:m3}, $p_{i}$ is the probability of initial pure state
$\vert i\rangle\langle i\vert$, $p(q\vert i,\phi)$ is the conditional
probability of entropy flow $q$ given the transition $i\!\shortrightarrow\!\phi$,
and
\begin{equation}
T_{\Phi}(\phi\vert i)=\mathrm{tr}\{\Phi(\vert i\rangle\langle i\vert)\vert\phi\rangle\langle\phi\vert\}\label{eq:jk2}
\end{equation}
is the conditional probability of the final pure state $\vert\phi\rangle\langle\phi\vert$
given the initial pure state $\vert i\rangle\langle i\vert$ under
$\Phi$.

In quantum stochastic thermodynamics, the terms $q$ and $p(q\vert i,\phi)$
have been defined and operationalized in various ways, including via
two-point projective measurements~\citep{esposito2006fluctuation,manzanoNonequilibriumPotentialFluctuation2015,manzanoQuantumFluctuationTheorems2018},
weak measurements~\citep{allahverdyanNonequilibriumQuantumFluctuations2014a},
POVMs~\citep{kwonFluctuationTheoremsQuantum2019}, and dynamic Bayesian
networks~\citep{micadeiQuantumFluctuationTheorems2020}. In all cases,
however, these terms are chosen so that two conditions are satisfied:
(1) fluctuating EP agrees with integrated EP in expectation,
\begin{equation}
\langle\sigma_{\rho}\rangle_{p_{\rho}}=\Sigma(\rho),\label{eq:m1}
\end{equation}
where $\langle\cdot\rangle_{p_{\rho}}$ indicates expectation under
$p_{\rho}(i,\phi,q)$, and (2) fluctuating EP obeys an integral
fluctuation theorem (IFT),
\begin{equation}
\langle e^{-\sigma_{\rho}}\rangle_{p_{\rho}}=\gamma,\label{eq:IFT1}
\end{equation}
where $\gamma$ is either equal to 1 or (more generally) some number between 0 and 1 that quantifies the ``absolute irreversibility'' of the process \citep{funo2015quantum}. Importantly,
our results below do not depend on the particular definition
of $q$ and $p(q\vert i,\phi)$, only on the fact that fluctuating
EP can be written in the general form of \ref{eq:fluctEP}.

Below, we define fluctuating mismatch cost as the trajectory-level
version of the mismatch cost $\Sigma(\rho)-\Sigma(\varphi)$, where
$\varphi$ is an optimal initial (mixed) state that minimizes EP.
Before proceeding, consider some convex set of states $\mathcal{S}\subseteq\mathcal{D}$.
Let $\varphi\in\mathop{\arg\min}_{\omega\in\mathcal{S}}\Sigma(\omega)$
indicate an optimizer in $\mathcal{S}$and let $\rho\in\mathcal{S}$
indicate some state in $\mathcal{S}$ such that $S(\rho\Vert\varphi)<\infty$.
We will assume that
\begin{equation}
\Sigma(\rho)-\Sigma(\varphi)=-\Delta S({\rho}\Vert{\varphi}).\label{eq:z4}
\end{equation}
By \ref{eq:equalityBasis}, \ref{eq:z4} is satisfied whenever $\mathcal{S}=\mathcal{D}_{P}$;
more generally, it is satisfied if the equality form of \ref{eq:ineqConvex}
holds.

Below we consider two cases differently: (1) the simpler ``commuting''
case, where the initial state $\rho$ commutes with $\varphi$ and
the final state $\Phi(\rho)$ commutes with $\Phi(\varphi)$ (note
that this special case includes all classical processes; see \ref{app:classical}
for details); (2) the more complicated ``non-commuting'' case, where
$\rho$ does not commute with $\varphi$ and/or $\Phi(\rho)$ does
not commute with $\Phi(\varphi)$.

\subsection{Commuting case}

\label{subsec:Mismatch-cost-forcomm}

\begin{figure}
\includegraphics[width=1\columnwidth]{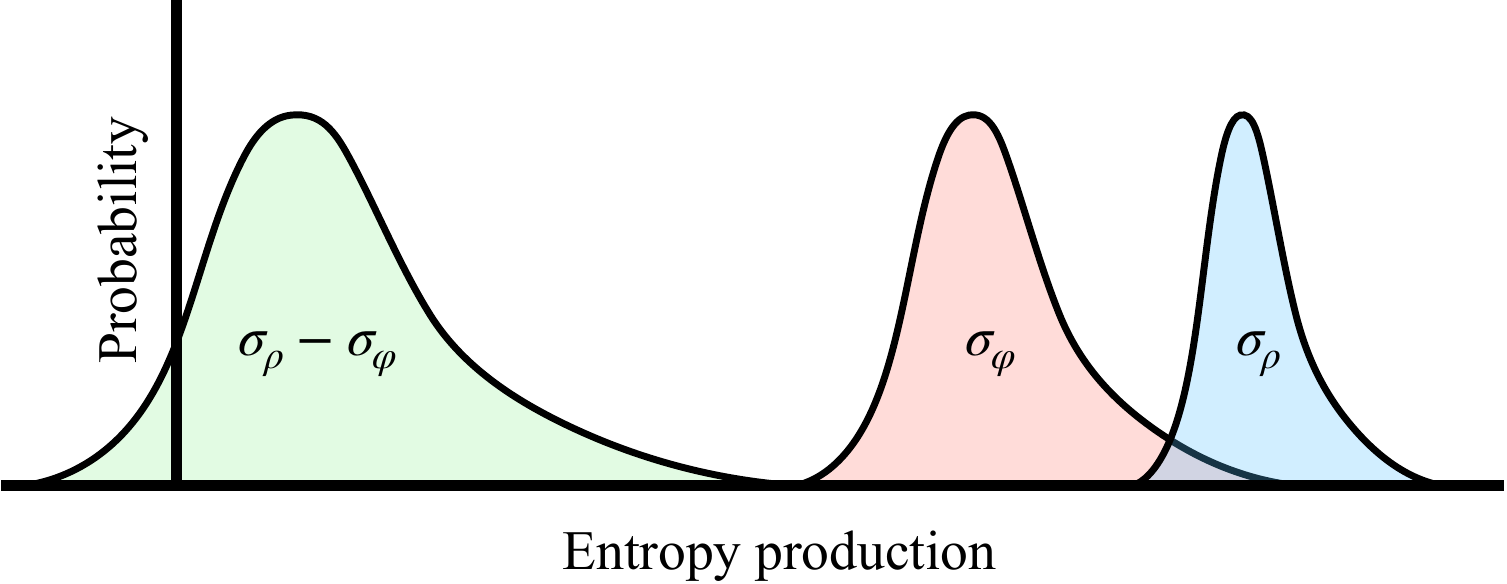}\caption{\label{fig:ft}Red and blue curves show the probability distribution
of $\sigma_{\rho}$ and $\sigma_{\varphi}$, the fluctuating EP incurred
by stochastic realizations sampled from some initial state $\rho$
and the optimal initial state $\varphi$ (which minimizes integrated
EP). Each of these fluctuating EP terms individually obeys an integral
fluctuating theorem (IFT), \ref{eq:IFT1}. We show that difference
of these fluctuating EP terms, $\sigma_{\rho}-\sigma_{\varphi}$,
is the fluctuating expression of mismatch cost, and that it also obeys
an IFT, \ref{eq:IFT}.}
\end{figure}

We first assume that the initial states $\rho$ and $\varphi$ commute,
as do the final states $\Phi(\rho)$ and $\Phi(\varphi)$. This means
that $\varphi$ can be diagonalized in the same basis as $\rho$,
$\varphi=\sum_{i}r_{i}\vert i\rangle\langle i\vert$, and $\Phi(\varphi)$
can be diagonalized in the same basis as $\Phi(\rho)$, $\Phi(\varphi)=\sum_{\phi}r'_{\phi}\vert\phi\rangle\langle\phi\vert$.

We then define the \emph{fluctuating mismatch cos}t of a given realization
$(i\!\shortrightarrow\!\phi,q)$ as the difference between $\sigma_{\rho}(i\!\shortrightarrow\!\phi,q)$,
the fluctuating EP of the actual realization, and $\sigma_{\varphi}(i\!\shortrightarrow\!\phi,q)$,
the fluctuating EP assigned to the same realization $(i\!\shortrightarrow\!\phi,q)$
if the physical process were started from the initial mixed state
$\varphi$:
\begin{align}
& \sigma_{\rho}(i\!\shortrightarrow\!\phi,q)-\sigma_{\varphi}(i\!\shortrightarrow\!\phi,q)\label{eq:mDef0}\\
& \qquad=(-\ln p'_{\phi}+\ln p_{i})-(-\ln r'_{\phi}+\ln r_{i}).\label{eq:mm0}
\end{align}
(Note that this is different from $\sigma_{\rho}(i\!\shortrightarrow\!\phi,q)-\Sigma(\varphi)$,
the additional fluctuating EP incurred by realization $(i\!\shortrightarrow\!\phi,q)$
under the initial state $\rho$, additional to the \emph{expected
EP} achieved by the optimal initial state $\varphi$.)

We now derive our main results for fluctuating mismatch cost, which
are also illustrated in \ref{fig:ft} (see \ref{app:fluctuating}
for all derivations). First, a simple calculation shows that \ref{eq:mDef0}
is a proper definition of fluctuating mismatch cost, in that its expectation
under $p_{\rho}(i,\phi,q)$ is equal to the mismatch cost for integrated
EP,
\begin{align}
\langle\sigma_{\rho}-\sigma_{\varphi}\rangle_{p_{\rho}} & =-\Delta S({\rho}\Vert{\varphi})=\Sigma(\rho)-\Sigma(\varphi).\label{eq:fluct0}
\end{align}
Second, the fluctuating mismatch cost obeys an IFT,
\begin{equation}
\langle e^{-(\sigma_{\rho}-\sigma_{\varphi})}\rangle_{p_{\rho}}=\gamma\in(0,1],\label{eq:IFT}
\end{equation}
where $\gamma$ is a ``correction factor'' that accounts for the
fact that some initial pure states are never seen when sampling from
$\rho$. Formally, this correction factor is defined as
\[
\gamma=\mathrm{tr}\{\Pi^{\rho}(\mathcal{R}_{\Phi}^{\varphi}(\Phi(\rho)))\},
\]
where $\mathcal{R}_{\Phi}^{\varphi}$ is the recovery map from \ref{eq:recov}
and $\Pi^{\rho}$ is the projection onto the support of $\rho$. This
correction factor achieves its maximum value of 1 when the $\rho$
has the same support as $\varphi$, and is closely related to the
notion of ``absolute irreversibility'' studied by Funo et al. \citep{funo2015quantum}.

Note that mismatch cost for integrated EP is always non-negative,
$\Sigma(\rho)-\Sigma(\varphi)\ge0$, since $\varphi$ is a minimizer
of EP. On the other hand, applying Jensen's inequality to the IFT
in \ref{eq:IFT} gives the lower bound $\Sigma(\rho)-\Sigma(\varphi)\ge-\ln\gamma$,
which is stronger than the first one whenever $\gamma<1$. Furthermore,
using standard techniques in stochastic thermodynamics (see \ref{app:fluctuating}),
the IFT in \ref{eq:IFT} implies that negative values of fluctuating
mismatch cost are exponentially unlikely,
\begin{equation}
\mathrm{Pr}\big[(\sigma_{\rho}-\sigma_{\varphi})\le-\xi\big]\le\gamma e^{-\xi}.\label{eq:fluctBound}
\end{equation}

In stochastic thermodynamics, the fluctuating EP of a trajectory typically
reflects how much the trajectory's probability violates time-reversal
symmetry between the process under consideration and a special ``time-reversed''
version of the process \citep{seifert2012stochastic,campisi2011colloquium}.
In contrast, our derivations do not explicitly involve any time-reversed
process. However, it is possible to interpret fluctuating mismatch
cost as \textit{implicitly} referencing the violation of time-reversal
symmetry. Let $\mathcal{R}_{\Phi}^{\varphi}$ indicate the Petz recovery
map, where the optimal initial state $\varphi$ is chosen as the reference
state, and let $T_{\mathcal{R}_{\Phi}^{\varphi}}(i\vert\phi)$ indicate
the corresponding conditional probability, defined as in \ref{eq:jk2}
but for the channel $\mathcal{R}_{\Phi}^{\varphi}$ rather than $\Phi$.
Then, as we show in \ref{app:fluctuating}, the fluctuating mismatch
cost in \ref{eq:mm0} can be written as
\begin{equation}
\sigma_{\rho}(i\!\shortrightarrow\!\phi,q)-\sigma_{\varphi}(i\!\shortrightarrow\!\phi,q)=\ln\frac{p_{i}T_{\Phi}(\phi\vert i)}{p'_{\phi}T_{\mathcal{R}_{\Phi}^{\varphi}}(i\vert\phi)}.\label{eq:ldb1}
\end{equation}
Thus, fluctuating mismatch cost reflects the breaking of time-reversal
symmetry, as quantified by the difference between the joint probability
of starting on pure state $\vert i\rangle\langle i\vert$ and ending
on pure state $\vert\phi\rangle\langle\phi\vert$ under the regular
process, versus the joint probability of starting on pure state $\vert\phi\rangle\langle\phi\vert$
and ending on pure state $\vert i\rangle\langle i\vert$ under the
time-reversed process specified by the Petz recovery map. (See also
Ref.~\citep{buscemi2020fluctuation} for a related fluctuation theorem
that also makes use of the Petz recovery map.)

\subsection{Non-commuting case}

We now consider the more general case when the pair of initial states
$\rho,\varphi$ and/or the pair of final states $\Phi(\rho),\Phi(\varphi)$
do not commute. In this case, the pair of initial states $\rho,\varphi$
and/or final states $\Phi(\rho),\Phi(\varphi)$ cannot be simultaneously
diagonalized, so one cannot define fluctuating mismatch cost as in
\ref{eq:mDef0}. Nonetheless, we show that it is still possible to
define a non-commuting version of \ref{eq:mm0}, which is a proper
trajectory-level measure of mismatch cost, obeys an IFT, and reflects the
breaking of time-reversal symmetry in a way analogous to \ref{eq:ldb1}.

To derive our results, we employ a framework recently developed by
Kwon and Kim \citep{kwonFluctuationTheoremsQuantum2019}, which provides
a fluctuation theorem for quantum processes which is stated in terms
of a quantum channel $\Phi$, an initial state $\rho$, and some arbitrary
``reference state'' $\varphi$. Write the spectral resolutions of
the initial mixed states as $\rho=\sum_{i}p_{i}\vert i\rangle\langle i\vert$
and $\varphi=\sum_{a}r_{a}\vert a\rangle\langle a\vert$, and write
the spectral resolutions of the final mixed states as $\Phi(\rho)=\sum_{\phi}p'_{\phi}\vert\phi\rangle\langle\phi\vert$
and $\Phi(\varphi)=\sum_{\alpha}r'_{\alpha}\vert\alpha\rangle\langle\alpha\vert$.
Then, in the framework of \citep{kwonFluctuationTheoremsQuantum2019},
each stochastic realization of a process that carries out $\Phi$
on initial state $\rho$ is characterized by four factors: (1) an
initial pure state $\vert i\rangle\langle i\vert$ in the basis of
$\rho$, (2) a final pure state $\vert\phi\rangle\langle\phi\vert$
in the basis of $\Phi(\rho)$, (3) an initial (generally off-diagonal)
term $\vert a\rangle\langle b\vert$ in the basis of the reference
state $\varphi$, and (4) a final (generally off-diagonal) term $\vert\alpha\rangle\langle\beta\vert$
in the basis of the reference state $\Phi(\varphi)$.

Given these four factors, each realization can be assigned the following
fluctuating quantity \mycitep[Eq.~11]{kwonFluctuationTheoremsQuantum2019},
\begin{align}
& m_{\rho,\varphi}(i,a,b\!\shortrightarrow\!\phi,\alpha,\beta)\label{eq:kwonFluctEP0}\\
& \quad=-\ln p'_{\phi}+\ln p_{i}+\frac{1}{2}\ln r'_{\alpha}r'_{\beta}-\frac{1}{2}\ln r_{a}r_{b}\label{eq:kwonFLUCTEP}\\
& \quad=-\ln p'_{\phi}+\ln p_{i}+\ln\frac{T_{\Phi}(\alpha,\beta\vert a,b)}{T_{\mathcal{R}_{\Phi}^{\varphi}}(a,b\vert\alpha,\beta)},\label{eq:ldb2}
\end{align}
where $T_{\Phi}$ and $T_{\mathcal{R}_{\Phi}^{\varphi}}$ encode the
forward and backward conditional \emph{quasiprobability} distributions,
\begin{align*}
T_{\Phi}(\alpha,\beta\vert a,b) & =\langle\alpha\vert\Phi(\vert a\rangle\langle b\vert)\vert\beta\rangle,\\
T_{\mathcal{R}_{\Phi}^{\varphi}}(a,b\vert \alpha,\beta) & =\langle a \vert\mathcal{R}_{\Phi}^{\varphi}(\vert \alpha \rangle\langle \beta \vert)\vert b \rangle.
\end{align*}
Note that the backward conditional quasiprobability distribution is
defined in terms of the Petz recovery map, \ref{eq:recov}. In Ref.~\citep{kwonFluctuationTheoremsQuantum2019},
the quantity $m_{\rho,\varphi}$ is interpreted as a kind of ``fluctuating
EP'' defined relative to an arbitrary reference state $\varphi$,
which is purely information-theoretic in nature (i.e., this fluctuating
EP does not \emph{a priori} have anything to do with thermodynamic
entropy production). As we discuss below, our interpretation of $m_{\rho,\varphi}$
will be somewhat different.

Before proceeding, we discuss how one might compute the expectation
of $m_{\rho,\varphi}$ under a joint probability distribution over
realizations $i,a,b\!\shortrightarrow\!\phi,\alpha,\beta$. In fact,
no such joint probability distribution can exist, because in general
it is impossible to assign valid joint probability to the outcomes
of non-commuting observables \citep{allahverdyanExcludingJointProbabilities2018}.
However, one can assign each realization $i,a,b\!\shortrightarrow\!\phi,\alpha,\beta$
the following \emph{quasiprobability} \mycitep[Eq.~13]{kwonFluctuationTheoremsQuantum2019},
\begin{multline}
\tilde{p}_{\rho}(i,a,b,\phi,\alpha,\beta):=\\
p_{i}\,\langle\phi\vert\alpha\rangle\langle\alpha\vert\Phi\big(\vert a\rangle\langle a\vert i\rangle\langle i\vert b\rangle\langle b\vert\big)\vert\beta\rangle\langle\beta\vert\phi\rangle.\label{eq:mm3}
\end{multline}
(See Appendix D in \citep{kwonFluctuationTheoremsQuantum2019} for
details of how the quasiprobability distribution in \ref{eq:mm3}
can be operationally measured.) Although the quasiprobability distribution
$\tilde{p}_{\rho}$ can take negative values for certain outcomes,
it nonetheless has positive and correct marginal distributions over
the outcomes of the individual observables. Using this, the expectation
of $m_{\rho,\varphi}$ (as defined in \ref{eq:kwonFluctEP0}) under
$\tilde{p}_{\rho}$ can be shown to be equal to the contraction of
relative entropy between $\rho$ and $\varphi$ \citep[Eq.~25]{kwonFluctuationTheoremsQuantum2019},
\begin{equation}
\langle m_{\rho,\varphi}(i,a,b\!\shortrightarrow\!\phi,\alpha,\beta)\rangle_{\tilde{p}_{\rho}}=-\Delta S({\rho}\Vert{\varphi}).\label{eq:fluctnn}
\end{equation}
Moreover, this quantity also satisfies an IFT (Appendix G in \citep{kwonFluctuationTheoremsQuantum2019}),
\begin{equation}
\langle e^{m_{\rho,\varphi}(i,a,b\!\shortrightarrow\!\phi,\alpha,\beta)}\rangle_{\tilde{p}_{\rho}}=\gamma,\label{eq:ift2}
\end{equation}
where $\gamma=\mathrm{tr}\{\Pi^{\rho}(\mathcal{R}_{\Phi}^{\varphi}(\Phi(\rho)))\}\in(0,1]$.

Our interpretation of the quantity $m_{\rho,\varphi}$ is somewhat
different from the one discussed in \citep{kwonFluctuationTheoremsQuantum2019}.
As mentioned, we choose the reference state $\varphi$ to be a minimizer of EP,
and assume that it satisfies the relation $\Sigma(\rho)-\Sigma(\varphi)=-\Delta S({\rho}\Vert{\varphi})$,
\ref{eq:z4}. Then, $m_{\rho,\varphi}$ acquires a concrete thermodynamic
meaning: given \ref{eq:fluctnn}, it is the expression of fluctuating mismatch cost (i.e., difference of thermodynamic entropy production terms), which applies even when states $\rho$ and
$\varphi$ do not commute. This holds because \ref{eq:fluctnn} and
\ref{eq:z4} together give the non-commuting analogue of \ref{eq:fluct0}:
\begin{equation}
\langle m_{\rho,\varphi}(i,a,b\!\shortrightarrow\!\phi,\alpha,\beta)\rangle_{\tilde{p}_{\rho}}=\Sigma(\rho)-\Sigma(\varphi).\label{eq:m32}
\end{equation}
Similarly, the expression of the breaking of time-reversal symmetry
in \ref{eq:ldb2} is the non-commuting analogue of \ref{eq:ldb1},
while the IFT in \ref{eq:ift2} is the non-commuting analogue of \ref{eq:IFT}.

As mentioned, the quasiprobability distribution $\tilde{p}_{\rho}$
can assign negative values to some joint outcomes. For this reason,
one cannot generally derive an exponential bound on the probability
of negative mismatch cost as in \ref{eq:fluctBound}. Nonetheless,
via the series expansion of the exponential function, the IFT in \ref{eq:ift2}
can still be shown to constrain all moments of fluctuating mismatch
cost \citep[p.~13]{kwonFluctuationTheoremsQuantum2019}.

Finally, in the case that the pair of initial states $\rho$ and
$\varphi$ as well as the pair of final states $\Phi(\rho)$ and
$\Phi(\varphi)$ commute --- and therefore can be diagonalized in
the same basis --- the quasiprobability distribution $\tilde{p}_{\rho}$
defined in \ref{eq:mm3} reduces to a regular (non-negative) probability
distribution,
\[
\tilde{p}_{\rho}(i,a,b,\phi,\alpha,\beta)=\begin{cases}
p_{\rho}(i,\phi) & \text{if \ensuremath{i\!=\!a\!=\!b} and \ensuremath{\phi\!=\!\alpha\!=\!\beta}}\\
0 & \text{otherwise}
\end{cases}
\]
where $p_{\rho}(i,\phi)=p_{i}\mathrm{tr}\{\Phi(\vert i\rangle\langle i\vert)\vert\phi\rangle\langle\phi\vert\}$
(as appeared in \ref{eq:m2} and \ref{eq:m3}). Then, taking expectations
under $\tilde{p}_{\rho}(i,a,b,\phi,\alpha,\beta)$ is equivalent to
taking expectations under $p_{\rho}(i,\phi)$, which recovers
the ``commuting case'' results (presented in the previous section)
as a special case of the more general analysis discussed in this section.

\subsection{Example}

\label{subsec:Example-reset-process}

We now illustrate our results for fluctuating mismatch cost using
the example of a ``reset'' process (see also analyses in \citep{riechers2020initial,riechersImpossibilityLandauerBound2021}).

Consider a finite-dimensional quantum process that maps any initial
state $\rho$ to the same final pure state $\vert\phi\rangle\langle\phi\vert$,
so that the dynamics are described by the following input-independent
channel:
\begin{equation}
\Phi(\rho)=\vert\phi\rangle\langle\phi\vert\qquad\forall\rho.\label{eq:input-ind}
\end{equation}
This type of process can represent erasure of information (e.g., the
reset of a qubit) or the preparation of some special pure state (e.g.,
preparation of some desired entangled state). Let $\varphi\in\mathop{\arg\min}_{\omega\in\mathcal{D}}\Sigma(\omega)$
indicate the initial mixed state that minimizes EP for this process,
and note that we do not assume that $\varphi$ achieves vanishing
EP. From \ref{eq:m234} and \ref{subsec:Support-conditions-integrated}, it is easy to verify that $\varphi$ must
have full support.

Now suppose that the process is initialized on some initial mixed
state $\rho$. For simplicity, we assume that $\rho$ commutes with
$\varphi$, so that both can be diagonalized in the same basis ($\rho=\sum_{i}p_{i}\vert i\rangle\langle i\vert$
and $\varphi=\sum_{i}r_{i}\vert i\rangle\langle i\vert$). Since we
assume a finite-dimensional system and $\varphi$ has full support,
$S(\rho\Vert\varphi)<\infty$ and
\[
\Sigma(\rho)-\Sigma(\varphi)=-\Delta S({\rho}\Vert{\varphi})=S(\rho\Vert\varphi)
\]
by \ref{eq:equalityBasis}. This means that \ref{eq:z4} holds, allowing
us to apply the results we derived for fluctuating mismatch cost in
the commuting case, such as \ref{eq:fluct0} and \ref{eq:IFT}.

In particular, consider some realization of the process in which the
system goes from an initial pure state $\vert i\rangle\langle i\vert$
to the final pure state $\vert\phi\rangle\langle\phi\vert$. The fluctuating
mismatch cost for this realization can be written in the following
simple form:
\begin{equation}
\sigma_{\rho}(i\!\shortrightarrow\!\phi,q)-\sigma_{\varphi}(i\!\shortrightarrow\!\phi,q)=\ln p_{i}-\ln r_{i},\label{eq:fluct8}
\end{equation}
where we used \ref{eq:mm0} and the fact that $p'_{\phi}=r'_{\phi}=1$. \ref{eq:fluct8} means that
the fluctuating mismatch cost incurred in mapping $i\!\shortrightarrow\!\phi$
is the log ratio of the probability of pure state $\vert i\rangle\langle i\vert$
under the actual initial mixed state $\rho$ and the optimal initial mixed
state $\varphi$ that minimizes EP.

Recall that fluctuating mismatch cost obeys the IFT in \ref{eq:IFT}.
Given \ref{eq:fluctBound}, this means that the probability of observing
negative mismatch is exponentially unlikely: the probability that
$\sigma_{\varphi}(i\!\shortrightarrow\!\phi,q)$ exceeds $\sigma_{\rho}(i\!\shortrightarrow\!\phi,q)$
by $\xi$ (or more) is upper bounded by $e^{-\xi}$.

\section{Mismatch Cost for EP Rate}

\label{sec:eprate}

In our third set of results, we analyze the state dependence of the
instantaneous EP rate. We consider an open quantum system coupled
to some number of reservoirs, which evolves according to a Lindblad
equation, ${\textstyle \frac{d}{dt}}\rho(t)\!=\!\mathcal{L}(\rho(t))$.
The EP rate incurred by state $\rho$ is \citep{spohn1978irreversible,spohn_entropy_1978,alicki_quantum_1979}
\begin{equation}
\dot{\Sigma}(\rho)={\textstyle {\textstyle \frac{d}{dt}}}S(\rho(t))+\dot{Q}(\rho),\label{eq:EPr}
\end{equation}
where $\dot{Q}:\mathcal{D}\to\mathbb{R}$ is a linear function that
reflects the rate of entropy flow to the environment. Note that the rate
of entropy change ${\textstyle {\textstyle \frac{d}{dt}}}S(\rho(t))$
depends on the Lindbladian $\mathcal{L}$. As above, the precise definition
of $\mathcal{L}$ or $\dot{Q}$ will generally reflect various details
of the system and the coupled reservoirs. For simplicity, here we
assume that $\dim\mathcal{H}<\infty$ (results for the $\dim\mathcal{H}=\infty$
case, which require some additional technicalities, are left for
\ref{app:EPrate}).

It is important to note that the derivative in \ref{eq:EPr} is evaluated
at $t=0$, meaning that $\dot{\Sigma}(\rho)$ expresses the instantaneous
EP rate incurred at the same time that the system is found in state
$\rho$. An alternative analysis, which we do not consider here, would
consider the EP rate incurred at some later time $t>0$, given that
the process is initialized in state $\rho$ at $t=0$.

Consider some set of states $\mathcal{D}_{P}$, defined as in \ref{eq:densBdef}
for a set of projection operators $P$. Let $\varphi_{P}\in\mathop{\arg\min}_{\omega\in\mathcal{D}_{P}}\dot{\Sigma}(\omega)$
indicate the state which minimizes the EP rate within this set. Then,
for any $\rho\in\mathcal{D}_{P}$ such that $S(\rho\Vert\varphi_{P})<\infty$,
the additional EP rate incurred by $\rho$ above that incurred by
$\varphi_{P}$ is given by the instantaneous rate of contraction of
the relative entropy between $\rho$ and $\varphi_{P}$,
\begin{equation}
\dot{\Sigma}(\rho)-\dot{\Sigma}(\varphi_{P})=-{\textstyle {\textstyle \frac{d}{dt}}}S(\rho(t)\Vert\varphi_{P}(t)),\label{eq:EPr-equality}
\end{equation}
which is the continuous-time analogue of \ref{eq:equalityBasis}.
The proof of this result is sketched at the end of this section, with
details left for \ref{app:EPrate}.

We refer to the additional instantaneous EP rate incurred by $\rho$,
above that incurred by an optimal state $\varphi_{P}$, the \emph{instantaneous
mismatch cost} of $\rho$. In the special case where $\mathcal{D}_{P}=\mathcal{D}$
(when $P=\{I\}$), \ref{eq:EPr-equality} expresses the global instantaneous
mismatch cost, reflecting the additional EP rate incurred by state
$\rho$ rather than a global optimizer, $\varphi_{\mathcal{D}}\in\mathop{\arg\min}_{\omega\in\mathcal{D}}\dot{\Sigma}(\omega)$.

We can decompose instantaneous mismatch cost by applying \ref{eq:EPr-equality}
in an iterative manner. In particular, we can derive a decomposition
into classical and quantum contributions analogous to \ref{eq:mdecomp}.
As above, define $P=\{\vert i\rangle\langle i\vert\}_{i}$ for an
orthonormal basis $\{\vert i\rangle\}_{i}$ that diagonalizes $\rho$. Then,
let $\varphi_{P}\in\mathop{\arg\min}_{\omega\in\mathcal{D}_{P}}\dot{\Sigma}(\omega)$
be an optimal state within $\mathcal{D}_{P}$, and let $\varphi_{\mathcal{D}}\in\mathop{\arg\min}_{\omega\in\mathcal{D}}\dot{\Sigma}(\omega)$
be a global optimizer. Using a similar derivation as in \ref{eq:mdecomp},
we can decompose the global instantaneous mismatch cost into two non-negative terms,
\begin{align}
\dot{\Sigma}(\rho)-\dot{\Sigma}(\varphi_{\mathcal{D}}) & =[\dot{\Sigma}(\rho)-\dot{\Sigma}(\varphi_{P})]+[\dot{\Sigma}(\varphi_{P})-\dot{\Sigma}(\varphi_{\mathcal{D}})]\nonumber \\
& =-{\textstyle {\textstyle \frac{d}{dt}}}S(\rho(t)\Vert\varphi_{P}(t))-{\textstyle {\textstyle \frac{d}{dt}}}S(\varphi_{P}(t)\Vert\varphi_{\mathcal{D}}(t)).\label{eq:mdecomp-1}
\end{align}
The first term, reflecting the mismatch between $\rho$ and $\varphi_{P}$
which are diagonal in the same basis, is the classical contribution
to instantaneous mismatch cost. The second term, reflecting the mismatch
between $\varphi_{P}$ and $\varphi_{\mathcal{D}}$, vanishes when
$\rho$ and $\varphi_{\mathcal{D}}$ can be diagonalized in the same
basis, and is the quantum contribution to instantaneous mismatch cost.

Our most generally applicable result concerns the instantaneous mismatch
cost of $\rho$ relative to an optimal state within some arbitrary
convex subset of states $\mathcal{S}\subseteq\mathcal{D}$. Given
any state $\rho\in\mathcal{S}$ and an optimizer $\varphi_{\mathcal{S}}\in\mathop{\arg\min}_{\omega\in\mathcal{S}}\dot{\Sigma}(\omega)$,
as long as $S(\rho\Vert\varphi_{\mathcal{S}})<\infty$, it is the
case that
\begin{equation}
\dot{\Sigma}(\rho)-\dot{\Sigma}(\varphi_{\mathcal{S}})\ge-{\textstyle {\textstyle \frac{d}{dt}}}S(\rho(t)\Vert\varphi_{\mathcal{S}}(t)),\label{eq:ineqConvex-1}
\end{equation}
with equality if $(1-\lambda)\varphi_{\mathcal{S}}+\lambda\rho\in\mathcal{S}$
for some $\lambda<0$. Since $\dot{\Sigma}(\varphi_{\mathcal{S}})\ge0$
for Lindbladian dynamics \citep{spohn_entropy_1978}, \ref{eq:ineqConvex-1}
implies
\begin{equation}
\dot{\Sigma}(\rho)\ge-{\textstyle {\textstyle \frac{d}{dt}}}S(\rho(t)\Vert\varphi_{\mathcal{S}}(t)).\label{eq:genBound-1}
\end{equation}
The RHS is non-negative by the monotonicity of relative entropy. This
provides a tighter bound on the EP rate than the second law, $\dot{\Sigma}(\rho)\ge0$,
which reflects a suboptimal choice of the state within some convex
set of states.

We now briefly sketch the proof idea behind \ref{eq:EPr-equality,eq:ineqConvex-1},
leading formal details for \ref{app:EPrate}. First, we use \ref{eq:EPr}
to define an integrated EP function as $\Sigma(\rho,t)=\int_{0}^{t}\dot{\Sigma}(\rho(t'))dt'$.
Given a pair of states $\rho,\varphi$ with finite EP rate and $S(\rho\Vert\varphi)<\infty$,
we then write the directional derivative of $\dot{\Sigma}$ at $\varphi$
in the direction of $\rho$ as
\begin{align}
& {\textstyle \partial_{\lambda}^{+}}\dot{\Sigma}(\varphi(\lambda),t)\vert_{\lambda=0}:={\textstyle \partial_{\lambda}^{+}}\partial_{t}\Sigma(\varphi(\lambda),t)\nonumber \\
& \qquad=\partial_{t}{\textstyle \partial_{\lambda}^{+}}\Sigma(\varphi(\lambda),t)\nonumber \\
& \qquad=\partial_{t}[\Sigma(\rho,t)-\Sigma(\varphi,t)+\Delta S({\rho}\Vert{\varphi})]\nonumber \\
& \qquad=\dot{\Sigma}(\rho)-\dot{\Sigma}(\varphi)+{\textstyle {\textstyle \frac{d}{dt}}}S(\rho(t)\Vert\varphi(t)),\label{eq:ddN}
\end{align}
where $\varphi(\lambda)=(1-\lambda)\varphi+\lambda\rho$. In the second
line, we used the symmetry of partial derivatives, which (as we prove
in the \cref{app:EPrate}) follows from convexity of $\Sigma$. In the third
line, we used the expression for the directional derivative of integrated
EP, \ref{eq:dd0}. \ref{eq:ineqConvex-1} follows from \ref{eq:ddN},
since the directional derivative at a minimizer must be non-negative.
To derive \ref{eq:EPr-equality}, note that if $S(\rho\Vert\varphi)<\infty$,
then $\mathrm{supp}\,\rho\subseteq\mathrm{supp}\,\varphi$ and so
it is possible to move from $\varphi\in\mathcal{D}_{P}$ both toward
and away from $\rho\in\mathcal{D}_{P}$ while remaining within the
set $\mathcal{D}_{P}$. Since $\varphi$ is a minimizer of the EP
rate within $\mathcal{D}_{P}$, this means that the directional derivative
${\textstyle \partial_{\lambda}^{+}}\dot{\Sigma}(\varphi(\lambda),t)\vert_{\lambda=0}$
must vanish.

\subsection{Support conditions}

Our result for mismatch cost, \ref{eq:EPr-equality}, only applies
when $S(\rho\Vert\varphi_{P})<\infty$. This condition in turn requires
that
\begin{align}
\mathrm{supp}\,\rho\subseteq\mathrm{supp}\,\varphi_{P}.\label{eq:suppcond-1-1}
\end{align}
Here, we show that \ref{eq:suppcond-1-1} is satisfied in many cases
of interest.

In \ref{prop:epr-fsupp} in the appendix, we prove that \ref{eq:suppcond-1-1}
holds for all $\rho\in\mathcal{D}_{P}$ and $\varphi_{P}\in\mathop{\arg\min}_{\omega\in\mathcal{D}_{P}}\dot{\Sigma}(\omega)$
as long as the Lindbladian $\mathcal{L}$ satisfies the following
``irreducibility'' condition:
\begin{equation}
\mathrm{supp}\,\mathcal{L}(\rho)\not\subseteq\mathrm{supp}\,\rho\quad\forall\rho\in\mathcal{D}_{P}:\mathrm{supp}\,\rho\ne\mathcal{H}_{P},\label{eq:fsupp-mn}
\end{equation}
where $\mathcal{H}_{P}$ is defined as in \ref{eq:hilbB} (we also
assume that $\dim\mathcal{H}<\infty$). \ref{eq:fsupp-mn} says that
whenever some state $\rho$ with partial support evolves under $\mathcal{L}$,
some probability ``leaks out'' of subspace spanned by $\rho$. In
the terminology of \citep{baumgartnerAnalysisQuantumSemigroups2008a,baumgartnerAnalysisQuantumSemigroups2008},
\ref{eq:fsupp-mn} means that $\mathcal{L}$ does not have any non-trivial
``lazy subspaces''.

If $\mathcal{L}$ is not irreducible in the sense of \ref{eq:fsupp-mn},
it may be possible to decompose the overall Hilbert space into a set
of irreducible subspaces such that \ref{eq:fsupp-1} holds in each
one~\citep{baumgartnerAnalysisQuantumSemigroups2008a,baumgartnerAnalysisQuantumSemigroups2008}.
Such subspaces have been called \emph{enclosures} in the literature
(see \citep{carboneIrreducibleDecompositionsStationary2016} for details)
and are the continuous-time analogue of ``basins'' discussed above.
We leave analysis of instantaneous mismatch cost with multiple enclosures
for future work.

\subsection{Example}

We briefly illustrate our results for instantaneous mismatch cost
by deriving a novel bound on the EP rate incurred in a non-equilibrium
stationary state.

Consider a finite-dimensional system that evolves in continuous time
according to some Lindbladian $\mathcal{L}$. Assume that the system
is coupled to multiple reservoirs and has an associated non-equilibrium
stationary state $\pi$. In addition, let $\varphi\in\mathop{\arg\min}_{\omega\in\mathcal{D}}\dot{\Sigma}(\omega)$
be a state that achieves the minimal EP rate. \ref{eq:genBound-1}
then implies the following bound on the stationary EP rate,
\begin{align}
\dot{\Sigma}(\pi) & \ge-{\textstyle {\textstyle \frac{d}{dt}}}S(\pi(t)\Vert\varphi(t))=-{\textstyle \frac{d}{dt}}S(\pi\Vert\varphi(t)),\label{eq:m1-1}
\end{align}
where ${\textstyle \frac{d}{dt}}\pi(t)=\mathcal{L}(\pi)=0$ by assumption of stationarity.

\ref{eq:m1-1} shows that for any continuous-time process, the stationary
EP rate is lower bounded by the rate at which the minimally dissipative
state $\varphi$ approaches the stationary state $\pi$ in relative
entropy.

\section{Mismatch Cost in Classical Systems}

\label{sec:classical}

We now discuss mismatch cost in the context of classical systems.
We consider both discrete-state classical systems (as might be derived
by coarse-graining an underlying phase space \citep{talknerRateDescriptionFokkerPlanck2004})
and continuous-state classical systems. For more details, see \ref{app:classical}.

\subsection{Classical integrated EP}

\label{subsec:classical-Integrated-EP}

We begin by overviewing the definition of integrated EP in classical
systems.

Consider a classical system with state space $X$ which undergoes
a driving protocol over time interval $t\in[0,\tau]$, while
coupled to some thermodynamic reservoirs. We use the notation $p'$
to indicate the final probability distribution at time $t=\tau$ corresponding
to the initial probability distribution $p$ at time $t=0$. (Note that we use
the term ``probability distribution'' to indicate a probability
mass function for discrete-state systems and a probability density
function for continuous-state systems.) In addition, following classical
stochastic thermodynamics \citep{van2013stochastic,seifert2012stochastic},
we use $\mathrm{P}(\bm{x}\vert x_{0})$ to indicate the conditional
probability of the system undergoing the trajectory $\bm{x}=\{x_{t}:t\in[0,\tau]\}$
under the regular (``forward'') protocol given initial microstate
$x_{0}$. Sometimes we will also consider the conditional probability
$\tilde{\mathrm{P}}(\tilde{\bm{x}}\vert\tilde{x}_{\tau})$ of observing
the \emph{time-reversed} trajectory $\tilde{\bm{x}}=\{\tilde{x}_{\tau-t}:t\in[0,\tau]\}$
under the \emph{time-reversed} driving protocol given initial microstate
$\tilde{x}_{\tau}$ (tilde notation like $\tilde{x}$ indicates conjugation
of odd variables such as momentum \citep{ford_entropy_2012,spinneyNonequilibriumThermodynamicsStochastic2012}).

For classical systems, there are several ways of defining integrated
EP as a function of the initial probability distribution. The first
way is the classical analogue of \ref{eq:EP0-ent},
\begin{align}
\mathsf{\Sigma}(p)=\mathsf{S}(p')-\mathsf{S}(p)+G(p),\label{eq:classicEP}
\end{align}
where $\mathsf{\Sigma}(\cdot)$ indicates classical EP as a function
of the initial probability distribution, $\mathsf{S}(\cdot)$ indicates
classical Shannon entropy and $G(\cdot)$ is a linear function that
reflects the entropy flow to the environment. As above, the precise
definition of the entropy flow term will depend on the physical setup,
such as the number and type of coupled reservoirs.

A second way to define integrated EP in classical systems is in terms
of the relative entropy between the trajectory probability distribution
under the forward process and the time-reversed backward process
\citep{spinneyEntropyProductionFull2012,esposito_three_2010},
\begin{equation}
\mathsf{\Sigma}(p)=D\big(\mathrm{P}(\bm{X}\vert X_{0})p(X_{0})\Vert\tilde{\mathrm{P}}(\tilde{\bm{X}}\vert\tilde{X}_{\tau})p'(X_{\tau})\big),\label{eq:EPclassicalKL}
\end{equation}
where $D(\cdot\Vert\cdot)$ indicates the classical relative entropy
(also called the Kullback-Leibler divergence). \ref{eq:EPclassicalKL}
expresses integrated EP directly in terms of the ``time-asymmetry''
of the stochastic process \citep{parrondoEntropyProductionArrow2009a}.
Note that the expression in \ref{eq:EPclassicalKL} is a special case
of the expression in \ref{eq:classicEP}, since it can be put in the
form of the latter by defining the entropy flow in \ref{eq:classicEP}
as the expectation $G(p)=\Big\langle\ln\mathrm{P}(\bm{x}\vert x_{0})-\ln\tilde{\mathrm{P}}(\tilde{\bm{x}}\vert\tilde{x}_{\tau})\Big\rangle_{\mathrm{P}(\bm{x}\vert x_{0})p(x_{0})}$,
and then performing some simple rearrangement.

There is also a third way to define integrated EP for continuous-state
classical systems in phase space. Consider some system $X$, and let
$Y$ indicate its explicitly modeled environment (typically, $Y$
will indicate the state of one or more heat baths). Assume that $X$
and $Y$ jointly evolve in a Hamiltonian manner starting from an initial
distribution $p(x_{0},y_{0})=p(x_{0})\pi(y_{0}\vert x_{0})$ at time
$t=0$ to some final distribution $p'(x_{\tau},y_{\tau})$ at time
$t=\tau$, where $\pi(y_{0}\vert x_{0})$ is the conditional equilibrium
distribution induced by some system-environment Hamiltonian. The integrated
EP incurred by initial distribution $p$ can then be defined as
\begin{equation}
\mathsf{\Sigma}(p)=D({p'(X_{\tau},Y_{\tau})}\Vert{p'(X_{\tau})\pi({Y_{\tau}\vert X_{\tau}})}).\label{eq:contm2PRE}
\end{equation}
(see \mycitep[Eq.~15]{millerEntropyProductionTime2017}, \mycitep[Eq.~49]{strasberg_stochastic_2017},
and \citep{seifertFirstSecondLaw2016a}). \ref{eq:contm2PRE} is the
classical analogue of \ref{eq:functype2-maintext2}, though generalized
to allow equilibrium correlations between the environment and the
system (see discussion in Appendix A of \citep{strasberg_stochastic_2017}).

We now discuss mismatch cost for classical integrated EP. First, consider
a discrete-state classical system, such that the state space $X$
is a countable set. In this case, our results for quantum mismatch
cost can be directly applied, since a discrete-state classical process
can be expressed as a special case of a quantum process. In particular,
let $\mathcal{D}_{P}$, defined as in \ref{eq:densBdef}, indicate
the set of density operators diagonal in some fixed reference basis.
Then, any probability distribution $p$ over $X$ can be expressed
as a density operator in $\mathcal{D}_{P}$, and any classical dynamics
can be expressed as a special quantum channel that maps elements of
$\mathcal{D}_{P}$ to elements of $\mathcal{D}_{P}$ (see \ref{app:classical-discrete}
for details). Under this mapping, the expressions for quantum and
classical EP (\ref{eq:EP0-ent} versus \cref{eq:classicEP,eq:EPclassicalKL,eq:contm2PRE})
become equivalent, and we can analyze mismatch cost for classical
integrated EP using the results presented above, such as \ref{eq:equalityBasis,eq:ineqConvex}.
For instance, we have the following classical analogue of \ref{eq:equalityBasis}:
given any initial distribution $p$ and an optimal initial distribution
within the set of all distributions, $r\in\mathop{\arg\min}_{s}\mathsf{\Sigma}(s)$,
mismatch cost can be written as
\begin{equation}
\mathsf{\Sigma}(p)-\mathsf{\Sigma}(r)=-\Delta D(p\Vert r),\label{eq:clmm}
\end{equation}
as long as $D(p\Vert r)<\infty$.

For classical systems in continuous state space, such that $X\subseteq\mathbb{R}^{n}$,
the mapping from our quantum results to classical mismatch cost is
not as direct, because in general it is not possible to represent
a continuous probability distribution in terms of a density operator
over a separable Hilbert space. Nonetheless, as long as an appropriate
``translation'' is carried out, the same proof techniques used to
derive mismatch cost results for quantum integrated EP can also be
used to derive analogous results for continuous classical systems,
such as \ref{eq:clmm}. This translation is described in detail in
\ref{app:classical-continuous}.

\subsection{Classical fluctuating EP}

Next, we show that our results for fluctuating mismatch cost also
apply to classical systems. The underlying logic of the derivation
is the same as for the commuting case for quantum systems described
in \ref{subsec:Mismatch-cost-forcomm}, though with somewhat different
notation.

Consider a classical system that undergoes a physical process, which
starts from the initial distribution $p$ and ends on the final distribution
$p'$. In general, the fluctuating EP incurred by some state trajectory
$\bm{x}$ can be expressed as \citep{seifert2012stochastic}
\[
\sigma_{p}(\bm{x})=\ln p(x_{0})-\ln p'(x_{\tau})+q(\bm{x}),
\]
where $q(\bm{x})$ is the increase of the entropy of all coupled reservoirs
incurred by trajectory $\bm{x}(t)$. Let $r$ indicate the initial
probability distribution that minimizes EP, so that \ref{eq:clmm}
holds, and let $r'$ indicate the corresponding final distribution.
We define classical fluctuating mismatch cost as the difference between
the fluctuating EP incurred by the trajectory $\bm{x}$ under the
actual initial distribution $p$ and the optimal initial distribution
$r$,
\begin{multline}
\sigma_{p}(\bm{x})-\sigma_{r}(\bm{x})=\\{}
[-\ln p'(x_{\tau})+\ln p(x_{0})]-[-\ln r'(x_{\tau})+\ln r(x_{0}))],\label{eq:clfluctm-3}
\end{multline}
which is the classical analogue of \ref{eq:mDef0}. It is easy to
verify that \ref{eq:clfluctm-3} is the proper trajectory-level expression
of classical mismatch cost,
\begin{equation}
\langle\sigma_{p}-\sigma_{r}\rangle_{\mathrm{P}(\bm{x}\vert x_{0})p(x_{0})}=-\Delta D(p\Vert r)=\mathsf{\Sigma}(p)-\mathsf{\Sigma}(r).\label{eq:j3}
\end{equation}
Moreover, using a derivation similar to the one in \ref{app:fluctuating},
it can be shown that \ref{eq:clfluctm-3} obeys an IFT,
\begin{equation}
\langle e^{-(\sigma_{p}-\sigma_{r})}\rangle_{\mathrm{P}(\bm{x}\vert x_{0})p(x_{0})}=\gamma,\label{eq:clift-3}
\end{equation}
where $\gamma\in(0,1]$ is a correction factor that equals 1 when $p$
and $r$ have the same support (see \ref{eq:classicalGamma} in the
appendix). \ref{eq:j3} and \ref{eq:clift-3} are the classical analogues
of \ref{eq:fluct0} and \ref{eq:IFT} respectively. Moreover, the
IFT in \ref{eq:clift-3} implies that same exponential bound on negative
mismatch cost as in \ref{eq:fluctBound}.

We can also derive the classical analogue of \ref{eq:ldb1}, which
expresses fluctuating mismatch cost in terms of the breaking of time-reversal
symmetry. Note that for a classical system, the Petz recovery map
is simply the Bayesian inverse of the conditional probability distribution
$\mathrm{P}(x_{\tau}\vert x_{0})$ with respect to the probability
distribution $r$, $\mathrm{P}(x_{0}\vert x_{\tau})=\mathrm{P}(x_{\tau}\vert x_{0})r(x_{0})/r'(x_{\tau})$
\citep{leiferFormulationQuantumTheory2013,wildeQuantumInformationTheory2017}.
In \ref{app:classical}, we show that
\begin{align}
\sigma_{p}(\bm{x})-\sigma_{r}(\bm{x}) & =\ln\frac{\mathrm{P}(x_{\tau}\vert x_{0})p(x_{0})}{\mathrm{P}(x_{0}\vert x_{\tau})p'(x_{\tau})},\label{eq:clPetz-2}
\end{align}
Thus, the fluctuating mismatch cost for a classical system quantifies
the time-asymmetry between the forward process and the reverse process,
as defined by the Bayesian inverse of the forward process run on the
optimal initial distribution $r$.

For more detailed derivations, see \ref{app:classical-fluct-discrete}
for discrete-state classical systems, and \ref{app:classical-fluct-cont}
for continuous-state classical systems.

\subsection{Classical EP rate}

Finally, we discuss instantaneous mismatch cost in the context of
classical systems.

Consider a classical system whose probability distribution at time
$t=0$ evolves according to a master equation,
\[
{\textstyle \frac{d}{dt}}p(t)=Lp(t),
\]
where $L$ is a linear operator that is the infinitesimal generator
of the dynamics. For a discrete-state classical system, $L$ will
be a rate matrix specifying transitions rates between different states,
while for a continuous-state classical system, $L$ will typically
be a Fokker-Planck operator. The classical EP rate can be written
as \citep{esposito2010three}
\begin{align}
\dot{\mathsf{\Sigma}}(p)={\textstyle \frac{d}{dt}}\mathsf{S}(p(t))+\dot{G}(p),\label{eq:classicEPrate-2}
\end{align}
where $\dot{G}(p)$ is the rate of entropy flow to environment.
As always, the form of $\dot{G}(p)$ will depend on the specifics
of the physical process, but can generally be expressed as an expectation
of some function over the microstates. \ref{eq:classicEPrate-2} is
the classical analogue of \ref{eq:EPr}.

A discrete-state classical system can be formulated as
a special case of a quantum system, as mentioned above in \ref{subsec:classical-Integrated-EP}
and described in more detail in \ref{app:classical}. In particular,
one can always express a discrete classical distribution as a density
matrix and a discrete rate matrix as a specially-constructed Lindbladian.
Under this mapping, the expressions for quantum and classical EP rate
(\ref{eq:EPr} versus \ref{eq:classicEPrate-2}) become equivalent,
and we can analyze instantaneous mismatch cost for discrete-state classical systems
using the results presented above for quantum systems, such as \ref{eq:EPr-equality},
\ref{eq:ineqConvex-1}, and \ref{eq:genBound-1}. In particular, we
have the following classical analogue of \ref{eq:EPr-equality}: given
any distribution $p$ and an optimal distribution $r\in\mathop{\arg\min}_{s}\dot{\mathsf{\Sigma}}(s)$
which minimizes EP rate,
\begin{equation}
\dot{\mathsf{\Sigma}}(p)-\dot{\mathsf{\Sigma}}(r)=-{\textstyle \frac{d}{dt}}D(p(t)\Vert r(t)),\label{eq:mnnn2-3-1-1-1}
\end{equation}
as long as $D(p\Vert r)<\infty$.

As mentioned above, for continuous-state classical systems, the mapping
to the quantum formalism is not as direct. Nonetheless, the same proof
techniques used to derive instantaneous mismatch cost for quantum
systems can be used to derive analogous results for continuous classical
systems, such as \ref{eq:mnnn2-3-1-1-1}. This can be done as long
as an appropriate ``translation'' is carried out between classical
and quantum formulations, which is described in detail in \ref{app:classical-EPrate-cont}.

\section{Logical vs. Thermodynamic Irreversibility}

\label{sec:logicalvsthermo}

The relationship between thermodynamic irreversibility (generation
of EP) and logical irreversibility (inability to know the initial
state corresponding to a given final state) is one of the foundational
issues in the thermodynamics of computation \citep{bennettNotesHistoryReversible1988}.
Despite some confusion in the early literature, it is now well-understood
that logically irreversible operations can in principle be carried
out in a thermodynamically reversible manner, without generating any
EP \citep{maroney_absence_2005,sagawa2014thermodynamic,wolpert2019stochastic}.

At the same time, our results demonstrate a different kind of universal
relationship between logical and thermodynamic irreversibility. By
\ref{eq:ineqConvex}, the mismatch cost of $\rho$ is lower-bounded
by the contraction of relative entropy $-\Delta S({\rho}\Vert{\varphi})$,
which is a principled information-theoretic measure of the logical
irreversibility of the quantum channel $\Phi$ on the pair of states
$\rho,\varphi$. This measure reaches its maximal value of $S(\rho\Vert\varphi)$
if and only if $\Phi(\varphi)=\Phi(\rho)$, in which case all information
about the choice of initial state ($\rho$ vs. $\varphi$) is lost.
It reaches its minimal value of 0 if and only if the map $\Phi$ is
logically reversible on the pair of states $\rho,\varphi$, meaning
that the Petz recovery map $\mathcal{R}_{\Phi}^{\varphi}$ can perfectly
restore both initial states $\rho$ and $\varphi$ from the output
of $\Phi$, $\mathcal{R}_{\Phi}^{\varphi}(\Phi(\varphi))=\varphi$
and $\mathcal{R}_{\Phi}^{\varphi}(\Phi(\rho))=\rho$ \citep{petz1988sufficiency,mosonyi2004structure,wildeRecoverabilityQuantumInformation2015a}.
For a unitary channel, $-\Delta S({\rho}\Vert{\varphi})=0$ for all
pairs of states $\rho,\varphi$.

Now imagine a physical process that implements some map $\Phi$ and
achieves minimal EP on some initial state $\varphi$. Our results
imply that the thermodynamic cost associated with choosing suboptimal
initial states, in terms of the additional EP that is generated on
those initial states above the minimum possible, increases with degree
of logical irreversibility of the channel $\Phi$. This is consistent
with the fact that the minimal EP incurred by a given process that
implements $\Phi$, $\min_{\omega}\Sigma(\omega)$, does not directly
depend on the logical irreversibility of $\Phi$ (and can
vanish even for logically irreversible channels).

Interestingly, recent work has uncovered the following inequality
between the contraction of relative entropy and the accuracy of ``recovery
maps''~\citep{jungeUniversalRecoveryMaps2018},
\begin{equation}
-\Delta S({\rho}\Vert{\varphi})\ge-2\ln F(\rho,\mathcal{N}_{\Phi}^{\varphi}(\Phi(\rho))),\label{eq:recoveryBound}
\end{equation}
where $F(\cdot,\cdot)$ is fidelity and $\mathcal{N}_{\Phi}^{\varphi}$
is a recovery map closely related to \ref{eq:recov}. This inequality
provides an information-theoretic condition for high-fidelity recovery
of an initial state $\rho$ that undergoes a noisy operation $\Phi$,
which is of fundamental interest in quantum error correction. Now
consider a process that implements the map $\Phi$ and achieves minimal
EP on the initial state $\varphi$. Combining \ref{eq:equalityBasis,eq:recoveryBound}
along with $\Sigma(\varphi)\ge0$ gives the inequality
\[
F(\rho,\mathcal{N}_{\Phi}^{\varphi}(\Phi(\rho))\ge e^{-\Sigma(\rho)/2},
\]
which implies that high-fidelity recovery of $\rho$ by $\mathcal{N}_{\Phi}^{\varphi}$
is possible only if the process incurs a small amount of EP on the
initial state $\rho$. Conversely, if $\mathcal{N}_{\Phi}^{\varphi}$
performs poorly at recovering $\rho$, then the EP incurred by initial
state $\rho$ must be large. While this relationship between the fidelity
of recovery and EP has been discussed for simple relaxation processes
to equilibrium~\citep{alhambraWorkReversibilityQuantum2018a}, our
results show that it actually holds for a much broader set of processes,
including ones with arbitrary driving and possibly coupled to multiple
reservoirs.

Motivated by these results, and in the spirit of recent work on information-theoretic
characterization of quantum channels~\citep{faist2019thermodynamic,gourHowQuantifyDynamical2019,gourEntropyQuantumChannel2020},
we propose the following measure of the logical irreversibility of
a given map $\Phi$:
\begin{equation}
\Lambda(\Phi):=\inf_{\varphi\in\mathcal{D}}\sup_{\rho\in\mathcal{D}}-\Delta S({\rho}\Vert{\varphi}).\label{eq:Cmeas}
\end{equation}
Our measure has a simple operational interpretation in thermodynamic
terms: for \emph{any} physical process that implements $\Phi$, there
must be \emph{some} initial state that incurs EP of at least $\Lambda(\Phi)$,
as follows from \ref{eq:equalityBasis} and $\Sigma(\varphi) \ge 0$. $\Lambda(\Phi)$ can be related
to some existing measures of logical irreversibility, such as the
``contraction coefficient of relative entropy'' from quantum information
theory \citep{choi_equivalence_1994,raginsky2002strictly,hiai2016contraction},
$\eta(\Phi)=\sup_{\rho\ne\omega}S(\Phi(\rho)\Vert\Phi(\omega))/S(\rho\Vert\omega)$.
Some simple algebra shows that $\Lambda(\Phi)\ge(1-\eta(\Phi))\ln d$,
where $d$ the dimension of the Hilbert space. More generally, it
is easy to verify that $\Lambda(\Phi)$ achieves its minimum value
of 0 if $\Phi$ is unitary, and achieves its maximum value of $\ln d$
if $\Phi$ is input-independent (where $d$ the dimension of the Hilbert
space).

Some care should be taken in relating these results to earlier arguments
concerning ``reversible computation''. Suppose that one wishes to
implement some logically irreversible map $\Phi$ while minimizing
EP. Our results show that it is possible to completely eliminate the
mismatch cost of running $\Phi$ by ``embedding'' $\Phi$ within
some larger logically reversible (i.e., unitary) map $\Phi'$, since
$\Lambda(\Phi')=0$. This is related to the idea of using logically
reversible embeddings to reduce the minimal generated heat involved
in carrying out a logically irreversible computation~\citep{bennett1973logical,bennettNotesHistoryReversible1988,fredkin1982conservative,
bennett1989time,levine1990note,lange2000reversible}.

At the same time, this strategy incurs an additional ``storage cost''
of having to encode extra output information in physical degrees of
freedom \citep{bennett1973logical}. This additional cost, which would
not exist in a direct (i.e., logically irreversible) implementation
of $\Phi$, can itself be interpreted thermodynamically, since it
involves an increase of the entropy of those extra physical degrees
of freedom (for further discussion of related issues, see \mycitep[Sec.~11]{wolpert2019stochastic}).
Thus, when considering implementing a desired channel $\Phi$ via
some larger embedding $\Phi'$, there is a tradeoff between mismatch
cost (which decreases as the logical reversibility of $\Phi'$ increases)
and storage cost (which increases as the logical reversibility of
$\Phi'$ increases).

Of course, one can avoid the storage cost by first carrying out the
larger embedding $\Phi'$ and then erasing the additional output information
with an erasure map $\Phi''$, so that the combined map recovers the
original logically irreversible map, $\Phi=\Phi''\circ\Phi'$. In
this case, however, the combined operation has mismatch cost of $S(\rho\Vert\varphi)-S(\Phi(\rho)\Vert\Phi(\varphi))$
on initial state $\rho$, where $\varphi$ is the initial state that
minimizes EP for this combined operation. For any given $\rho$, this
mismatch cost may be larger or smaller than the mismatch cost incurred
by some other implementation of $\Phi$ (such as an implementation
that does not make use of logically reversible intermediate steps),
depending on the optimal initial state of that other implementation
(see also related discussion in \ref{sec:integratedEP}).

\subsection{Example}

Consider a qubit which undergoes an input-independent reset process,
so that all input states $\rho$ are mapped to the output pure state
$\vert0\rangle\langle0\vert$,
\[
\Phi(\rho)=\vert0\rangle\langle0\vert\quad\forall\rho.
\]
(See also \ref{subsec:Example-reset-process}.) For this process,
\begin{align}
\Lambda(\Phi) & :=\inf_{\varphi\in\mathcal{D}}\sup_{\rho\in\mathcal{D}}-\Delta S({\rho}\Vert{\varphi})\nonumber \\
& =\inf_{\varphi\in\mathcal{D}}\sup_{\rho\in\mathcal{D}}S(\rho\Vert\varphi)\label{eq:gf2}
\end{align}
It is easy to verify that this optimization problem is solved by taking
$\varphi$ to be the maximally mixed state, $\varphi=(\vert0\rangle\langle0\vert+\vert1\rangle\langle1\vert)/2$,
and taking $\rho$ to be any pure state \footnote{Since relative entropy is convex in both arguments, the inner maximization
is satisfied by some pure state. Then, \ref{eq:gf2} can be rewritten
as \unexpanded{$\inf_{\varphi\in\mathcal{D}}\sup_{\vert i\rangle\langle i\vert\in\mathcal{D}}\langle i\vert-\ln\varphi\vert i\rangle$}.
Using the minimax principle, the inner maximization is satisfied by
the largest eigenvalue of \unexpanded{$-\ln\varphi$}, which for
a density matrix has to be no less than \unexpanded{$-\ln\frac{1}{d}=\ln d=\ln2$}.}, which gives
\[
\Lambda(\Phi)=\ln2.
\]
Thus, for \emph{any} physical implementation of a qubit reset, there
must exist initial states $\rho$ which have $\Sigma(\rho)\ge\ln2$.

\section{Mismatch Cost Beyond EP}

\label{sec:EP-Type-functionals}

This paper was formulated in terms of the state dependence of entropy
production. However, our results also apply to many other important
``cost functions'' that appear in nonequilibrium thermodynamics
and quantum information theory. In particular, as we show in \ref{app:proofs-integrated}
and \ref{app:classical}, our results for mismatch costs hold not
only for EP, but for any ``EP-type'' function $C(\rho)$ that can
be written in the following general form:
\begin{equation}
C(\rho)=S(\Phi(\rho))-S(\rho)+F(\rho),\label{eq:cost0}
\end{equation}
where $\Phi$ is any quantum channel (which may have different input
and output Hilbert spaces) and $F$ is any linear functional. Our
expression of EP, \ref{eq:EP0-ent}, is a special case of \ref{eq:cost0},
which arises when $F$ is defined as the entropy flow $Q$. (There
is also an analogous generalization of EP as alternatively defined in \ref{eq:functype2-maintext2};
see the appendix for details).

There are many costs beyond EP can be expressed in the form of \ref{eq:cost0},
including:
\begin{enumerate}[wide,labelindent=0pt,labelwidth=!]
\item \emph{Nonadiabatic EP}, the contribution to EP arising from the system
being out of stationarity. For a Markovian system evolving over $t\in[0,\tau]$,
nonadiabatic EP can be written as \citep{horowitz2013entropy,horowitz2014equivalent,esposito2010three,manzanoQuantumFluctuationTheorems2018},
\[
C(\rho)=S(\Phi(\rho))-S(\rho)-\int_{0}^{\tau}\mathrm{tr}\{(\partial_{t}\Phi_{t}(\rho))\ln\rho_{t}^{\text{st}}\}\,dt,
\]
where $\Phi_{t}(\rho)$ is the system's state at time $t$ given initial
state $\rho$ (and $\Phi(\rho)=\Phi_{\tau}(\rho)$) and $\rho_{t}^{\text{st}}$
is the nonequilibrium stationary state at time $t$. With some rearranging,
nonadiabatic EP for non-Markovian evolution \citep{garcia2012nonadiabatic}
and for quantum processes with measurement \citep{manzanoNonequilibriumPotentialFluctuation2015,manzanoQuantumFluctuationTheorems2018}
can also be put into the form of \ref{eq:cost0}.
\item The \emph{free energy loss} \citep{kolchinsky2017maximizing,faist2019thermodynamic,navascues2015nonthermal,muller2018correlating},
\begin{align}
C(\rho) & =\beta[\mathcal{F}_{\beta}(\rho,H_{0})-\mathcal{F}_{\beta}(\Phi(\rho),H_{\tau})]\nonumber \\
& =S(\Phi(\rho))-S(\rho)+\beta\mathrm{tr}\{\rho(H_{0}-\Phi^{\dagger}(H_{\tau}))\},\label{eq:distillable}
\end{align}
where $\mathcal{F}_{\beta}(\rho,H)=\mathrm{tr}\{\rho H\}-\beta^{-1}S(\rho)$
is the nonequilibrium free energy at inverse temperature $\beta$,
and $H_{0}$ and $H_{\tau}$ are the initial and final Hamiltonians.
Note that \ref{eq:distillable} can be negative, in which case it
reflects a net \emph{gain} of free energy (from this point of view,
the optimal initial state that minimizes \ref{eq:distillable} can
also be seen as the state that maximizes harvesting of free energy~\citep{kolchinsky2017maximizing}).
When $H_{0}=H_{\tau}$, the minimum value of \ref{eq:distillable}
across all initial states, as would be achieved by an optimizer $\varphi$,
is sometimes called the ``thermodynamic capacity'' of $\Phi$, and
provides operational bounds on quantum work extraction \citep{navascues2015nonthermal,faist2019thermodynamic}.
\item The drop of \emph{availability}, which is also called extractable
work \citep{schlogl1985thermodynamic,deffnerNonequilibriumEntropyProduction2011a,kolchinsky2017maximizing,esposito2011second},
\begin{align}
C(\rho) & =S(\rho\Vert\pi_{0})-S(\Phi(\rho)\Vert\pi_{\tau})\nonumber \\
& =S(\Phi(\rho))-S(\rho)+\mathrm{tr}\{\rho[\Phi^{\dagger}(\ln\pi_{\tau})-\ln\pi_{0}]\}\label{eq:dist2}
\end{align}
where $\pi_{t}=e^{-\beta H_{t}}/Z$ is the Gibbs state at time $t$.
Note that the difference between \ref{eq:dist2} and \ref{eq:distillable}
is $\beta$ times the decrease of \emph{equilibrium} free energy,
which is a constant that doesn't depend on $\rho$ and vanishes when
$H_{0}=H_{\tau}$.
\item The \emph{entropy gain} of a given channel and initial state \citep{holevo2011entropy,ramakrishnan2019non,holevo2011entropyB,buscemiApproximateReversibilityContext2016,das2018fundamental,alicki2004isotropic},
\begin{equation}
C(\rho)=S(\Phi(\rho))-S(\rho).\label{eq:eg}
\end{equation}
The minimum entropy gain for a given channel has been considered when
analyzing the capacity of quantum channels \citep{alicki2004isotropic}.
\end{enumerate}
It turns out that our results for mismatch cost for integrated EP,
such as \ref{eq:equalityBasis} and \ref{eq:ineqConvex}, apply to
all EP-type functions having the form \ref{eq:cost0}, including all
of the costs listed above. In particular, the additional cost incurred
by some state $\rho$, relative to an optimal state $\varphi\in\mathop{\arg\min}_{\omega}C(\omega)$
that minimizes that cost, has the universal information-theoretic
form,
\[
C(\rho)-C(\omega)=-\Delta S({\rho}\Vert{\varphi}),
\]
as long as the assumptions behind \ref{eq:equalityBasis} are satisfied.
(See also Appendix O in \citep{riechers2020initial} for a related
analysis.)

It is important to note that the optimal state $\varphi$ will vary
depending on the cost; for instance, in general the state that minimizes
drop of nonequilibrium free energy will not be the same
state that minimizes EP. Also, unlike EP, not all EP-type
functions are non-negative; for instance, the entropy gain incurred
by a given $\rho$ may in general be positive or negative. While our
main results do not assume that the non-negativity of EP-type functions,
some expressions (such as \ref{eq:genBound}) do assume non-negativity,
and therefore do not hold for those EP-type functions which may be
positive or negative.

Similarly, our results for fluctuating mismatch cost, such as \ref{eq:fluct0}
and \ref{eq:IFT}, hold for any fluctuating expression of the form
$-\ln\lambda_{\phi}^{\Phi(\rho)}+\ln\lambda_{i}^{\rho}+f$, where
$f$ is some arbitrary trajectory-level term. Different fluctuating
costs can be considered by selecting different $f$, including not
only fluctuating EP (\ref{eq:fluctEP}, which arises when $f$ is
defined as the trajectory-level entropy flow) but also fluctuating
nonadiabatic EP \citep{esposito_three_2010,manzanoQuantumFluctuationTheorems2018},
fluctuating drop in nonequilibrium free energy, and so on.

Finally, our instantaneous mismatch cost results, such as \ref{eq:EPr-equality}
and \ref{eq:ineqConvex-1}, hold for a general family of ``EP rate''-type
functions, which can be written as $\dot{C}(\rho)={\textstyle {\textstyle \frac{d}{dt}}}S(\rho(t))+\dot{F}(\rho)$,
where $\dot{F}$ is some arbitrary linear function. By appropriate
choice of $\dot{F}$, our results apply to the instantaneous rates
of various EP-type functions, such as the costs outlined above. For
example, our results imply that the rate of free energy loss incurred
by some state $\rho$, additional to that incurred by an optimal state
$\varphi$ that minimizes the rate of free energy loss, is given by
$-{\textstyle {\textstyle \frac{d}{dt}}}S(\rho(t)\Vert\varphi(t))$.

\section{Discussion}

\label{sec:Discussion}

EP is a central quantity of interest in both classical and quantum
thermodynamics. In this paper, we analyze how the EP incurred by a
fixed physical process varies as one changes the initial state of
a fixed physical process. We derive a universal information-theoretic
expression for the additional EP incurred by some initial state $\rho$,
relative to the optimal initial state $\varphi$ which minimizes EP. We show that versions of this result hold for integrated EP, fluctuating
trajectory-level EP, and instantaneous EP rate. Our approach can be
contrasted to much of the existing research in the field, which
considers how EP varies as one changes the driving protocol that is
applied to some fixed initial state.

At a high level, our results can be interpreted as a kind of ``strengthening''
of the second law of thermodynamics. The second law states that integrated
EP is non-negative, $\Sigma(\rho)\ge0$, as is the EP rate for Markovian
dynamics, $\dot{\Sigma}(\rho)\ge0$. We show that when the initial
state of a process is chosen sub-optimally, these bounds can be tightened,
via \ref{eq:genBound} and \ref{eq:genBound-1}. Similarly, stochastic
thermodynamics has demonstrated that fluctuating trajectory-level
EP obeys an integral fluctuation theorem, $\langle e^{-\sigma_{\rho}}\rangle=1$,
which implies that negative EP values are exponentially unlikely,
$P(\sigma_{\rho}<-\xi)<e^{-\xi}$. We show that, when the initial
state of a process is chosen sub-optimally, this fluctuation theorem
and bound can be modified via \ref{eq:IFT} and \ref{eq:fluctBound}.

It is interesting to note that, that unlike most work in stochastic
thermodynamics, our results do make explicit use of the connection
between entropy production and breaking of time-reversal symmetry.
Instead, they are derived by exploiting the algebraic structure of
EP, along with the mathematical properties of convex optimization.
Nonetheless, as we discuss in \ref{sec:Fluctuating-EP}, one can interpret
mismatch cost as implicitly referring to a violation of time-reversal
symmetry by a ``Bayesian inverse'' process, as expressed in \ref{eq:ldb1}
using the Petz recovery map.

Due to their generality and simplicity, we believe that
our results will be useful for analyzing the thermodynamics of various
biological and artificial systems, including engines and energy-harvesting
devices \citep{kolchinsky2017maximizing}, information-processing
systems \citep{wolpert2020thermodynamic,riechers2020initial,riechersImpossibilityLandauerBound2021},
and even quantum computers. Ultimately, they should also help in
design of such systems. Moreover, as we demonstrate in \ref{sec:logicalvsthermo},
our results imply a universal relationship between thermodynamic and
logical irreversibility, which we argue has implications for the thermodynamics
of quantum error correction.

\section*{Acknowledgements}

We thank Paul Riechers and Camille L. Latune for helpful comments
on this manuscript. We thank the Santa Fe Institute for helping to
support this research. This research was supported by Grant No. CHE-1648973
from the U.S. National Science Foundation, Grant No. FQXi-RFP-1622
from the Foundational Questions Institute, and Grant No. FQXi-RFP-IPW-1912
from the Foundational Questions Institute and Fetzer Franklin Fund,
a donor advised fund of Silicon Valley Community Foundation.

\noindent 
\bibliographystyle{IEEEtran}
\bibliography{main}

\begin{thebibliography}{100}
\providecommand{\url}[1]{#1}
\csname url@samestyle\endcsname
\providecommand{\newblock}{\relax}
\providecommand{\bibinfo}[2]{#2}
\providecommand{\BIBentrySTDinterwordspacing}{\spaceskip=0pt\relax}
\providecommand{\BIBentryALTinterwordstretchfactor}{4}
\providecommand{\BIBentryALTinterwordspacing}{\spaceskip=\fontdimen2\font plus
\BIBentryALTinterwordstretchfactor\fontdimen3\font minus
  \fontdimen4\font\relax}
\providecommand{\BIBforeignlanguage}[2]{{%
\expandafter\ifx\csname l@#1\endcsname\relax
\typeout{** WARNING: IEEEtran.bst: No hyphenation pattern has been}%
\typeout{** loaded for the language `#1'. Using the pattern for}%
\typeout{** the default language instead.}%
\else
\language=\csname l@#1\endcsname
\fi
#2}}
\providecommand{\BIBdecl}{\relax}
\BIBdecl

\bibitem{seifert2012stochastic}
U.~Seifert, ``Stochastic thermodynamics, fluctuation theorems and molecular
  machines,'' \emph{Reports on Progress in Physics}, vol.~75, no.~12, p.
  126001, 2012.

\bibitem{deffnerQuantumThermodynamicsIntroduction2019}
S.~Deffner and S.~Campbell, \emph{\BIBforeignlanguage{en}{Quantum
  {{Thermodynamics}}: {{An Introduction}} to the {{Thermodynamics}} of
  {{Quantum Information}}}}.\hskip 1em plus 0.5em minus 0.4em\relax {Morgan \&
  Claypool Publishers}, Jul. 2019.

\bibitem{esposito2010entropy}
M.~Esposito, K.~Lindenberg, and C.~Van~den Broeck, ``Entropy production as
  correlation between system and reservoir,'' \emph{New Journal of Physics},
  vol.~12, no.~1, p. 013013, 2010.

\bibitem{deffnerNonequilibriumEntropyProduction2011a}
S.~Deffner and E.~Lutz, ``\BIBforeignlanguage{en}{Nonequilibrium {{Entropy
  Production}} for {{Open Quantum Systems}}},''
  \emph{\BIBforeignlanguage{en}{Physical Review Letters}}, vol. 107, no.~14,
  Sep. 2011.

\bibitem{landi2020irreversible}
G.~T. Landi and M.~Paternostro, ``Irreversible entropy production: From
  classical to quantum,'' \emph{Reviews of Modern Physics}, vol.~93, no.~3, p.
  035008, 2021.

\bibitem{van2013stochastic}
C.~Van~den Broeck \emph{et~al.}, ``Stochastic thermodynamics: A brief
  introduction,'' \emph{Phys. Complex Colloids}, vol. 184, pp. 155--193, 2013.

\bibitem{jarzynski_equalities_2011}
C.~Jarzynski, ``Equalities and inequalities: irreversibility and the second law
  of thermodynamics at the nanoscale,'' \emph{Annu. Rev. Condens. Matter
  Phys.}, vol.~2, no.~1, pp. 329--351, 2011.

\bibitem{gingrich2016dissipation}
T.~R. Gingrich, J.~M. Horowitz, N.~Perunov, and J.~L. England, ``Dissipation
  bounds all steady-state current fluctuations,'' \emph{Physical Review
  Letters}, vol. 116, no.~12, p. 120601, 2016.

\bibitem{gingrich2017fundamental}
T.~R. Gingrich and J.~M. Horowitz, ``Fundamental bounds on first passage time
  fluctuations for currents,'' \emph{Physical Review Letters}, vol. 119,
  no.~17, p. 170601, 2017.

\bibitem{sivak2012thermodynamic}
D.~A. Sivak and G.~E. Crooks, ``Thermodynamic metrics and optimal paths,''
  \emph{Physical Review Letters}, vol. 108, no.~19, p. 190602, 2012.

\bibitem{esposito2010finite}
M.~Esposito, R.~Kawai, K.~Lindenberg, and C.~Van~den Broeck, ``Finite-time
  thermodynamics for a single-level quantum dot,'' \emph{EPL (Europhysics
  Letters)}, vol.~89, no.~2, p. 20003, 2010.

\bibitem{shiraishi_speed_2018}
N.~Shiraishi, K.~Funo, and K.~Saito, ``\BIBforeignlanguage{en}{Speed limit for
  classical stochastic processes},'' \emph{\BIBforeignlanguage{en}{Physical
  Review Letters}}, vol. 121, no.~7, Aug. 2018.

\bibitem{wilming_second_2016}
H.~Wilming, R.~Gallego, and J.~Eisert, ``\BIBforeignlanguage{en}{Second law of
  thermodynamics under control restrictions},''
  \emph{\BIBforeignlanguage{en}{Physical Review E}}, vol.~93, no.~4, Apr. 2016.

\bibitem{kolchinsky2020entropy}
A.~Kolchinsky and D.~H. Wolpert, ``Work, entropy production, and thermodynamics
  of information under protocol constraints,'' \emph{Physical Review X}, in
  press (arXiv:2008.10764).

\bibitem{kolchinsky2017maximizing}
A.~Kolchinsky, I.~Marvian, C.~Gokler, Z.-W. Liu, P.~Shor, O.~Shtanko,
  K.~Thompson, D.~Wolpert, and S.~Lloyd, ``Maximizing free energy gain,''
  \emph{arXiv preprint arXiv:1705.00041}, 2017.

\bibitem{kolchinsky2016dependence}
A.~Kolchinsky and D.~H. Wolpert, ``Dependence of dissipation on the initial
  distribution over states,'' \emph{Journal of Statistical Mechanics: Theory
  and Experiment}, p. 083202, 2017.

\bibitem{breuer2002theory}
H.-P. Breuer, F.~Petruccione \emph{et~al.}, \emph{The theory of open quantum
  systems}.\hskip 1em plus 0.5em minus 0.4em\relax Oxford University Press,
  2002.

\bibitem{Note1}
The fact that any equilibrium state can be chosen as the reference state
  follows immediately from our results as stated later in the paper, such as
  Eq.~(\protect \ref {eq:equalityBasis}). Consider any two equilibrium states
  $\pi ,\pi '$ and EP $\Sigma $ defined relative to reference equilibrium state
  $\pi $, as in Eq.~(\protect \ref {eq:relax}). Since $\pi '$ is also an
  equilibrium state, it must (1) be a minimizer of $\Sigma $, (2) achieve
  $\Sigma (\pi ')=0$, and (3) satisfy $\Phi (\pi ')=\pi '$. Then, as long as
  $S(\rho \Vert \pi ')<\infty $, Eq.~(\protect \ref {eq:equalityBasis}) gives
  $\Sigma (\rho )=S(\rho \Vert \pi ')-S(\Phi (\rho )\Vert \pi ')$, which means
  that EP defined relative to reference equilibrium state $\pi $ (LHS) is equal
  to EP defined relative to reference equilibrium state $\pi '$ (RHS).

\bibitem{wolpert2020thermodynamic}
D.~H. Wolpert and A.~Kolchinsky, ``Thermodynamics of computing with circuits,''
  \emph{New Journal of Physics}, 2020.

\bibitem{Note2}
See also \protect \citep {kolchinsky2020thermodynamic} for a derivation of
  Eq.~(\protect \ref {eq:classicalMismatch}) for a classical system with a
  countably infinite state space but deterministic dynamics.

\bibitem{muller2017monotonicity}
A.~M{\"u}ller-Hermes and D.~Reeb, ``Monotonicity of the quantum relative
  entropy under positive maps,'' in \emph{Annales Henri Poincar{\'e}}, vol.~18,
  no.~5.\hskip 1em plus 0.5em minus 0.4em\relax Springer, 2017, pp. 1777--1788.

\bibitem{riechers2020initial}
P.~M. Riechers and M.~Gu, ``Initial-state dependence of thermodynamic
  dissipation for any quantum process,'' \emph{Physical Review E}, vol. 103,
  no.~4, p. 042145, Apr. 2021.

\bibitem{riechersImpossibilityLandauerBound2021}
------, ``Impossibility of achieving landauer's bound for almost every quantum
  state,'' \emph{Physical Review A}, vol. 104, no.~1, p. 012214, 2021.

\bibitem{Note3}
Although Ref.~\protect \citep {riechers2020initial} never explicitly states the
  assumption of a finite-dimensional Hilbert space, it is implicit in the
  derivations of that paper. For example, in infinite dimensional spaces, it
  cannot be assumed that the directional derivative can be written in terms of
  the gradient (as in the derivation of Theorem 1 in \protect \citep
  {riechers2020initial}), that the directional derivative at the optimizer with
  full support vanishes (as in Eq. 10 in \protect \citep
  {riechers2020initial}), or that $S(\rho \Vert \varphi )<\infty $ whenever
  $\mathrm {supp}\,\rho \subseteq \mathrm {supp}\,\varphi $.

\bibitem{campisi2011colloquium}
M.~Campisi, P.~H{\"a}nggi, and P.~Talkner, ``Colloquium: Quantum fluctuation
  relations: Foundations and applications,'' \emph{Reviews of Modern Physics},
  vol.~83, no.~3, p. 771, 2011.

\bibitem{horowitz2013entropy}
J.~M. Horowitz and J.~M. Parrondo, ``Entropy production along nonequilibrium
  quantum jump trajectories,'' \emph{New Journal of Physics}, vol.~15, no.~8,
  p. 085028, 2013.

\bibitem{horowitz2014equivalent}
J.~M. Horowitz and T.~Sagawa, ``Equivalent definitions of the quantum
  nonadiabatic entropy production,'' \emph{Journal of Statistical Physics},
  vol. 156, no.~1, pp. 55--65, 2014.

\bibitem{esposito2010three}
M.~Esposito and C.~Van~den Broeck, ``Three faces of the second law. {I}.
  {M}aster equation formulation,'' \emph{Physical Review E}, vol.~82, no.~1, p.
  011143, 2010.

\bibitem{manzanoQuantumFluctuationTheorems2018}
G.~Manzano, J.~M. Horowitz, and J.~M.~R. Parrondo,
  ``\BIBforeignlanguage{en}{Quantum {{Fluctuation Theorems}} for {{Arbitrary
  Environments}}: {{Adiabatic}} and {{Nonadiabatic Entropy Production}}},''
  \emph{\BIBforeignlanguage{en}{Physical Review X}}, vol.~8, no.~3, Aug. 2018.

\bibitem{faist2019thermodynamic}
P.~Faist, M.~Berta, and F.~Brand{\~a}o, ``Thermodynamic capacity of quantum
  processes,'' \emph{Physical Review Letters}, vol. 122, no.~20, p. 200601,
  2019.

\bibitem{plastino1995fisher}
A.~Plastino and A.~Plastino, ``Fisher information and bounds to the entropy
  increase,'' \emph{Physical Review E}, vol.~52, no.~4, p. 4580, 1995.

\bibitem{holevo2011entropy}
A.~S. Holevo, ``Entropy gain and the {C}hoi-{J}amiolkowski correspondence for
  infinite-dimensional quantum evolutions,'' \emph{Theoretical and Mathematical
  Physics}, vol. 166, no.~1, pp. 123--138, 2011.

\bibitem{holevo2011entropyB}
------, ``The entropy gain of quantum channels,'' in \emph{2011 IEEE
  International Symposium on Information Theory Proceedings}.\hskip 1em plus
  0.5em minus 0.4em\relax IEEE, 2011, pp. 289--292.

\bibitem{kolchinsky2020thermodynamic}
A.~Kolchinsky and D.~H. Wolpert, ``Thermodynamic costs of {T}uring machines,''
  \emph{Physical Review Research}, vol.~2, no.~3, p. 033312, 2020.

\bibitem{maroney2009generalizing}
O.~Maroney, ``Generalizing {L}andauer's principle,'' \emph{Physical Review E},
  vol.~79, no.~3, p. 031105, 2009.

\bibitem{wolpert_arxiv_beyond_bit_erasure_2015}
D.~H. Wolpert, ``Extending {L}andauer's bound from bit erasure to arbitrary
  computation,'' 2015.

\bibitem{turgut_relations_2009}
S.~Turgut, ``Relations between entropies produced in nondeterministic
  thermodynamic processes,'' \emph{Physical Review E}, vol.~79, no.~4, p.
  041102, Apr. 2009.

\bibitem{wildeQuantumInformationTheory2017}
M.~Wilde, \emph{\BIBforeignlanguage{en}{Quantum information theory}}, second
  edition~ed.\hskip 1em plus 0.5em minus 0.4em\relax Cambridge, UK ; New York:
  Cambridge University Press, 2017.

\bibitem{Note4}
The definition in Eq.~(\protect \ref {eq:recov}) holds for finite dimensional
  spaces and $\rho $ such that $\mathrm {supp}\,\rho \subseteq \mathrm
  {supp}\,\varphi $. For a more general definition, see \protect \citep
  {petz1988sufficiency,jungeUniversalRecoveryMaps2018}.

\bibitem{leiferFormulationQuantumTheory2013}
M.~S. Leifer and R.~W. Spekkens, ``Towards a formulation of quantum theory as a
  causally neutral theory of {Bayesian} inference,'' \emph{Physical Review A},
  vol.~88, no.~5, p. 052130, Nov. 2013.

\bibitem{ptaszynskiEntropyProductionOpen2019}
K.~Ptaszy{\'n}ski and M.~Esposito, ``\BIBforeignlanguage{en}{Entropy production
  in open systems: The predominant role of intraenvironment correlations},''
  \emph{\BIBforeignlanguage{en}{Physical Review Letters}}, vol. 123, no.~20,
  Nov. 2019.

\bibitem{baumgratz_quantifying_2014}
T.~Baumgratz, M.~Cramer, and M.~B. Plenio,
  ``\BIBforeignlanguage{en}{Quantifying coherence},''
  \emph{\BIBforeignlanguage{en}{Physical Review Letters}}, vol. 113, no.~14,
  Sep. 2014.

\bibitem{santos_role_2019}
J.~P. Santos, L.~C. C{\'e}leri, G.~T. Landi, and M.~Paternostro, ``The role of
  quantum coherence in non-equilibrium entropy production,'' \emph{npj Quantum
  Information}, vol.~5, no.~1, pp. 1--7, 2019.

\bibitem{francica_role_2019}
G.~Francica, J.~Goold, and F.~Plastina, ``\BIBforeignlanguage{en}{Role of
  coherence in the nonequilibrium thermodynamics of quantum systems},''
  \emph{\BIBforeignlanguage{en}{Physical Review E}}, vol.~99, no.~4, p. 042105,
  Apr. 2019.

\bibitem{francica2020quantum}
G.~Francica, F.~C. Binder, G.~Guarnieri, M.~T. Mitchison, J.~Goold, and
  F.~Plastina, ``Quantum coherence and ergotropy,'' \emph{Physical Review
  Letters}, vol. 125, no.~18, p. 180603, 2020.

\bibitem{shirokovLowerSemicontinuityEntropic2017a}
M.~E. Shirokov and A.~S. Holevo, ``\BIBforeignlanguage{en}{On lower
  semicontinuity of the entropic disturbance and its applications in quantum
  information theory},'' \emph{\BIBforeignlanguage{en}{Izvestiya:
  Mathematics}}, vol.~81, no.~5, pp. 1044--1060, Oct. 2017.

\bibitem{buscemiUnifiedApproachInformationDisturbance2009}
F.~Buscemi and M.~Horodecki, ``\BIBforeignlanguage{en}{Towards a unified
  approach to information-disturbance tradeoffs in quantum measurements},''
  \emph{\BIBforeignlanguage{en}{Open Systems \& Information Dynamics}},
  vol.~16, no.~01, pp. 29--48, Mar. 2009.

\bibitem{buscemiApproximateReversibilityContext2016}
F.~Buscemi, S.~Das, and M.~M. Wilde, ``\BIBforeignlanguage{en}{Approximate
  reversibility in the context of entropy gain, information gain, and complete
  positivity},'' \emph{\BIBforeignlanguage{en}{Physical Review A}}, vol.~93,
  no.~6, Jun. 2016.

\bibitem{Note5}
This follows by writing $\Sigma (\rho )=S(\Phi (\rho ))-S(\rho )+\mathrm
  {tr}\{\Pi A\Pi \rho \}$$=S(\rho \Vert \varphi )+\text {const}$, where
  $\varphi $ is defined as in Eq.~(\protect \ref {eq:optsol}).

\bibitem{ramakrishnan2019non}
N.~Ramakrishnan, R.~Iten, V.~B. Scholz, and M.~Berta, ``Non-commutative
  {B}lahut-{A}rimoto algorithms,'' \emph{arXiv preprint arXiv:1905.01286},
  2019.

\bibitem{Note6}
In general, this decomposition will not be unique: imagine the trivial case
  where, in \ref {eq:EP0-ent}, $\Phi =\mathrm {Id}$ and $Q(\rho )=0$; then,
  $\Sigma (\rho )=0$ for all $\rho $, and any complete basis $\{\vert i\rangle
  \}$ can be used to define a basin decomposition.

\bibitem{janzing_quantum_2006}
D.~Janzing, ``\BIBforeignlanguage{en}{Quantum {Thermodynamics} with {Missing}
  {Reference} {Frames}: {Decompositions} of {Free} {Energy} {Into}
  {Non}-{Increasing} {Components}},'' \emph{\BIBforeignlanguage{en}{Journal of
  Statistical Physics}}, vol. 125, no.~3, pp. 761--776, Nov. 2006.

\bibitem{marvianHowQuantifyCoherence2016}
I.~Marvian and R.~W. Spekkens, ``\BIBforeignlanguage{en}{How to quantify
  coherence: {{Distinguishing}} speakable and unspeakable notions},''
  \emph{\BIBforeignlanguage{en}{Physical Review A}}, vol.~94, no.~5, Nov. 2016.

\bibitem{vaccaro_tradeoff_2008}
J.~A. Vaccaro, F.~Anselmi, H.~M. Wiseman, and K.~Jacobs,
  ``\BIBforeignlanguage{en}{Tradeoff between extractable mechanical work,
  accessible entanglement, and ability to act as a reference system, under
  arbitrary superselection rules},'' \emph{\BIBforeignlanguage{en}{Physical
  Review A}}, vol.~77, no.~3, Mar. 2008.

\bibitem{holevoNoteCovariantDynamical1993}
A.~Holevo, ``\BIBforeignlanguage{en}{A note on covariant dynamical
  semigroups},'' \emph{\BIBforeignlanguage{en}{Reports on Mathematical
  Physics}}, vol.~32, no.~2, pp. 211--216, Apr. 1993.

\bibitem{esposito2006fluctuation}
M.~Esposito and S.~Mukamel, ``Fluctuation theorems for quantum master
  equations,'' \emph{Physical Review E}, vol.~73, no.~4, p. 046129, 2006.

\bibitem{esposito2009nonequilibrium}
M.~Esposito, U.~Harbola, and S.~Mukamel, ``Nonequilibrium fluctuations,
  fluctuation theorems, and counting statistics in quantum systems,''
  \emph{Reviews of modern physics}, vol.~81, no.~4, p. 1665, 2009.

\bibitem{manzanoNonequilibriumPotentialFluctuation2015}
G.~Manzano, J.~M. Horowitz, and J.~M.~R. Parrondo,
  ``\BIBforeignlanguage{en}{Nonequilibrium potential and fluctuation theorems
  for quantum maps},'' \emph{\BIBforeignlanguage{en}{Physical Review E}},
  vol.~92, no.~3, Sep. 2015.

\bibitem{allahverdyanNonequilibriumQuantumFluctuations2014a}
A.~E. Allahverdyan, ``\BIBforeignlanguage{en}{Nonequilibrium quantum
  fluctuations of work},'' \emph{\BIBforeignlanguage{en}{Physical Review E}},
  vol.~90, no.~3, Sep. 2014.

\bibitem{kwonFluctuationTheoremsQuantum2019}
H.~Kwon and M.~S. Kim, ``\BIBforeignlanguage{en}{Fluctuation {{Theorems}} for a
  {{Quantum Channel}}},'' \emph{\BIBforeignlanguage{en}{Physical Review X}},
  vol.~9, no.~3, Aug. 2019.

\bibitem{micadeiQuantumFluctuationTheorems2020}
K.~Micadei, G.~T. Landi, and E.~Lutz, ``\BIBforeignlanguage{en}{Quantum
  {{Fluctuation Theorems}} beyond {{Two}}-{{Point Measurements}}},''
  \emph{\BIBforeignlanguage{en}{Physical Review Letters}}, vol. 124, no.~9,
  Mar. 2020.

\bibitem{funo2015quantum}
K.~Funo, Y.~Murashita, and M.~Ueda, ``Quantum nonequilibrium equalities with
  absolute irreversibility,'' \emph{New Journal of Physics}, vol.~17, no.~7, p.
  075005, 2015.

\bibitem{buscemi2020fluctuation}
F.~Buscemi and V.~Scarani, ``Fluctuation theorems from {B}ayesian
  retrodiction,'' \emph{Physical Review E}, vol. 103, no.~5, p. 052111, 2021.

\bibitem{allahverdyanExcludingJointProbabilities2018}
A.~E. Allahverdyan and A.~Danageozian, ``Excluding joint probabilities from
  quantum theory,'' \emph{Physical Review A}, vol.~97, no.~3, p. 030102, 2018.

\bibitem{spohn1978irreversible}
H.~Spohn and J.~L. Lebowitz, ``Irreversible thermodynamics for quantum systems
  weakly coupled to thermal reservoirs,'' \emph{Adv. Chem. Phys}, vol.~38, pp.
  109--142, 1978.

\bibitem{spohn_entropy_1978}
H.~Spohn, ``Entropy production for quantum dynamical semigroups,''
  \emph{Journal of Mathematical Physics}, vol.~19, no.~5, pp. 1227--1230, 1978.

\bibitem{alicki_quantum_1979}
R.~Alicki, ``The quantum open system as a model of the heat engine,''
  \emph{Journal of Physics A: Mathematical and General}, vol.~12, no.~5, p.
  L103, 1979.

\bibitem{baumgartnerAnalysisQuantumSemigroups2008a}
B.~Baumgartner, H.~Narnhofer, and W.~Thirring,
  ``\BIBforeignlanguage{en}{Analysis of quantum semigroups with
  {{GKS}}\textendash{{Lindblad}} generators: {{I}}. {{Simple}} generators},''
  \emph{\BIBforeignlanguage{en}{Journal of Physics A: Mathematical and
  Theoretical}}, vol.~41, no.~6, p. 065201, Feb. 2008.

\bibitem{baumgartnerAnalysisQuantumSemigroups2008}
B.~Baumgartner and H.~Narnhofer, ``Analysis of quantum semigroups with
  {{GKS}}--{{Lindblad}} generators {{II}}. {{General}},'' \emph{Journal of
  Physics A: Mathematical and Theoretical}, vol.~41, no.~39, p. 395303, Oct.
  2008.

\bibitem{carboneIrreducibleDecompositionsStationary2016}
R.~Carbone and Y.~Pautrat, ``\BIBforeignlanguage{en}{Irreducible decompositions
  and stationary states of quantum channels},''
  \emph{\BIBforeignlanguage{en}{Reports on Mathematical Physics}}, vol.~77,
  no.~3, pp. 293--313, Jun. 2016.

\bibitem{talknerRateDescriptionFokkerPlanck2004}
P.~Talkner and J.~{\L}uczka, ``\BIBforeignlanguage{en}{Rate description of
  {{Fokker}}-{{Planck}} processes with time-dependent parameters},''
  \emph{\BIBforeignlanguage{en}{Physical Review E}}, vol.~69, no.~4, Apr. 2004.

\bibitem{ford_entropy_2012}
I.~J. Ford and R.~E. Spinney, ``Entropy production from stochastic dynamics in
  discrete full phase space,'' \emph{Physical Review E}, vol.~86, no.~2, p.
  021127, 2012.

\bibitem{spinneyNonequilibriumThermodynamicsStochastic2012}
R.~E. Spinney and I.~J. Ford, ``Nonequilibrium {{Thermodynamics}} of
  {{Stochastic Systems}} with {{Odd}} and {{Even Variables}},'' \emph{Physical
  Review Letters}, vol. 108, no.~17, p. 170603, Apr. 2012.

\bibitem{spinneyEntropyProductionFull2012}
------, ``Entropy production in full phase space for continuous stochastic
  dynamics,'' \emph{Physical Review E}, vol.~85, no.~5, p. 051113, 2012.

\bibitem{esposito_three_2010}
M.~Esposito and C.~V.~d. Broeck, ``Three detailed fluctuation theorems,''
  \emph{Physical Review Letters}, vol. 104, no.~9, Mar. 2010.

\bibitem{parrondoEntropyProductionArrow2009a}
J.~M.~R. Parrondo, C.~V. den Broeck, and R.~Kawai,
  ``\BIBforeignlanguage{en}{Entropy production and the arrow of time},''
  \emph{\BIBforeignlanguage{en}{New Journal of Physics}}, vol.~11, no.~7, p.
  073008, Jul. 2009.

\bibitem{millerEntropyProductionTime2017}
H.~J.~D. Miller and J.~Anders, ``\BIBforeignlanguage{en}{Entropy production and
  time asymmetry in the presence of strong interactions},''
  \emph{\BIBforeignlanguage{en}{Physical Review E}}, vol.~95, no.~6, p. 062123,
  Jun. 2017.

\bibitem{strasberg_stochastic_2017}
P.~Strasberg and M.~Esposito, ``Stochastic thermodynamics in the strong
  coupling regime: {An} unambiguous approach based on coarse graining,''
  \emph{Physical Review E}, vol.~95, no.~6, p. 062101, Jun. 2017.

\bibitem{seifertFirstSecondLaw2016a}
U.~Seifert, ``\BIBforeignlanguage{en}{First and {{Second Law}} of
  {{Thermodynamics}} at {{Strong Coupling}}},''
  \emph{\BIBforeignlanguage{en}{Physical Review Letters}}, vol. 116, no.~2, p.
  020601, Jan. 2016.

\bibitem{bennettNotesHistoryReversible1988}
C.~H. Bennett, ``Notes on the history of reversible computation,'' \emph{ibm
  Journal of Research and Development}, vol.~32, no.~1, pp. 16--23, 1988.

\bibitem{maroney_absence_2005}
O.~J.~E. Maroney, ``The (absence of a) relationship between thermodynamic and
  logical reversibility,'' \emph{Studies in History and Philosophy of Science
  Part B: Studies in History and Philosophy of Modern Physics}, vol.~36, no.~2,
  pp. 355--374, Jun. 2005.

\bibitem{sagawa2014thermodynamic}
T.~Sagawa, ``Thermodynamic and logical reversibilities revisited,''
  \emph{Journal of Statistical Mechanics: Theory and Experiment}, vol. 2014,
  no.~3, p. P03025, 2014.

\bibitem{wolpert2019stochastic}
D.~H. Wolpert, ``The stochastic thermodynamics of computation,'' \emph{Journal
  of Physics A: Mathematical and Theoretical}, vol.~52, no.~19, p. 193001,
  2019.

\bibitem{petz1988sufficiency}
D.~Petz, ``Sufficiency of channels over von {N}eumann algebras,'' \emph{The
  Quarterly Journal of Mathematics}, vol.~39, no.~1, pp. 97--108, 1988.

\bibitem{mosonyi2004structure}
M.~Mosonyi and D.~Petz, ``Structure of sufficient quantum coarse-grainings,''
  \emph{Letters in Mathematical Physics}, vol.~68, no.~1, pp. 19--30, 2004.

\bibitem{wildeRecoverabilityQuantumInformation2015a}
M.~M. Wilde, ``Recoverability in quantum information theory,''
  \emph{Proceedings of the Royal Society A: Mathematical, Physical and
  Engineering Sciences}, vol. 471, no. 2182, p. 20150338, 2015.

\bibitem{jungeUniversalRecoveryMaps2018}
M.~Junge, R.~Renner, D.~Sutter, M.~M. Wilde, and A.~Winter, ``Universal
  recovery maps and approximate sufficiency of quantum relative entropy,'' in
  \emph{Annales {{Henri Poincar\'e}}}, vol.~19.\hskip 1em plus 0.5em minus
  0.4em\relax {Springer}, 2018, pp. 2955--2978.

\bibitem{alhambraWorkReversibilityQuantum2018a}
{\'A}.~M. Alhambra, S.~Wehner, M.~M. Wilde, and M.~P. Woods,
  ``\BIBforeignlanguage{en}{Work and reversibility in quantum
  thermodynamics},'' \emph{\BIBforeignlanguage{en}{Physical Review A}},
  vol.~97, no.~6, Jun. 2018.

\bibitem{gourHowQuantifyDynamical2019}
G.~Gour and A.~Winter, ``\BIBforeignlanguage{en}{How to {Quantify} a
  {Dynamical} {Quantum} {Resource}},'' \emph{\BIBforeignlanguage{en}{Physical
  Review Letters}}, vol. 123, no.~15, p. 150401, Oct. 2019.

\bibitem{gourEntropyQuantumChannel2020}
G.~Gour and M.~M. Wilde, ``\BIBforeignlanguage{en}{Entropy of a {Quantum}
  {Channel}: {Definition}, {Properties}, and {Application}},'' in
  \emph{\BIBforeignlanguage{en}{2020 {IEEE} {International} {Symposium} on
  {Information} {Theory} ({ISIT})}}.\hskip 1em plus 0.5em minus 0.4em\relax Los
  Angeles, CA, USA: IEEE, Jun. 2020, pp. 1903--1908.

\bibitem{choi_equivalence_1994}
M.-D. Choi, M.~B. Ruskai, and E.~Seneta, ``Equivalence of certain entropy
  contraction coefficients,'' \emph{Linear Algebra and its Applications}, vol.
  208-209, pp. 29--36, Sep. 1994.

\bibitem{raginsky2002strictly}
M.~Raginsky, ``Strictly contractive quantum channels and physically realizable
  quantum computers,'' \emph{Physical Review A}, vol.~65, no.~3, p. 032306,
  2002.

\bibitem{hiai2016contraction}
F.~Hiai and M.~B. Ruskai, ``Contraction coefficients for noisy quantum
  channels,'' \emph{Journal of Mathematical Physics}, vol.~57, no.~1, p.
  015211, 2016.

\bibitem{bennett1973logical}
C.~H. Bennett, ``Logical reversibility of computation,'' \emph{IBM journal of
  Research and Development}, vol.~17, no.~6, pp. 525--532, 1973.

\bibitem{fredkin1982conservative}
E.~Fredkin and T.~Toffoli, ``Conservative logic,'' \emph{International Journal
  of Theoretical Physics}, vol.~21, no.~3, pp. 219--253, 1982.

\bibitem{bennett1989time}
C.~H. Bennett, ``Time/space trade-offs for reversible computation,'' \emph{SIAM
  Journal on Computing}, vol.~18, no.~4, pp. 766--776, 1989.

\bibitem{levine1990note}
R.~Y. Levine and A.~T. Sherman, ``A note on bennett's time-space tradeoff for
  reversible computation,'' \emph{SIAM Journal on Computing}, vol.~19, no.~4,
  pp. 673--677, 1990.

\bibitem{lange2000reversible}
K.-J. Lange, P.~McKenzie, and A.~Tapp, ``Reversible space equals deterministic
  space,'' \emph{Journal of Computer and System Sciences}, vol.~60, no.~2, pp.
  354--367, 2000.

\bibitem{Note7}
Since relative entropy is convex in both arguments, the inner maximization is
  satisfied by some pure state. Then, \ref {eq:gf2} can be rewritten as $\inf
  _{\varphi \in \mathcal {D}}\sup _{\vert i\rangle \langle i\vert \in \mathcal
  {D}}\langle i\vert -\ln \varphi \vert i\rangle $. Using the minimax
  principle, the inner maximization is satisfied by the largest eigenvalue of
  $-\ln \varphi $, which for a density matrix has to be no less than $-\ln
  \frac {1}{d}=\ln d=\ln 2$.

\bibitem{garcia2012nonadiabatic}
R.~Garc{\'\i}a-Garc{\'\i}a, ``Nonadiabatic entropy production for non-markov
  dynamics,'' \emph{Physical Review E}, vol.~86, no.~3, p. 031117, 2012.

\bibitem{navascues2015nonthermal}
M.~Navascu{\'e}s and L.~P. Garc{\'\i}a-Pintos, ``Nonthermal quantum channels as
  a thermodynamical resource,'' \emph{Physical Review Letters}, vol. 115,
  no.~1, p. 010405, 2015.

\bibitem{muller2018correlating}
M.~P. M{\"u}ller, ``Correlating thermal machines and the second law at the
  nanoscale,'' \emph{Physical Review X}, vol.~8, no.~4, p. 041051, 2018.

\bibitem{schlogl1985thermodynamic}
F.~Schl{\"o}gl, ``Thermodynamic metric and stochastic measures,''
  \emph{Zeitschrift f{\"u}r Physik B Condensed Matter}, vol.~59, no.~4, pp.
  449--454, 1985.

\bibitem{esposito2011second}
M.~Esposito and C.~Van~den Broeck, ``Second law and {L}andauer principle far
  from equilibrium,'' \emph{EPL (Europhysics Letters)}, vol.~95, no.~4, p.
  40004, 2011.

\bibitem{das2018fundamental}
S.~Das, S.~Khatri, G.~Siopsis, and M.~M. Wilde, ``Fundamental limits on quantum
  dynamics based on entropy change,'' \emph{Journal of Mathematical Physics},
  vol.~59, no.~1, p. 012205, 2018.

\bibitem{alicki2004isotropic}
R.~Alicki, ``Isotropic quantum spin channels and additivity questions,''
  \emph{arXiv preprint quant-ph/0402080}, 2004.

\bibitem{wehrlGeneralPropertiesEntropy1978}
A.~Wehrl, ``\BIBforeignlanguage{en}{General properties of entropy},''
  \emph{\BIBforeignlanguage{en}{Reviews of Modern Physics}}, vol.~50, no.~2,
  pp. 221--260, Apr. 1978.

\bibitem{robertsConvexFunctions1973}
A.~W. Roberts and D.~E. Varberg, \emph{\BIBforeignlanguage{en}{Convex
  functions}}, ser. Pure and applied mathematics; a series of monographs and
  textbooks.\hskip 1em plus 0.5em minus 0.4em\relax New York: Academic Press,
  1973, no.~57.

\bibitem{audenaertQuantumSkewDivergence2014}
K.~M.~R. Audenaert, ``\BIBforeignlanguage{en}{Quantum skew divergence},''
  \emph{\BIBforeignlanguage{en}{Journal of Mathematical Physics}}, vol.~55,
  no.~11, p. 112202, Nov. 2014.

\bibitem{lindbladExpectationsEntropyInequalities1974}
G.~Lindblad, ``\BIBforeignlanguage{en}{Expectations and entropy inequalities
  for finite quantum systems},'' \emph{\BIBforeignlanguage{en}{Communications
  in Mathematical Physics}}, vol.~39, no.~2, pp. 111--119, Jun. 1974.

\bibitem{rockafellarConvexAnalysis1970}
R.~T. Rockafellar, \emph{Convex Analysis}.\hskip 1em plus 0.5em minus
  0.4em\relax {Princeton University Press}, 1970.

\bibitem{audenaert2011telescopic}
K.~M. Audenaert, ``Telescopic relative entropy,'' in \emph{Conference on
  Quantum Computation, Communication, and Cryptography}.\hskip 1em plus 0.5em
  minus 0.4em\relax Springer, 2011, pp. 39--52.

\bibitem{donaldFurtherResultsRelative1987}
M.~J. Donald, ``\BIBforeignlanguage{en}{Further results on the relative
  entropy},'' \emph{\BIBforeignlanguage{en}{Mathematical Proceedings of the
  Cambridge Philosophical Society}}, vol. 101, no.~2, pp. 363--373, Mar. 1987.

\bibitem{petzQuantumInformationTheory2008}
D.~Petz, \emph{\BIBforeignlanguage{en}{Quantum Information Theory and Quantum
  Statistics}}, ser. Theoretical and Mathematical Physics.\hskip 1em plus 0.5em
  minus 0.4em\relax {Berlin}: {Springer}, 2008.

\bibitem{bertaSmoothEntropyFormalism2016a}
M.~Berta, F.~Furrer, and V.~B. Scholz, ``\BIBforeignlanguage{en}{The smooth
  entropy formalism for von {Neumann} algebras},''
  \emph{\BIBforeignlanguage{en}{Journal of Mathematical Physics}}, vol.~57,
  no.~1, p. 015213, Jan. 2016.

\bibitem{tomamichelQuantumInformationProcessing2016}
M.~Tomamichel, \emph{\BIBforeignlanguage{en}{Quantum {{Information Processing}}
  with {{Finite Resources}}}}, ser. {{SpringerBriefs}} in {{Mathematical
  Physics}}.\hskip 1em plus 0.5em minus 0.4em\relax {Cham}: {Springer
  International Publishing}, 2016, vol.~5.

\bibitem{Seifert2005}
U.~Seifert, ``Entropy production along a stochastic trajectory and an integral
  fluctuation theorem,'' \emph{Physical Review Letters}, vol.~95, no.~4, p.
  040602, 2005.

\bibitem{van2010three}
C.~Van~den Broeck and M.~Esposito, ``Three faces of the second law. {II}.
  {F}okker-{P}lanck formulation,'' \emph{Physical Review E}, vol.~82, no.~1, p.
  011144, 2010.

\bibitem{polyanskiyLectureNotesInformation2019}
Y.~Polyanskiy and Y.~Wu, \emph{Lecture Notes on {{Information Theory}}}.\hskip
  1em plus 0.5em minus 0.4em\relax {MIT (6.441), UIUC (ECE 563)}, 2019.

\bibitem{cover_elements_2006}
T.~M. Cover and J.~A. Thomas, \emph{Elements of information theory}.\hskip 1em
  plus 0.5em minus 0.4em\relax John Wiley \& Sons, 2006.

\bibitem{posnerRandomCodingStrategies1975}
E.~Posner, ``\BIBforeignlanguage{en}{Random coding strategies for minimum
  entropy},'' \emph{\BIBforeignlanguage{en}{IEEE Transactions on Information
  Theory}}, vol.~21, no.~4, pp. 388--391, Jul. 1975.

\bibitem{topsoe1979information}
F.~Tops{\o}e, ``Information-theoretical optimization techniques,''
  \emph{Kybernetika}, vol.~15, no.~1, pp. 8--27, 1979.

\bibitem{pinskerInformationInformationStability1964}
M.~S. Pinsker, \emph{\BIBforeignlanguage{en}{Information and {{Information
  Stability}} of {{Random Variables}} and {{Processes}}}}.\hskip 1em plus 0.5em
  minus 0.4em\relax {Holden-Day}, 1964.

\bibitem{harremoesInformationTopologiesApplications2007}
P.~Harremo{\"e}s, ``\BIBforeignlanguage{en}{Information {{Topologies}} with
  {{Applications}}},'' in \emph{\BIBforeignlanguage{en}{Entropy, {{Search}},
  {{Complexity}}}}, I.~Csisz{\'a}r, G.~O.~H. Katona, G.~Tardos, and G.~Wiener,
  Eds.\hskip 1em plus 0.5em minus 0.4em\relax {Berlin, Heidelberg}: {Springer
  Berlin Heidelberg}, 2007, vol.~16, pp. 113--150.

\end{thebibliography}

\clearpage{}

\appendix

\section{Mismatch Cost for Integrated EP}

\newcommand{\klconvexref}{I}
\newcommand{\kllimAaRef}{II}
\newcommand{\kllimBaRef}{III}
\newcommand{\klupperboundpropref}{IV}
\newcommand{\klmonopropref}{V}
\newcommand{\kllimBref}{VII}
\newcommand{\klconvexprop}{\klconvexref. }
\newcommand{\klupperboundprop}{\klupperboundpropref. } \newcommand{\klmonoprop}{\klmonopropref. }
\newcommand{\kllimAaProp}{\kllimAaRef. }
\newcommand{\kllimBaProp}{\kllimBaRef. }

\label{app:proofs-integrated}

\subsection{Preliminaries}

In this appendix, we formally derive our results for mismatch cost
for integrated EP. We first introduce some notation.

We write $\mathcal{H}_{X'}$ and $\mathcal{H}'_{X'}$ --- or sometimes
simply $\mathcal{H}$ and $\mathcal{H}'$ --- to indicate two separable
Hilbert spaces (below, these will indicate the ``input'' and ``output''
spaces of a quantum channel $\Phi$), and write $\mathcal{T}$ and
$\mathcal{T}'$ to indicate the set of trace-class operators over
$\mathcal{H}_{X}$ and $\mathcal{H}'_{X'}$ respectively. We write
$\mathcal{D}\subseteq\mathcal{T}$ to indicate the set of density
operators over $\mathcal{H}$.

For any functional $f:\mathcal{D}\to\mathbb{R}\cup\{\infty\}$,
(semi)continuity is meant in the sense of the trace norm.
The \emph{support} of a density operator is the orthogonal complement of its kernel; we use $\mathrm{supp}\,\rho$ to indicate the support of $\rho\in\mathcal{D}$. For any pair of self-adjoint operators $\rho$ and $\omega$, we use the standard notation like $\rho \ge \omega$ to indicate that $\rho - \omega$ is positive.

We indicate a linear mixture of two states $\rho,\varphi\in\mathcal{D}$
with coefficient $\lambda\in\mathbb{R}$ as
\begin{align}
\varphi(\lambda) & :=(1-\lambda)\varphi+\lambda\rho.\label{eq:qCdef}
\end{align}

We will derive our results for a general family of ``EP-type'' functions.
An EP-type function, which we write generically as $\Sigma:\mathcal{D}\to\mathbb{R}\cup\{\infty\}$,
can take one of two mathematical forms. The first form is
\begin{equation}
\Sigma(\rho)=S(\Phi(\rho))-S(\rho)+Q(\rho),\label{eq:functype1}
\end{equation}
where $\Phi$ is a positive and trace-preserving map and $Q:\mathcal{D}\to\mathbb{R}\cup\{\infty\}$
is a lower semicontinuous linear functional. This form appears in
the main text as \ref{eq:EP0-ent}.

To introduce the second form, consider some system of interest coupled
to an environment $Y$, and let the separable Hilbert spaces $\mathcal{H}_{Y}$
and $\mathcal{H}_{Y'}'$ represent the environment at the beginning
and end of the protocol. Assume the system dynamically evolves according
to the completely-positive and trace-preserving (CPTP) map $\Phi:\mathcal{T}\to\mathcal{T}'$
with the representation $\Phi(\rho)=\mathrm{tr}_{Y'}\{V(\rho\otimes\omega)V^{\dagger}\}$
for some isometry $V:\mathcal{H}_{X}\otimes\mathcal{H}_{Y}\to\mathcal{H}'_{X'}\otimes\mathcal{H}_{Y'}'$
and fixed density operator $\omega$ over $\mathcal{H}_{Y}$. Then,
the second form of EP is given by
\begin{equation}
\Sigma(\rho)=S(V(\rho\otimes\omega)V^{\dagger}\Vert\Phi(\rho)\otimes\omega)+Q'(\rho),\label{eq:functype2}
\end{equation}
where $Q':\mathcal{D}\to\mathbb{R}\cup\{\infty\}$ is any lower-semicontinuous
linear functional. A special case of \ref{eq:functype2} appeared
in the main text as \ref{eq:functype2-maintext2} (where we took $V$
to be some unitary $U$ over system-and-environment and took $Q'=0$).

We draw attention to several important aspects of our definitions
of EP-type functions.

\begin{enumerate}[wide,labelindent=0pt,labelwidth=!]
\item Under both definitions \ref{eq:functype1} and \ref{eq:functype2},
the input and output spaces of the quantum channel $\Phi$
may be different.
\item For both definitions, we assume that $\Sigma(\rho)>-\infty$ for all
$\rho$ (so that a minimizer exists).
\item Unlike \ref{eq:functype2}, the definition in \ref{eq:functype1}
does not require that $\Phi$ be completely positive, but only positive.
\item The assumption of lower-semicontinuity of $Q$ in \ref{eq:functype1},
or of $Q'$ in \ref{eq:functype2} is only used in \ref{thm:eqresult}.
For many other results, it can be omitted.
\item For EP-type functions as in \ref{eq:functype1}, in infinite dimensions
there are states $\rho\in\mathcal{D}$ with infinite entropy, $S(\rho)=\infty$,
in which case \ref{eq:functype1} is not well-defined. To make $\Sigma$
well-defined for all $\rho\in\mathcal{D}$, we assume that $\Sigma(\rho)=\infty$
whenever $S(\rho)=\infty$. However, \ref{eq:functype2} is better suited for analyzing EP incurred by states with infinite entropy, $S(\rho)=\infty$,
since it can finite in such cases (unlike \ref{eq:functype1}). (Note,
however, that states with infinite entropy are sometimes argued to
be ``unphysical'' \citep{wehrlGeneralPropertiesEntropy1978}).
\item Many of our results reference \ref{prop:epConvMix}, \ref{prop:KLprop},
and \ref{prop:epLSC} below (along with some other useful lemmas),
which prove general properties of quantum relative entropy and EP-type
functions.
\end{enumerate}
As mentioned in the main text, for states with finite entropy, \ref{eq:functype2}
can always be re-written in the form of \ref{eq:functype1}, and vice
versa. This is proved in the following result.
\begin{prop}
\label{prop:formequiv}Given an isometry $V:\mathcal{H}_{X}\otimes\mathcal{H}_{Y}\to\mathcal{H}'_{X'}\otimes\mathcal{H}_{Y'}'$,
a CPTP map $\Phi(\rho)=\mathrm{tr}_{Y'}\{V(\rho\otimes\omega)V^{\dagger}\}$,
and any $\rho$ such that $S(\rho)<\infty$,
\begin{align}
& S(V(\rho\otimes\omega)V^{\dagger}\Vert\Phi(\rho)\otimes\omega)+Q'(\rho)\label{eq:form1}\\
& \qquad=S(\Phi(\rho))-S(\rho)+Q(\rho),\label{eq:form2}
\end{align}
where
\begin{equation}
Q(\rho):=Q'(\rho)-\mathrm{tr}\{\mathrm{tr}_{X'}\{V(\rho\otimes\omega)V^{\dagger}\}\ln\omega)\}-S(\omega).\label{eq:formEqF}
\end{equation}
\end{prop}

\begin{proof}
Expand the RHS of \ref{eq:form1} as
\begin{align}
& S(V(\rho\otimes\omega)V^{\dagger}\Vert\Phi(\rho)\otimes\omega)+Q'(\rho)=\label{eq:hg2}\\
& Q'(\rho)-\mathrm{tr}\{(V(\rho\otimes\omega)V^{\dagger})\ln(\Phi(\rho)\otimes\omega)\}-S(V(\rho\otimes\omega)V^{\dagger}).\nonumber
\end{align}
One can rewrite the second term on the RHS of \cref{eq:hg2} as
\begin{align}
& \mathrm{tr}\{(V(\rho\otimes\omega)V^{\dagger})\ln(\Phi(\rho)\otimes\omega)\}\nonumber \\
& =\mathrm{tr}\{\mathrm{tr}_{Y'}\{(V(\rho\otimes\omega)V^{\dagger})\}\ln\Phi(\rho)\}\nonumber \\
& \qquad\qquad+\mathrm{tr}\{\mathrm{tr}_{X'}\{V(\rho\otimes\omega)V^{\dagger}\}\ln\omega)\}\nonumber \\
& =\mathrm{tr}\{\Phi(\rho)\ln\Phi(\rho)\}+\mathrm{tr}\{\mathrm{tr}_{X'}\{V(\rho\otimes\omega)V^{\dagger}\}\ln\omega)\}\nonumber \\
& =-S(\Phi(\rho))+\mathrm{tr}\{\mathrm{tr}_{X'}\{V(\rho\otimes\omega)V^{\dagger}\}\ln\omega)\}.\label{eq:hg3}
\end{align}
One can rewrite the third term on the RHS of \cref{eq:hg2} as
\begin{equation}
S(V(\rho\otimes\omega)V^{\dagger})=S(\rho\otimes\omega)=S(\rho)+S(\omega),\label{eq:hg4}
\end{equation}
where we've used that entropy is invariant under isometries and additive
for product states. Plugging \ref{eq:hg3,eq:hg4} into \ref{eq:hg2},
and then using \ref{eq:formEqF}, gives \ref{eq:form2}.
\end{proof}

\subsection{Main proofs}

Our first result shows that the directional derivative of $\Sigma$,
defined as in \ref{eq:functype1} or \ref{eq:functype2}, has a simple
information-theoretic form. This result appears as \ref{eq:dd0} in
the main text.
\begin{prop}
\label{thm:genf}For any $\rho,\varphi\in\mathcal{D}$ such that $\left|\Sigma(\rho)\right|<\infty,\left|\Sigma(\varphi)\right|<\infty,S(\rho\Vert\varphi)<\infty$,
\begin{align}
{\textstyle {\textstyle \partial_{\lambda}^{+}}}\Sigma(\varphi(\lambda))\vert_{\lambda=0} & :=\lim_{{\lambda\to0^{+}}}\frac{\Sigma(\varphi(\lambda))-\Sigma(\varphi)}{\lambda}\nonumber \\
& =\Sigma(\rho)-\Sigma(\varphi)+\Delta S({\rho}\Vert{\varphi}).\label{eq:ddeq-1}
\end{align}
\end{prop}

\begin{proof}
First, rearrange \ref{eq:conv1} in \ref{prop:epConvMix} and take
the $\lambda\to0^{+}$ limit to give
\begin{multline}
{\textstyle {\textstyle \partial_{\lambda}^{+}}}\Sigma(\varphi(\lambda))\vert_{\lambda=0}=\Sigma(\rho)-\Sigma(\varphi)+\\
\lim_{{\lambda\to0^{+}}}\big[\Delta S({\rho}\Vert{\varphi(\lambda)})+\frac{1-\lambda}{\lambda}\Delta S({\varphi}\Vert{\varphi(\lambda)})\big].\label{eq:dm}
\end{multline}
We now separately evaluate limits of the two terms inside the brackets
in \ref{eq:dm}. Before proceeding, note that $S(\rho\Vert\varphi)<\infty$
implies $S(\Phi(\rho)\Vert\Phi(\varphi))<\infty$ by \ref{prop:KLprop}(\klmonopropref).
Then,
\begin{align*}
\Delta S({\rho}\Vert{\varphi}) & :=S(\Phi(\rho)\Vert\Phi(\varphi))-S(\rho\Vert\varphi)\\
& =\lim_{{\lambda\to0^{+}}}S(\Phi(\rho)\Vert\Phi(\varphi(\lambda)))-\lim_{{\lambda\to0^{+}}}S(\rho\Vert\varphi(\lambda))\\
& =\lim_{{\lambda\to0^{+}}}[S(\Phi(\rho)\Vert\Phi(\varphi(\lambda)))-S(\rho \Vert \varphi(\lambda))]\\
& =\lim_{{\lambda\to0^{+}}}\Delta S({\rho}\Vert{\varphi(\lambda)}),
\end{align*}
where we first used \ref{eq:ddDef} and then applied
\ref{prop:KLprop}(\kllimAaRef) twice. Then,
\begin{align*}
& \lim_{{\lambda\to0^{+}}}\frac{1-\lambda}{\lambda}\Delta S({\varphi}\Vert{\varphi(\lambda)})\\
& =\lim_{{\lambda\to0^{+}}}\frac{1-\lambda}{\lambda}S(\Phi(\varphi)\Vert\Phi(\varphi(\lambda)))-\lim_{{\lambda\to0^{+}}}\frac{1-\lambda}{\lambda}S(\varphi\Vert\varphi(\lambda))\\
& =0,
\end{align*}
where we applied \ref{prop:KLprop}(\kllimBaRef) twice. Plugging
into \ref{eq:dm} gives \ref{eq:ddeq-1}.
\end{proof}
Next, we derive general bounds on the mismatch cost of $\rho$, relative
to the optimal state within some convex set of states. \ref{eq:convexBound,eq:dn1}
appear in the main text as \ref{eq:ineqConvex}.
\begin{prop}
\label{thm:ineqs}Given a convex set of states $\mathcal{S}\subseteq\mathcal{D}$,
for any $\varphi\in\mathop{\arg\min}_{\omega\in\mathcal{S}}\Sigma(\omega)$
and $\rho\in\mathcal{S}$ with $S(\rho\Vert\varphi)<\infty$,
\begin{align}
\Sigma(\rho)-\Sigma(\varphi) & \ge-\Delta S({\rho}\Vert{\varphi}).\label{eq:convexBound}
\end{align}
Furthermore, if $(1-\lambda)\varphi+\lambda\rho\in\mathcal{S}$ for
some $\lambda<0$,
\begin{align}
\Sigma(\rho)-\Sigma(\varphi) & =-\Delta S({\rho}\Vert{\varphi}).\label{eq:dn1}
\end{align}
\end{prop}

\begin{proof}
Since $\varphi$ is a minimizer, $\Sigma(\varphi)<\infty$ and $\Sigma(\omega)>-\infty$
for all $\omega\in\mathcal{S}$. Then, \ref{eq:convexBound} is trivially
true if $\Sigma(\rho)=\infty$. If $\Sigma(\rho)<\infty$, then the
directional derivative from the minimizer $\varphi$ to $\rho$ can
be expressed as \ref{eq:ddeq-1}. At the same time, the directional
derivative from the minimizer $\varphi$ to any $\rho\in\mathcal{S}$
must be non-negative, since otherwise one could achieve a smaller
value of $\Sigma$ by moving slightly from $\varphi$ toward $\rho$.
Thus, ${\textstyle {\textstyle \partial_{\lambda}^{+}}}\Sigma(\varphi(\lambda))\vert_{\lambda=0}\ge0$,
which gives \ref{eq:convexBound} when combined with \ref{eq:ddeq-1}.

We now prove \ref{eq:dn1}. Let $\omega:=(1-\alpha)\varphi+\alpha\rho\in\mathcal{S}$
for some $\alpha<0$ (which exists by assumption), and note that $\varphi$
can be written as the convex mixture $\varphi=(1-\lambda^{*})\omega+\lambda^{*}\rho$
with $\lambda^{*}=-\alpha/(1-\alpha)$. Note that for any pair of
states $\rho,\omega\in\mathcal{D}$ and $\lambda\in[0,1]$,
\begin{equation}
0\le-(1-\lambda)\Delta S({\rho}\Vert{\omega(\lambda)})-\lambda\Delta S({\rho}\Vert{\omega(\lambda)})\le h_{2}(\lambda),\label{eq:bnd42}
\end{equation}
where $\omega(\lambda)=(1-\lambda)\omega+\lambda\rho$ and $h_{2}(\lambda)=-\lambda\ln\lambda-(1-\lambda)\ln(1-\lambda)$
is the binary entropy function. The lower bound in \ref{eq:bnd42}
follows from the monotonicity of relative entropy, \ref{prop:KLprop}(\klmonopropref). The upper bound follows from $-\Delta S({\rho}\Vert{\omega(\lambda)})\le S(\rho\Vert\omega(\lambda))\le-\ln\lambda$,
\ref{prop:KLprop}(\klupperboundpropref), and similarly for $\Delta S({\omega}\Vert{\omega(\lambda)})$.
\ref{prop:epConvMix} then implies that for all $\lambda\in[0,1]$,
\begin{equation}
0\le(1-\lambda)\Sigma(\omega)+\lambda\Sigma(\rho)-\Sigma(\omega(\lambda))\le h_{2}(\lambda).\label{eq:prfineq}
\end{equation}
Since $h_{2}(\lambda)<\ln2$, \ref{eq:prfineq} implies that
\[
(1-\lambda^{*})\Sigma(\omega)+\lambda^{*}\Sigma(\rho)\le\Sigma(\omega(\lambda^{*}))+\ln2,
\]
thus $\Sigma(\rho),\Sigma(\omega)<\infty$. The lower bound in \ref{eq:prfineq}
also implies that $\Sigma$ is convex, so therefore $\Sigma(\omega(\lambda))<\infty$
for all $\lambda\in[0,1]$. In addition, $S(\rho\Vert\omega(\lambda))\le-\ln\lambda<\infty$
for all $\lambda\in(0,1)$ by \ref{prop:KLprop}(\klupperboundpropref),
hence $S(\Phi(\rho)\Vert\Phi(\omega(\lambda)))<\infty$ by monotonicity.

We now write the directional derivative of $\Sigma$ at $\omega(\lambda)$
toward $\rho$ as a function of $\lambda$,
\begin{align}
& f(\lambda):=\partial_{\eta}^{+}\Sigma((1-\eta)\omega(\lambda)+\eta\rho)=\nonumber \\
& \Sigma(\rho)-\Sigma(\omega(\lambda))+S(\Phi(\rho)\Vert\Phi(\omega(\lambda)))-S(\rho\Vert\omega(\lambda)),\label{eq:j1}
\end{align}
where in the second line we used \ref{thm:genf}. Since $\omega(\lambda^{*})=\varphi$,
by \ref{eq:convexBound},
\[
f(\lambda^{*})=\partial_{\eta}^{+}\Sigma((1-\eta)\varphi+\eta\rho)\ge0.
\]
At the same time, it must be that $f(\lambda)\le0$ for $\lambda<\lambda^{*}$,
since otherwise we'd have $\Sigma(\varphi(\lambda))<\Sigma(\varphi)$
by convexity of $\Sigma$, contradicting the assumption that $\varphi$
is a minimizer.

Finally, observe that by definition, $\eta(\lambda)$ is a linear combination of three functions of $\lambda$:
$\Sigma(\omega(\lambda))$, $S(\rho\Vert\omega(\lambda))$, and $S(\Phi(\rho)\Vert\Phi(\omega(\lambda)))$.
All three are finite on $\lambda\in(0,1)$ as we showed above, and
all three are also convex: $\Sigma$ is convex by the lower bound
in \ref{eq:prfineq}, while $S(\cdot\Vert\cdot)$ is convex by \ref{prop:KLprop}(\klconvexref).
Hence, by \mycitep[Theorem~I.11.A]{robertsConvexFunctions1973}, all
three are continuous functions of $\lambda$ in the interval $(0,1)$,
so $f(\lambda)$ is also continuous. Therefore, since $f(\lambda)\le0$
for $\lambda<\lambda^{*}$ and $f(\lambda^{*})\ge0$, it must be that
$f(\lambda^{*})=0$. This gives \ref{eq:dn1} when combined with \ref{eq:j1}
and $\omega(\lambda^{*})=\varphi$.
\end{proof}
We now derive the equality form of mismatch cost that appears as \ref{eq:equalityBasis}
in the main text.
\begin{prop}
\label{thm:eqresult}For any $\varphi\in\mathop{\arg\min}_{\omega\in\mathcal{D}_{P}}\Sigma(\omega)$
and $\rho\in\mathcal{D}_{P}$ with $S(\rho\Vert\varphi)<\infty$,
\begin{equation}
\Sigma(\rho)-\Sigma(\varphi)=-\Delta S({\rho}\Vert{\varphi}).\label{eq:equalityCor}
\end{equation}
\end{prop}

\begin{proof}
First, consider the case when $\varphi\ge\alpha\rho$ for some $\alpha\in(0,1)$.
Then, $S(\rho\Vert\varphi)<\infty$ by \ref{prop:KLprop}(\klmonopropref), and $(1-\lambda)\varphi+\lambda\rho\in\mathcal{D}_{P}$
for $\lambda\in[-\alpha/(1-\alpha),1]$. Applying \ref{eq:dn1} gives
\ref{eq:equalityCor}.

Now consider the case where $S(\rho\Vert\varphi)<\infty$ but it is
not the case $\varphi\ge\alpha\rho$ for any $\alpha>0$ (which can
happen in infinite dimensions). Consider the sequence of states $\{\rho_{n}\}\subset\mathcal{D}_{P}$
defined in \ref{prop:epLSC}. By \ref{prop:epLSC}(I) for all $n$
there is some $\alpha_{n}>0$ such that $\rho_{n}\ge\alpha_{n}\varphi$.
Using the first part of this proof, this implies
\begin{equation}
0=\Sigma(\rho_{n})-\Sigma(\varphi)+\Delta S({\rho_{n}}\Vert{\varphi})\quad\forall n.\label{eq:gnd12}
\end{equation}
Taking the $n\to\infty$ limit infimum of both sides gives
\begin{equation}
0\ge\Sigma(\rho)-\Sigma(\varphi)+\Delta S({\rho}\Vert{\varphi}),\label{eq:bound0app-1}
\end{equation}
where we've used \ref{prop:epLSC}(II). At the same time, since $\mathcal{D}_{P}$
is a convex set, \ref{eq:convexBound} implies
\begin{equation}
0\le\Sigma(\rho)-\Sigma(\varphi)+\Delta S({\rho}\Vert{\varphi}).\label{eq:bound1app}
\end{equation}
Combining \ref{eq:bound0app-1} and \ref{eq:bound1app} gives \ref{eq:equalityCor}.
\end{proof}
The next results proves that the support of any optimizer $\varphi_{P}\in\mathop{\arg\min}_{\omega\in\mathcal{D}_{P}}\Sigma(\rho)$
and its orthogonal complement must be non-interacting subspaces under
the action of $\Phi$.
\begin{prop}
\label{thm:connectedmeansfullsupport}If $\Sigma(\vert i\rangle\langle i\vert)<\infty$
for all pure states $\vert i\rangle\langle i\vert\in\mathcal{D}_{P}$,
then for all $\varphi\in\mathop{\arg\min}_{\omega\in\mathcal{D}_{P}}\Sigma(\omega)$,
\begin{equation}
\Phi(\varphi)\perp\Phi(\rho)\quad\forall\rho\in\mathcal{D}_{P}:\rho\perp\varphi.\label{eq:connres-1}
\end{equation}
\end{prop}

\begin{proof}
The result holds trivially if $\varphi$ has maximal support, $\mathrm{supp}\,\varphi=\mathrm{supp}\,\sum_{\Pi\in P}\Pi$,
since then $\{\rho\in\mathcal{D}_{P}:\rho\perp\varphi\}$ is an empty
set. Therefore, we assume that $\mathrm{supp}\,\varphi\ne\mathrm{supp}\,\sum_{\Pi\in P}\Pi$
and prove the result by contradiction.

Pick some $\rho\in\mathcal{D}_{P}$ such that $\rho\perp\varphi$
and $\Phi(\rho)\not\perp\Phi(\varphi)$. Let $\rho$ have a spectral
resolution $\rho=\sum_{i}p_{i}\vert i\rangle\langle i\vert$, and
note that $\rho\perp\varphi$ implies that
\begin{align}
\vert i\rangle\langle i\vert\perp\varphi\qquad\forall i:p_{i}>0.\label{eq:tv3}
\end{align}
Thus $\Phi(\rho)\not\perp\Phi(\varphi)$ implies that $\Phi(\vert i\rangle\langle i\vert)\not\perp\Phi(\varphi)$
for some $i$ such that $p_{i}>0$, which means that
\begin{align}
1>\frac{1}{2}\left\Vert \Phi(\vert i\rangle\langle i\vert)-\Phi(\varphi)\right\Vert .\label{eq:tvbnd}
\end{align}
Given some pure state $\vert i\rangle\langle i\vert$ that satisfies
\ref{eq:tv3,eq:tvbnd}, define $\varphi(\lambda):=(1-\lambda)\varphi+\lambda\vert i\rangle\langle i\vert$.
Rearrange \ref{eq:conv1} in \ref{prop:epConvMix} to write
\begin{multline*}
\Sigma(\vert i\rangle\langle i\vert)-\Sigma(\varphi)=\\
\frac{\Sigma(\varphi(\lambda))-\Sigma(\varphi)-(1-\lambda)\Delta S({\varphi}\Vert{\varphi(\lambda)})}{\lambda}-\Delta S({\vert i\rangle\langle i\vert}\Vert{\varphi(\lambda)}).
\end{multline*}
Since $\Sigma(\varphi(\lambda))-\Sigma(\varphi)\ge0$ (since $\varphi$
is a minimizer) and $-\Delta S({\varphi}\Vert{\varphi(\lambda)})\ge0$
by monotonicity (\ref{prop:KLprop}(\klmonopropref )),
\begin{equation}
\Sigma(\vert i\rangle\langle i\vert)-\Sigma(\varphi)\ge-\Delta S({\vert i\rangle\langle i\vert}\Vert{\varphi(\lambda)}).\label{eq:w1}
\end{equation}
Next, rewrite the RHS as
\begin{multline}
-\Delta S({\vert i\rangle\langle i\vert}\Vert{\varphi(\lambda)})=\\
(-\ln\lambda)\frac{S(\vert i\rangle\langle i\vert\Vert\varphi(\lambda))-S(\Phi(\vert i\rangle\langle i\vert)\Vert\Phi(\varphi(\lambda)))}{-\ln\lambda}.\label{eq:w0}
\end{multline}
Audenaert showed that $S(\rho\Vert\varphi(\lambda))/(-\ln\lambda)=1$
when $\rho\perp\varphi$ \mycitep[Thm.~1]{audenaertQuantumSkewDivergence2014}
and $S(\Phi(\rho)\Vert\Phi(\varphi(\lambda)))/(-\ln\lambda)\le\frac{1}{2}\left\Vert \Phi(\rho)-\Phi(\varphi)\right\Vert _{1}$
\mycitep[Thm.~9]{audenaertQuantumSkewDivergence2014}. Plugging into
\ref{eq:w0} gives
\[
-\Delta S({\vert i\rangle\langle i\vert}\Vert{\varphi(\lambda)})\ge(-\ln\lambda)\left[1-\frac{1}{2}\left\Vert \Phi(\vert i\rangle\langle i\vert)-\Phi(\varphi)\right\Vert _{1}\right].
\]
Given \ref{eq:tvbnd}, the term inside the brackets must be strictly
positive. Therefore,
\begin{multline}
\lim_{{\lambda\to0^{+}}}-\Delta S({\vert i\rangle\langle i\vert}\Vert{\varphi(\lambda)})\ge\\
\left[1-\frac{1}{2}\left\Vert \Phi(\vert i\rangle\langle i\vert)-\Phi(\varphi)\right\Vert _{1}\right]\lim_{{\lambda\to0^{+}}}(-\ln\lambda)=\infty.\label{eq:lim1}
\end{multline}
Combining with \ref{eq:w1} gives
\[
\Sigma(\vert i\rangle\langle i\vert)-\Sigma(\varphi)\ge-\lim_{{\lambda\to0^{+}}}\Delta S({\rho}\Vert{\varphi(\lambda)})=\infty.
\]
This can only hold if $\Sigma(\vert i\rangle\langle i\vert)=\infty$,
contradicting our assumption that $\Sigma$ is finite for pure states.
Thus, $\varphi$ cannot be a minimizer.
\end{proof}

\subsection{Properties of quantum relative entropy and EP}
\begin{prop}
\label{prop:KLprop}For any $\rho,\varphi\in\mathcal{D}$ and positive
map $\Phi$, the relative entropy $S(\rho\Vert\varphi)$ obeys the
following properties:
\begin{enumerate}
\item[\klconvexprop] \textup{$S(\rho\Vert\varphi)$ is jointly convex in both arguments.}
\item[\kllimAaProp] $\lim_{{\lambda\to0^{+}}}S(\rho\Vert(1-\lambda)\varphi+\lambda\rho)=S(\rho\Vert\varphi)$.
\item[\kllimBaProp] \emph{If $S(\rho\Vert\varphi)<\infty$, then}
\begin{equation}
\lim_{{\lambda\to0^{+}}}\frac{1-\lambda}{\lambda}S(\varphi\Vert(1-\lambda)\varphi+\lambda\rho)=0.\label{eq:mz0AA}
\end{equation}
\item[\klupperboundprop] \textup{If $\ensuremath{\varphi\ge\alpha\rho}$ for some $\ensuremath{\alpha>0}$,}
\textup{then}
\begin{align}
S(\rho\Vert\varphi)\le-\ln\alpha<\infty\text{.}\label{eq:boundKL}
\end{align}
\item[\klmonoprop] Monotonicity\emph{: if $S(\rho\Vert\varphi)<\infty$, then
\[
\Delta S({\rho}\Vert{\varphi}):=S(\Phi(\rho)\Vert\Phi(\varphi))-S(\rho\Vert\varphi)\le0.
\]
}
\end{enumerate}
\end{prop}

\begin{proof}
\klconvexprop Proved in \mycitep[Lemma~2]{lindbladExpectationsEntropyInequalities1974}.\\

\noindent \kllimAaProp It is clear that $\lim_{{\lambda\to0^{+}}}(1-\lambda)\varphi+\lambda\rho=\varphi$
in the topology of the trace norm. Note that relative entropy is convex
and lower-semicontinuous in trace norm \citep{wehrlGeneralPropertiesEntropy1978}.
The result then follows from \mycitep[Corollary~7.5.1]{rockafellarConvexAnalysis1970}.\\

\noindent \noindent \kllimBaProp Define $f(\lambda):=-\frac{1-\lambda}{\lambda}\ln(1-\lambda)$
and then write
\begin{align}
& \lim_{{\lambda\to0^{+}}}\frac{1-\lambda}{\lambda}S(\varphi\Vert(1-\lambda)\varphi+\lambda\rho)\nonumber \\
& =\lim_{{\lambda\to0^{+}}}f(\lambda)\lim_{{\lambda\to0^{+}}}\frac{S(\varphi\Vert(1-\lambda)\varphi+\lambda\rho)}{-\ln(1-\lambda)}\\
& =\lim_{{\lambda\to0^{+}}}\frac{S(\varphi\Vert(1-\lambda)\varphi+\lambda\rho)}{-\ln(1-\lambda)}\label{eq:mz0a}\\
& =1-\mathrm{tr}\{\Pi^{\varphi}\rho\},\label{eq:mz0b}
\end{align}
where $\Pi^{\varphi}$ indicates a projection onto the support of
$\varphi$. In \ref{eq:mz0a}, we used that $\lim_{{\lambda\to0^{+}}}f(\lambda)=1$
from L'H\^{o}pital's rule, and in \ref{eq:mz0b} we used \mycitep[Thm.~1]{audenaert2011telescopic}.
From the definition of relative entropy in \ref{eq:relentDef}, $S(\rho\Vert\varphi)<\infty$
implies that $\mathrm{supp}\,\rho\subseteq\mathrm{supp}\,\varphi$,
so $\mathrm{tr}\{\Pi^{\varphi}\rho\}=1$. Plugging into \ref{eq:mz0b}
gives \ref{eq:mz0AA}.\\

\noindent \noindent \klupperboundprop By monotonicity of operator
logarithm, $\ensuremath{\varphi\ge\alpha\rho}$ implies $\ln\ensuremath{\varphi\ge\ln\alpha\rho=\ln\alpha+\ln\rho}$.
The claim follows by plugging this into the definition of relative entropy in \ref{eq:relentDef}.\\

\noindent \noindent \klmonoprop Proved in \citep{muller2017monotonicity}.
\end{proof}
\begin{prop}
\label{prop:epConvMix}Consider an EP-type function $\Sigma$, as
in \ref{eq:functype1} or \ref{eq:functype2}. Then, for any $\rho,\varphi\in\mathcal{D}$,
$\lambda\in(0,1)$ such that $\Sigma(\varphi(\lambda))<\infty$:
\begin{multline}
(1-\lambda)\Sigma(\varphi)+\lambda\Sigma(\rho)-\Sigma(\varphi(\lambda))=\\
-(1-\lambda)\Delta S({\varphi}\Vert{\varphi(\lambda)})-\lambda\Delta S({\rho}\Vert{\varphi(\lambda)}).\label{eq:conv1}
\end{multline}
\end{prop}

\begin{proof}
\emph{EP-type functions as in \ref{eq:functype1}.} Assume that $\Sigma(\varphi(\lambda))<\infty$.
Then, given the definition in \ref{eq:functype1}, it must be that
$S(\varphi(\lambda))$, $S(\Phi(\varphi(\lambda)))$, and $Q(\varphi(\lambda))$
are finite. By concavity of entropy, this implies that $S(\rho)$,
$S(\varphi)$, $S(\Phi(\rho))$, and $S(\Phi(\varphi))$ are finite.
Since $Q$ is linear, $Q(\varphi(\lambda))=(1-\lambda)Q(\varphi)+\lambda Q(\rho)$,
which implies that $Q(\rho)$ and $Q(\varphi)$ are finite. Again
using that $Q$ is linear, write
\begin{align*}
& (1-\lambda)\Sigma(\varphi)+\lambda\Sigma(\rho)-\Sigma(\varphi(\lambda))\\
& =[S(\varphi(\lambda))-(1-\lambda)S(\varphi)-\lambda S(\rho)]\\
& \quad-[S(\Phi(\varphi(\lambda)))-(1-\lambda)S(\Phi(\varphi))-\lambda S(\Phi(\rho))].
\end{align*}
\ref{eq:conv1} follows from the following identity (Eq. 3 in \citep{shirokovLowerSemicontinuityEntropic2017a}):
\begin{align*}
& S(\varphi(\lambda))-(1-\lambda)S(\varphi)-\lambda S(\rho)\\
& \quad=(1-\lambda)S(\varphi\Vert\varphi(\lambda))+\lambda S(\rho\Vert\varphi(\lambda)),
\end{align*}
as well as the analogous identity for $S(\Phi(\varphi(\lambda)))$.

\vspace{5pt}
\emph{EP-type functions as in \ref{eq:functype2}.} For notational
convenience define $\Psi(\rho):=V(\rho\otimes\omega)V^{\dagger}$.
Donald's identity \mycitep[Lemma~2.9]{donaldFurtherResultsRelative1987}
states that for any state $\rho'\in\mathcal{D}$ and any convex mixture
$\bar{\rho}:=\sum_{i}z_{i}\rho_{i}$,
\begin{equation}
S(\bar{\rho}\Vert\rho')=\sum_{i}z_{i}[S(\rho_{i}\Vert\rho')-S(\rho_{i}\Vert\bar{\rho})],\label{eq:donald}
\end{equation}
Using this, we write
\begin{align*}
& S(\Psi(\bar{\rho})\Vert\Phi(\bar{\rho})\otimes\omega)\\
& =S({\textstyle \sum_{i}z_{i}}\Psi(\rho_{i})\Vert\Phi(\bar{\rho})\otimes\omega)\\
& =\sum_{i}z_{i}[S(\Psi(\rho_{i})\Vert\Phi(\bar{\rho})\otimes\omega)-S(\Psi(\rho_{i})\Vert\Psi(\bar{\rho}))]\\
& =\sum_{i}z_{i}[S(\Psi(\rho_{i})\Vert\Phi(\bar{\rho})\otimes\omega)-S(\rho_{i}\Vert\bar{\rho})]
\end{align*}
where in the last line we used the invariance of relative entropy
under isometries. Then, using \mycitep[Thm.~3.12]{petzQuantumInformationTheory2008},
\begin{align}
& S(\Psi(\rho_{i})\Vert\Phi(\bar{\rho})\otimes\omega)\label{eq:petzTr}\\
& =S(\mathrm{tr}_{Y'}\Psi(\rho_{i})\Vert\Phi(\bar{\rho}))+S(\Psi(\rho_{i})\Vert\mathrm{tr}_{Y'}\Psi(\rho_{i})\otimes\omega)\nonumber \\
& =S(\Phi(\rho_{i})\Vert\Phi(\bar{\rho}))+S(\Psi(\rho_{i})\Vert\Phi(\rho_{i})\otimes\omega)\nonumber
\end{align}
Combining gives
\begin{align}
& S(\Psi(\bar{\rho})\Vert\Phi(\bar{\rho})\otimes\omega)\label{eq:zz1}\\
& =\sum_{i}z_{i}[S(\Psi(\rho_{i})\Vert\Phi(\rho_{i})\otimes\omega)+\Delta S({\rho_{i}}\Vert{\bar{\rho}})].\nonumber
\end{align}
Taking $z_{1}=1-\lambda,z_{2}=\lambda$ and $\bar{\rho}=\varphi(\lambda),\rho_{1}=\varphi,\rho_{2}=\rho$
in this identity and rearranging leads to \ref{eq:conv1}:
\begin{align}
& (1-\lambda)\Delta S(\varphi\Vert\varphi(\lambda))+\lambda\Delta S({\rho}\Vert{\varphi(\lambda)})\label{eq:zzz0}\\
& =S(\Psi(\varphi(\lambda))\Vert\Phi(\varphi(\lambda))\otimes\omega)\nonumber \\
& \quad-(1-\lambda)S(\Psi(\varphi)\Vert\Phi(\varphi)\otimes\omega)-\lambda S(\Psi(\rho)\Vert\Phi(\rho)\otimes\omega)\nonumber \\
& =\Sigma(\varphi(\lambda))-Q'(\varphi(\lambda))\nonumber \\
& \quad-(1-\lambda)S(\Psi(\varphi)\Vert\Phi(\varphi)\otimes\omega)-\lambda S(\Psi(\rho)\Vert\Phi(\rho)\otimes\omega)\nonumber \\
& =\Sigma(\varphi(\lambda))-(1-\lambda)Q'(\varphi)-\lambda Q'(\rho)\nonumber \\
& \quad-(1-\lambda)S(\Psi(\varphi)\Vert\Phi(\varphi)\otimes\omega)-\lambda S(\Psi(\rho)\Vert\Phi(\rho)\otimes\omega)\nonumber \\
& =\Sigma(\varphi(\lambda))-(1-\lambda)\Sigma(\varphi)-\lambda\Sigma(\rho).\nonumber
\end{align}
\end{proof}
\begin{prop}
\label{prop:epLSC}Consider an EP-type function $\Sigma$, as in \ref{eq:functype1}
and \ref{eq:functype2}. For any $\rho,\varphi\in\mathcal{D}_{P}$
with $\Sigma(\rho),\Sigma(\varphi),S(\rho\Vert\varphi)<\infty$, there
is a sequence $\{\rho_{n}\}\subset\mathcal{D}_{P}$ such that:
\begin{enumerate}
\item[I.] For all $n$, there is some $\alpha_{n}>0$ such that $\rho_{n}\ge\alpha_{n}\varphi$.
\item[II.] $\liminf_{n\to\infty}\Sigma(\rho_{n})+\Delta S({\rho_{n}}\Vert{\varphi})\ge\Sigma(\rho)+\Delta S({\rho}\Vert{\varphi})$.
\end{enumerate}
\end{prop}

\begin{proof}
Write a spectral resolution of $\varphi$ as $\varphi=\sum_{i}r_{i}\vert i\rangle\langle i\vert$,
where $r_{1},r_{2},\dots$ indicate the non-zero eigenvalues of $\varphi$
in decreasing order. Let $\Pi_{n}^{\varphi}:=\sum_{i=1}^{n}\vert i\rangle\langle i\vert$
indicate the projection onto the top $n$ eigenvectors of $\varphi$,
and let
\begin{equation}
\rho_{n}:=\Pi_{n}^{\varphi}\rho\Pi_{n}^{\varphi}/\mathrm{tr}\{\Pi_{n}^{\varphi}\rho\}\label{eq:pndef}
\end{equation}
indicate the normalized projection of $\rho$. Note that the basis
$\{\vert i\rangle\}$ can always be chosen so that $\rho_{n}\in\mathcal{D}_{P}$
for all $n$, by \ref{lem:staysinset} below. We then have the following
inequalities:
\begin{equation}
\mathrm{tr}\{\Pi_{n}^{\varphi}\rho\}\rho_{n}=\Pi_{n}^{\varphi}\rho\Pi_{n}^{\varphi}\le\Pi_{n}^{\varphi}I\Pi_{n}^{\varphi}=\Pi_{n}^{\varphi}\le\frac{1}{r_{n}}\varphi.\label{eq:nz0}
\end{equation}
\ref{eq:nz0} implies that $\varphi\ge\alpha_{n}\rho_{n}$ for $\alpha_{n}=r_{n}\mathrm{tr}\{\Pi_{n}^{\varphi}\rho\}>0$.
This proves part I.

Below in \ref{prop:lscappcl-2} we show that EP-type functions, as
in \ref{eq:functype1} and \ref{eq:functype2}, obey
\begin{equation}
\liminf_{n\to\infty}\Sigma(\rho_{n})\ge\Sigma(\rho).\label{eq:appEPLim0}
\end{equation}
One can also show that
\begin{align}
\lim_{n\to\infty}S(\rho_{n}\Vert\varphi) & =\lim_{n\to\infty}S(\rho_{n}\Vert\varphi_{n})-\ln\mathrm{tr}\{\Pi_{n}^{\varphi}\varphi\}\label{eq:mn1}\\
& =S(\rho\Vert\varphi).\label{eq:appKLlim0}
\end{align}
In the first line we defined $\varphi_{n}=\Pi_{n}^{\varphi}\varphi\Pi_{n}^{\varphi}/\mathrm{tr}\{\Pi_{n}^{\varphi}\varphi\}$,
and in the second line we used that $\mathrm{tr}\{\Pi_{n}^{\varphi}\varphi\}\to1$
and $S(\rho_{n}\Vert\varphi_{n})\to S(\rho\Vert\varphi)$ by \mycitep[Lemma~2.5]{donaldFurtherResultsRelative1987}.
Finally,
\begin{equation}
\liminf_{n\to\infty}S(\Phi(\rho_{n})\Vert\Phi(\varphi))\ge S(\Phi(\rho)\Vert\Phi(\varphi)),\label{eq:appKLlim1}
\end{equation}
by the lower-semicontinuity of relative entropy \citep{wehrlGeneralPropertiesEntropy1978}.
Combining \ref{eq:appEPLim0}, \ref{eq:appKLlim0}, and \ref{eq:appKLlim1}
proves part II.
\end{proof}
\begin{lem}
\label{prop:lscappcl-2}For any $\rho,\varphi\in\mathcal{D}_{P}$
with $\Sigma(\rho),\Sigma(\varphi),S(\rho\Vert\varphi)<\infty$, let
the sequence of states $\{\rho_{n}\}_{n}$ be defined as in the proof
of \ref{prop:epLSC}. Then, EP-type functions as in \ref{eq:functype1}
and \ref{eq:functype2} obey $\liminf_{n\to\infty}\Sigma(\rho_{n})\ge\Sigma(\rho)$.
\end{lem}

\begin{proof}
\emph{EP-type functions as in \ref{eq:functype1}}. Since $\Sigma(\rho)<\infty$,
it must be that $S(\rho)<\infty$. Then, $\liminf_{n\to\infty}S(\Phi(\rho_{n}))\ge S(\Phi(\rho))$
since entropy is lower-semicontinuous~\citep{wehrlGeneralPropertiesEntropy1978},
$\lim_{n}S(\rho_{n})=\lim_{n}S(\rho)$ by \mycitep[Lemma~4]{lindbladExpectationsEntropyInequalities1974},
and $\liminf_{n\to\infty}Q(\rho_{n})\ge Q(\rho_{n})$ by assumption
that $Q$ is lower-semicontinuous. Combining with the definition in
\ref{eq:functype1} gives \ref{eq:appEPLim0}.

\vspace{5pt}
\emph{EP-type functions as in \ref{eq:functype2}}. \ref{eq:appEPLim0}
holds because $\Sigma$, as defined in \ref{eq:functype2}, is lower-semicontinuous
(being the sum of two lower-semicontinuous functions, the relative
entropy \citep{wehrlGeneralPropertiesEntropy1978} and $Q'$).
\end{proof}

\subsection{Auxiliary lemma}

For the next result, we use the following notation: for any orthonormal
basis $\{\vert i\rangle\}$ and any subset of vectors $A\subseteq\{\vert i\rangle\}$,
\begin{equation}
\Pi^{A}=\sum_{\vert i\rangle\in A}\vert i\rangle\langle i\vert\label{eq:piS}
\end{equation}
indicate the projection onto the subspace spanned by $A$. In addition,
in analogy to \ref{eq:densBdef}, we use the following notation to
indicate the set of trace-class operators that are incoherent relative
to a set of orthogonal projections $P$.
\begin{equation}
\mathcal{T}_{P}:=\{\rho\in\mathcal{T}:\rho=\sum_{\Pi \in P}\Pi \rho\Pi \}\label{eq:tcsaBdef}
\end{equation}

\begin{lem}
\label{lem:staysinset}For any $\varphi,\rho\in\mathcal{T}_{P}$,
there is an orthonormal basis $\{\vert i\rangle\}$ such that $\varphi=\sum_{i}r_{i}\vert i\rangle\langle i\vert$
and for any $A\subseteq\{\vert i\rangle\}$, $\Pi^{A}\rho\Pi^{A}\in\mathcal{T}_{P}$.
\end{lem}

\begin{proof}
For any $\Pi\in P$, let $B_{\Pi}:=\{\vert\phi\rangle,\vert\phi'\rangle,\dots\}$
be a complete orthonormal basis for the Hilbert subspace $\Pi\mathcal{H}$
that diagonalizes $\Pi\varphi\Pi$. Since $\varphi\in\mathcal{T}_{P}$,
it obeys $\varphi=\sum_{\Pi\in P}\Pi\varphi\Pi$. Since each $\Pi\varphi\Pi$
is diagonal in the basis $B_{\Pi}$, $\varphi$ can be diagonalized
in the basis $B:=\bigcup_{\Pi\in P}B_{\Pi}$. It is easy to show that
$B$ is orthogonal. In particular, consider any pair of vectors in
this basis, $\vert\phi\rangle\ne\vert\psi\rangle$. If these two vectors
belong to the same $B_{\Pi}$, they are orthogonal because each $B_{\Pi}$
is an orthogonal basis. If they belong to different $B_{\Pi}\ne B_{\Pi'}$,
they are orthogonal because $\Pi$ and $\Pi'$ are orthogonal.

For any $A\subseteq B$, define $\Pi^{A}$ as in \ref{eq:piS}. Since $\Pi^A$ can be diagonalized in the same basis as all of the $\Pi\in P$, $\Pi$ and $\Pi^{A}$
commute. Then,
\begin{align*}
\Pi^{A}\rho\Pi^{A}=\Pi^{A}\Big(\sum_{\Pi\in P}\Pi\rho\Pi\Big)\Pi^{A}=\sum_{\Pi\in P}\Pi(\Pi^{A}\rho\Pi^{A})\Pi\in\mathcal{T}_{P},
\end{align*}
where in the first equality we used that $\rho\in\mathcal{T}_{P}$.
\end{proof}

\section{Mismatch Cost for Fluctuating EP}

\label{app:fluctuating}Here we derive our results for fluctuating
mismatch cost, in the case when actual initial mixed state $\rho$
and the optimal initial mixed state $\varphi$ commute. (For the non-commuting
case, we exploit results from \citep{kwonFluctuationTheoremsQuantum2019}.)

As in the main text, let $\mathcal{S}\subseteq\mathcal{D}$ be some
convex set of states, and consider some $\rho\in\mathcal{S}$ and
$\varphi\in\mathop{\arg\min}_{\omega\in\mathcal{S}}\Sigma(\omega)$
such that $S(\rho\Vert\varphi)<\infty$ and $\Sigma(\rho)-\Sigma(\varphi)=-\Delta S({\rho}\Vert{\varphi})$.
Assume that the pair of states $\rho,\varphi$ commutes, and can therefore
be simultaneously diagonalized in the same basis $\vert i\rangle\langle i\vert$,
as does the pair of states $\Phi(\rho),\Phi(\varphi)$, and can therefore
be simultaneously diagonalized in the same basis $\vert\phi\rangle\langle\phi\vert$.
For notational convenience, define
\begin{align}
p_{\rho}(i,\phi):=p_{i}T_{\Phi}(\phi\vert i)=p_{i}\mathrm{tr}\{\Phi(\vert i\rangle\langle i\vert)\vert\phi\rangle\langle\phi\vert\}.\label{eq:app2pdef}
\end{align}

\subsection{Derivation of \ref{eq:fluct0}}

Given the above definition, \ref{eq:fluct0} follows by taking the
expectation of \ref{eq:mDef0},
\begin{align}
& \langle\sigma_{\rho}-\sigma_{\varphi}\rangle_{p_{\rho}}\nonumber \\
& =\sum_{i,\phi}p_{\rho}(i,\phi)\Big[(\ln p_{i}-\ln r_{i})-(\ln p'_{\phi}-\ln r'_{\phi})\Big]\nonumber \\
& \stackrel{(a)}{=}\sum_{{i:p_{i}>0}}p_{i}(\ln p_{i}-\ln r_{i})-\sum_{{\phi:p'_{\phi}>0}}p'_{\phi}(\ln p'_{\phi}-\ln r'_{\phi})\nonumber \\
& =S(\rho\Vert\varphi)-S(\Phi(\rho)\Vert\Phi(\varphi)).\label{eq:mmm1}
\end{align}
where in $(a)$ we used
\begin{align*}
p_{\rho}(i)=\sum_{\phi}p_{\rho}(i,\phi) & =\sum_{\phi}p_{i}\mathrm{tr}\{\Phi(\vert i\rangle\langle i\vert)\vert\phi\rangle\langle\phi\vert\}\\
& =p_{i}\mathrm{tr}\{\Phi(\vert i\rangle\langle i\vert)\}=p_{i}\\
p_{\rho}(\phi)=\sum_{i}p_{\rho}(i,\phi) & =\sum_{i}p_{i}\mathrm{tr}\{\Phi(\vert i\rangle\langle i\vert)\vert\phi\rangle\langle\phi\vert\}\\
& =\mathrm{tr}\{\Phi(\rho)\vert\phi\rangle\langle\phi\vert\}=p'_{\phi},
\end{align*}
and in \ref{eq:mmm1} we used that $\rho$ and $\varphi$ can be diagonalized
in the same basis, and similarly for $\Phi(\rho)$ and $\Phi(\varphi)$.
\ref{eq:fluct0} then follows from our assumption that $\Sigma(\rho)-\Sigma(\varphi)=-\Delta S({\rho}\Vert{\varphi})$.

\subsection{Derivation of \ref{eq:IFT}}

The derivation proceeds as follows:
\begin{align}
& \langle e^{-(\sigma_{\rho}-\sigma_{\varphi})}\rangle_{p_{\rho}}\nonumber \\
& =\;\sum_{{i,\phi:p_{\rho}(i,\phi)>0}}p_{\rho}(i,\phi)e^{-[(\ln p_{i}-\ln p'_{\phi})-(\ln r_{i}-\ln r'_{\phi})]}\nonumber\\
& =\;\sum_{{i,\phi:p_{\rho}(i,\phi)>0}}p_{\rho}(i,\phi)\frac{p'_{\phi}}{p_{i}}\frac{r_{i}}{r'_{\phi}}\nonumber \\
& =\sum_{i:p_{i}>0}\sum_{\phi}p_{i}\mathrm{tr}\{\Phi(\vert i\rangle\langle i\vert)\vert\phi\rangle\langle\phi\vert\}\frac{p'_{\phi}}{p_{i}}\frac{r_{i}}{r'_{\phi}}\nonumber \\
& =\sum_{i:p_{i}>0}\sum_{\phi}\mathrm{tr}\{\Phi(\vert i\rangle\langle i\vert)\vert\phi\rangle\langle\phi\vert\}\frac{p'_{\phi}}{r'_{\phi}}r_{i}\nonumber \\
& =\sum_{i:p_{i}>0}\mathrm{tr}\{\Phi(\vert i\rangle\langle i\vert)\Phi(\rho)\Phi(\varphi)^{-1}\}r_{i}\nonumber \\
& =\mathrm{tr}\{\Phi(\varphi\Pi^{\rho})\Phi(\rho)\Phi(\varphi)^{-1}\},\label{eq:IFTres}
\end{align}
where we've used that $\Phi(\varphi\Pi^{\rho})=\sum_{i:p_{i}>0}\Phi(\vert i\rangle\langle i\vert)r_{i}$
and $\Phi(\rho)\Phi(\varphi)^{-1}=\sum_{\phi}\vert\phi\rangle\langle\phi\vert p'_{\phi}/r'_{\phi}$.
Using the definition of the Petz recovery map in \ref{eq:recov},
and the fact that the pairs $\rho,\varphi$ and $\Phi(\rho),\Phi(\varphi)$
commute, we have
\begin{align*}
\gamma & :=\mathrm{tr}\{\Pi^{\rho}\mathcal{R}_{\Phi}^{\varphi}(\Phi(\rho))\}\\
& =\mathrm{tr}\{\Pi^{\rho}\varphi^{1/2}\Phi^{\dagger}(\Phi(\varphi)^{-1/2}(\Phi(\rho))\Phi(\varphi)^{-1/2})\varphi^{1/2}\}\\
& =\mathrm{tr}\{\varphi\Pi^{\rho}\Phi^{\dagger}(\Phi(\rho)\Phi(\varphi)^{-1})\}\\
& =\mathrm{tr}\{\Phi(\varphi\Pi^{\rho})\Phi(\rho)\Phi(\varphi)^{-1}\},
\end{align*}
Combining this with \ref{eq:IFTres} gives \ref{eq:IFT}. Note that
$\gamma\in(0,1]$, since $\gamma$ is the trace of $\rho$ (with trace 1) passed through a composition of three positive
non-trace-increasing maps: $\Phi$, $\mathcal{R}_{\Phi}^{\varphi}$
\citep{jungeUniversalRecoveryMaps2018}, and $\Pi^{\rho}$. When $\rho$
has the same support as $\varphi$, $\Pi^{\rho}\varphi=\varphi$ and
therefore
\begin{align*}
\gamma=\mathrm{tr}\{\Phi(\Pi^{\rho}\varphi)\Phi(\rho)\Phi(\varphi)^{-1}\}=\mathrm{tr}\{\Phi(\varphi)\Phi(\rho)\Phi(\varphi)^{-1}\}=1.
\end{align*}

\subsection{Derivation of \ref{eq:fluctBound}}

Our derivation is standard (e.g., see Eq. 20 in \citep{jarzynski_equalities_2011})
and proceeds as follows:
\begin{align*}
& \mathrm{Pr}\big[(\sigma_{\rho}-\sigma_{\varphi})\le-\xi\big]\\
& =\sum_{i,\phi}p_{\rho}(i,\phi)\Theta(-\xi-(\sigma_{\rho}-\sigma_{\varphi}))\\
& \le\sum_{i,\phi}p_{\rho}(i,\phi)\Theta(-\xi-(\sigma_{\rho}-\sigma_{\varphi}))e^{-\xi-(\sigma_{\rho}-\sigma_{\varphi})}\\
& =e^{-\xi}\sum_{i,\phi}p_{\rho}(i,\phi)\Theta(-\xi-(\sigma_{\rho}-\sigma_{\varphi}))e^{-(\sigma_{\rho}-\sigma_{\varphi})}\end{align*}
\begin{align*}
& \le e^{-\xi}\sum_{i,\phi}p_{\rho}(i,\phi)e^{-(\sigma_{\rho}-\sigma_{\varphi})}=\gamma e^{-\xi}.
\end{align*}
where $\Theta$ is the Heavyside function ($\Theta(x)=1$ if $x\ge0$
and $\Theta(x)=0$ otherwise) and the last line used the IFT.

\subsection{Derivation of \ref{eq:ldb1}}

First, write
\begin{align*}
& \frac{T_{\Phi}(\phi\vert i)}{T_{\mathcal{R}_{\Phi}^{\varphi}}(i\vert\phi)}=\frac{\mathrm{tr}\{\Phi(\vert i\rangle\langle i\vert)\vert\phi\rangle\langle\phi\vert\}}{\mathrm{tr}\{\mathcal{R}_{\Phi}^{\varphi}(\vert\phi\rangle\langle\phi\vert)\vert i\rangle\langle i\vert\}}\\
& =\frac{\mathrm{tr}\{\Phi(\vert i\rangle\langle i\vert)\vert\phi\rangle\langle\phi\vert\}}{\mathrm{tr}\{\varphi^{1/2}\Phi^{\dagger}(\Phi(\varphi)^{-1/2}(\vert\phi\rangle\langle\phi\vert)\Phi(\varphi)^{-1/2})\varphi^{1/2}\vert i\rangle\langle i\vert\}}\\
& =\frac{\mathrm{tr}\{\Phi(\vert i\rangle\langle i\vert)\vert\phi\rangle\langle\phi\vert\}}{\mathrm{tr}\{\Phi^{\dagger}(\vert\phi\rangle\langle\phi\vert/r'_{\phi})\vert i\rangle\langle i\vert r_{i}\}}=\frac{r'_{\phi}}{r_{i}},
\end{align*}
where we used the definition of the Petz recovery map in \ref{eq:recov}.
The result then follows by combining with \ref{eq:mm0}.

\section{Mismatch Cost for EP rate}

\label{app:EPrate}

\subsection{Main proofs}

Here we analyze mismatch cost for the EP rate, which has the general
form
\begin{equation}
\dot{\Sigma}(\rho)={\textstyle {\textstyle \frac{d}{dt}}}S(\rho(t))+\dot{Q}(\rho),\label{eq:EPrn}
\end{equation}
where $\rho$ evolves according to a Lindblad equation ${\textstyle \frac{d}{dt}}\rho(t)\!=\!\mathcal{L}(\rho(t))$,
and $\dot{Q}:\mathcal{D}\to\mathbb{R}\cup\{\infty\}$ is a linear
functional that reflects the rate of entropy flow into the environment.
Note that our results also apply to other ``EP rate''-type functionals
(such as rate of nonadiabatic EP, entropy gain, etc.), which correspond
to different choices of the linear functional $\dot{Q}$.

Consider some pair of states $\varphi,\rho\in\mathcal{D}$ such that
$\dot{\Sigma}(\rho)<\infty,\dot{\Sigma}(\varphi)<\infty,S(\rho\Vert\varphi)<\infty$. As before, let $\varphi(\lambda)=(1-\lambda)\varphi+\lambda\rho$
indicate a linear mixture of the two states. Our results will reference
the following regularity assumptions regarding the behavior of the
EP rate $\dot{\Sigma}(\varphi(\lambda))$ in the neighborhood of $\lambda=0$.
\begin{condition} \label{cond:eprate1}The following (one-sided)
partial derivatives at $\lambda=0,t=0$ are symmetric:
\begin{align}
\partial_{\lambda}^{+}\dot{\Sigma}(\varphi(\lambda)) & =\partial_{t}^{+}\partial_{\lambda}^{+}\int_{0}^{t}\dot{\Sigma}(e^{t'\mathcal{L}}(\rho))\,dt'.\label{eq:symmP}
\end{align}
\end{condition}

\begin{condition} \label{cond:eprate2}If $\varphi\ge\alpha\rho$
for some $\alpha>0$, then $\lambda\mapsto\dot{\Sigma}(\varphi(\lambda))$
is finite and continuously differentiable in some neighborhood of
$\lambda=0$. \end{condition}

Importantly, these two conditions always hold in finite dimensions,
as shown below in \ref{thm:eprateDeriv}.

If \ref{cond:eprate1} holds, then it is straightforward to show that
the directional derivative of $\dot{\Sigma}$ in the direction of
$\varphi$ at $\rho$ has a simple information-theoretic form. In
particular, use $\dot{\Sigma}$ to define a time-dependent integrated
EP as a function of the initial state $\rho$ at $t=0$,
\begin{align}
\Sigma(\rho,t) & =\int_{0}^{t}\dot{\Sigma}(e^{t'\mathcal{L}}(\rho))\,dt'\label{eq:EPinstEP}\\
& =S(e^{t\mathcal{L}}(\rho))-S(\rho)+Q(\rho,t),\nonumber
\end{align}
where $Q(\rho,t)=\int_{0}^{t}\dot{Q}(e^{t'\mathcal{L}}(p))\,dt'$ is
the integrated entropy flow. This is an EP-type function of type \ref{eq:functype1}
(technically, we have not shown that $Q$ is lower-semicontinuous in $\rho$;
however, this will not be required for the integrated EP results we
reference in our analysis of EP rate). One can then write
\begin{align}
\partial_{\lambda}^{+}\dot{\Sigma}(\varphi(\lambda))\vert_{\lambda=0} & =\partial_{t}^{+}\partial_{\lambda}^{+}\Sigma(\rho,t)\nonumber \\
& =\partial_{t}^{+}[\Sigma(\rho,t)-\Sigma(\varphi,t)+\Delta S({\rho}\Vert{\varphi})]\nonumber \\
& =\dot{\Sigma}(\rho)-\dot{\Sigma}(\varphi)+{\textstyle {\textstyle \frac{d}{dt}}}S(\rho(t)\Vert\varphi(t)),\label{eq:EPratedd}
\end{align}
where we used \ref{eq:symmP} and \ref{thm:genf}.

We use this result to derive bounds on instantaneous mismatch cost
(i.e., mismatch cost for instantaneous EP rate). \ref{eq:convexBoundEPR}
and \ref{eq:dn1EPR} appear in the main text as \ref{eq:ineqConvex-1}.
\begin{prop}
\label{thm:ineqsEPR}Given a convex set of states $\mathcal{S}\subseteq\mathcal{D}$,
consider any $\varphi\in\mathop{\arg\min}_{\omega\in\mathcal{S}}\dot{\Sigma}(\omega)$
and $\rho\in\mathcal{S}$. If $S(\rho\Vert\varphi)<\infty$ and \ref{cond:eprate1}
holds,
\begin{align}
\dot{\Sigma}(\rho)-\dot{\Sigma}(\varphi) & \ge-{\textstyle {\textstyle \frac{d}{dt}}}S(\rho(t)\Vert\varphi(t)).\label{eq:convexBoundEPR}
\end{align}
Furthermore, if $\varphi(\lambda)\in\mathcal{S}$ for some $\lambda<0$
and \ref{cond:eprate2} holds,
\begin{align}
\dot{\Sigma}(\rho)-\dot{\Sigma}(\varphi) & =-{\textstyle {\textstyle \frac{d}{dt}}}S(\rho(t)\Vert\varphi(t)).\label{eq:dn1EPR}
\end{align}
\end{prop}

\begin{proof}
Within the convex set $\mathcal{S}$, the directional derivative from
the minimizer $\varphi$ of $\dot{\Sigma}$ toward any $\rho$ must
be non-negative, ${\textstyle \partial_{\lambda}^{+}}\dot{\Sigma}(\varphi(\lambda))\vert_{\lambda=0}\ge0$.
\ref{eq:convexBoundEPR} then follows from \ref{cond:eprate1} and
\ref{eq:EPratedd}.

To derive \ref{eq:dn1EPR}, consider some $\alpha<0$ such that $(1-\alpha)\varphi+\alpha\rho\in\mathcal{S}$.
Then, $\varphi\ge-\alpha\rho/(1-\alpha)$ and so by \ref{cond:eprate2}
the function $\lambda\mapsto\dot{\Sigma}(\varphi(\lambda))$ is finite
and continuously differentiable in some neighborhood of $\lambda=0$.
That means that the directional derivative must vanish at the minimizer
$\lambda=0$, ${\textstyle \partial_{\lambda}^{+}}\dot{\Sigma}(\varphi(\lambda))\vert_{\lambda=0}=0$.
\ref{eq:dn1EPR} then follows from \ref{eq:EPratedd}.
\end{proof}
We now derive the equality form of instantaneous mismatch cost, which
appears as \ref{eq:EPr-equality} in the main text. To derive the
next result, we require that
\begin{equation}
\varphi\ge\alpha\rho\qquad\text{for some} \text{\ensuremath{\alpha}}>0.\label{eq:condDMAX}
\end{equation}
It is simple to show that in finite dimensions, \ref{eq:condDMAX}
is equivalent to requiring that $S(\rho\Vert\varphi)<\infty$ (this
is the condition mentioned in the main text when presenting \ref{eq:EPr-equality},
where only the finite dimensional case is analyzed). In infinite dimensions,
\ref{eq:condDMAX} is stronger that $S(\rho\Vert\varphi)<\infty$.
Interestingly, \ref{eq:condDMAX} can be restated in information-theoretic
terms as $S_{\max}(\rho\Vert\varphi)<\infty$, where $S_{\max}$ is
the so-called ``max-relative entropy'' \mycitep[Defn.10]{bertaSmoothEntropyFormalism2016a},
\[
S_{\max}(\rho\Vert\varphi)=\inf\{x\in\mathbb{R}:\varphi\ge2^{-x}\rho\}.
\]

\begin{prop}
\label{thm:eprate}Consider any $\varphi\in\mathop{\arg\min}_{\omega\in\mathcal{D}_{P}}\dot{\Sigma}(\omega)$
and $\rho\in\mathcal{D}_{P}$ such that $\dot{\Sigma}(\rho)<\infty$.
If $\varphi\ge\alpha\rho$ for some $\alpha>0$ and \ref{cond:eprate1,cond:eprate2}
holds,
\begin{equation}
\dot{\Sigma}(\rho)-\dot{\Sigma}(\varphi)=-{\textstyle {\textstyle \frac{d}{dt}}}S(\rho(t)\Vert\varphi(t)).\label{eq:equalityEPrateApp}
\end{equation}
\end{prop}

\begin{proof}
$\varphi\ge\alpha\rho$ for some $\alpha>0$ implies that $(1+\alpha)\varphi-\alpha\rho\ge0$,
so $\varphi\ge\frac{\alpha}{1+\alpha}\rho$ and therefore $S(\rho\Vert\varphi)<\infty$
by \ref{prop:KLprop}(\klupperboundpropref). \ref{eq:equalityEPrateApp} then follows from
\ref{eq:dn1EPR}.
\end{proof}
Our next results shows that our technical assumptions about $\dot{\Sigma}$
are always satisfied in finite dimensions.
\begin{prop}
\label{thm:eprateDeriv}Assume that $\dim\mathcal{H}<\infty$. Then,
\ref{cond:eprate1,cond:eprate2} hold for any pair of states $\varphi,\rho\in\mathcal{D}$
such that $\dot{\Sigma}(\rho),\dot{\Sigma}(\varphi),S(\rho\Vert\varphi)<\infty$.
\end{prop}

\begin{proof}
First, note that in finite dimensions, $S(\rho\Vert\varphi)<\infty$ implies that $\mathrm{supp}\,\rho\subseteq\mathrm{supp}\,\varphi$
which, by \ref{lem:EPRsupp} below, means there is some $\alpha>0$
such that $\varphi(\lambda)\ge0$ for all $\lambda\in(-\alpha,1)$.

We now show that $\vert\dot{\Sigma}(\varphi(\lambda))\vert<\infty$
for all $\lambda\in(-\alpha,1)$. It is easy to see that $\dot{\Sigma}(\rho),\dot{\Sigma}(\varphi)<\infty$
implies that $\dot{Q}(\rho),\dot{Q}(\varphi)<\infty$ (see \ref{eq:EPrn}).
Since $\dot{Q}$ is a linear function, $\dot{Q}(\varphi(\lambda))=(1-\lambda)\dot{Q}(\varphi)+\lambda\dot{Q}(\rho)<\infty$
for all $\lambda\in(-\alpha,1)$. Then, in finite dimensions, the
derivative of the entropy obeys~\citep{spohn_entropy_1978,das2018fundamental}
\begin{equation}
{\textstyle {\textstyle \frac{d}{dt}}}S(\rho(t))=-\mathrm{tr}\{\mathcal{L}(\rho)\ln\rho\}=-\sum_{i}\langle i\vert\mathcal{L}(\rho)\vert i\rangle\ln p_{i},\label{eq:hf2}
\end{equation}
where we used the spectral resolution $\rho=\sum_{i}p_{i}\vert i\rangle\langle i\vert$
in some complete basis $\{\vert i\rangle\}$, and assume $0\ln0=0$
(as standard). From this expression, it is easy to see that $|{\textstyle {\textstyle \frac{d}{dt}}}S(\rho(t))|<\infty$
if and only if there is no $i$ such that $\langle i\vert\mathcal{L}(\rho)\vert i\rangle>0,p_{i}=0$,
or in other words iff $\mathrm{supp}\,\mathcal{L}(\rho)\subseteq\mathrm{supp}\,\rho$.
Given our assumption that $\dot{\Sigma}(\varphi),\dot{\Sigma}(\rho)<\infty$,
it must be that ${\textstyle {\textstyle \frac{d}{dt}}}S(\varphi(t)),{\textstyle {\textstyle \frac{d}{dt}}}S(\rho(t))<\infty$.
Therefore, $\mathrm{supp}\,\mathcal{L}(\varphi)\subseteq\mathrm{supp}\,\varphi$
and $\mathrm{supp}\,\mathcal{L}(\rho)\subseteq\mathrm{supp}\,\rho$.
Furthermore, $S(\rho\Vert\varphi)<\infty$ implies $\mathrm{supp}\,\rho\subseteq\mathrm{supp}\,\varphi$,
which means that $\mathrm{supp}\,\mathcal{L}(\rho)\subseteq\mathrm{supp}\,\varphi$.
This means that for $\lambda\in(-\alpha,1]$,
\begin{align}
\mathrm{supp}\,\mathcal{L}(\varphi(\lambda))\!=\!\mathrm{supp}\,[(1-\lambda)\mathcal{L}(\varphi)\!+\!\lambda\mathcal{L}(\rho)]\!\subseteq\!\mathrm{supp}\,\varphi.\label{eq:supp1app}
\end{align}
Combining \ref{eq:supp1app} with \ref{eq:suppapp2} in \ref{lem:EPRsupp}
gives
\[
\mathrm{supp}\,\mathcal{L}(\varphi(\lambda))\subseteq\mathrm{supp}\,\varphi(\lambda)\qquad\text{for all \ensuremath{\lambda\in(-\alpha,1)}}.
\]
Thus, $\vert{\textstyle {\textstyle \frac{d}{dt}}}S(\varphi(\lambda)(t))\vert<\infty$
for all $\lambda\in(-\alpha,1)$, which also means that $\vert\dot{\Sigma}(\varphi(\lambda))\vert<\infty$,
therefore proving the first part of \ref{cond:eprate2}.

Now consider the (two-sided) of the function $\lambda\mapsto\dot{\Sigma}(\varphi(\lambda))$
in the neighborhood of $\lambda=0$. Using \ref{eq:EPrn} and \ref{eq:hf2},
we write
\[
{\textstyle \partial_{\lambda}}\dot{\Sigma}(\varphi(\lambda))=-\partial_{\lambda}\mathrm{tr}\{\mathcal{L}(\varphi(\lambda))\ln\varphi(\lambda)\}+\dot{Q}(\rho-\varphi).
\]
This derivative is continuous in $\lambda$, since $\lambda\mapsto\mathcal{L}(\varphi(\lambda))$,
$\lambda\mapsto\ln\varphi(\lambda)$ are continuous in finite dimensions.
This proves the second part of f \ref{cond:eprate2}.

To prove \ref{cond:eprate1}, define the integrated EP function $\Sigma(\rho,t)$
as in \ref{eq:EPinstEP}. As we showed, the following limit is finite
for all $\lambda\in(-\alpha,1)$,
\begin{equation}
\dot{\Sigma}(\varphi(\lambda))=\partial_{t}^{+}\Sigma(\varphi(\lambda),t)=\lim_{t\to0^{+}}\frac{1}{t}\Sigma(\varphi(\lambda),t).\label{eq:pD}
\end{equation}
In addition, for each $t>0$, the map $\rho\mapsto\Sigma(\rho,t)$
is an EP-type function as in \ref{eq:functype1}. Therefore, the function
$\lambda\mapsto\Sigma(\varphi(\lambda),t)$ is convex over $\lambda\in(-\alpha,1)$.
This means that $\lim_{{\lambda\to0^{+}}}\frac{1}{t}\frac{1}{\lambda}\Sigma(\varphi(\lambda),t)$
exists for all $t$ \mycitep[Thm.~23.1]{rockafellarConvexAnalysis1970}.
Sequences of convex functions converge uniformly, and in particular
$\lim_{t\to0^{+}}\frac{1}{t}\frac{1}{\lambda}\Sigma(\varphi(\lambda),t)$
converges uniformly over $\lambda\in[0,1/2]$~\mycitep[Thm.~10.8]{rockafellarConvexAnalysis1970}.
This allows us to exchange the order of limits,
\[
\lim_{{\lambda\to0^{+}}}\frac{1}{\lambda}\lim_{t\to0^{+}}\frac{1}{t}\Sigma(\varphi(\lambda),t)=\lim_{t\to0^{+}}\frac{1}{t}\lim_{{\lambda\to0^{+}}}\frac{1}{\lambda}\Sigma(\varphi(\lambda),t),
\]
which proves \ref{cond:eprate1}.
\end{proof}
The next result is used to show that in many cases of interest, the
minimizer of EP rate will have full support.
\begin{prop}
\label{prop:epr-fsupp}Assume that $\dim\mathcal{H}<\infty$, and
suppose that
\begin{equation}
\mathrm{supp}\,\mathcal{L}(\rho)\not\subseteq\mathrm{supp}\,\rho\quad\forall\rho\in\mathcal{D}_{P}:\mathrm{supp}\,\rho\ne\mathcal{H}_{P}.\label{eq:fsupp-1}
\end{equation}
Then, any $\varphi\in\mathop{\arg\min}_{\omega\in\mathcal{D}_{P}}\dot{\Sigma}(\omega)$
obeys $\mathrm{supp}\,\varphi=\mathcal{H}_{P}$.
\end{prop}

\begin{proof}
Note that ${\textstyle {\textstyle \frac{d}{dt}}}S(\rho(t))=\infty$
(and hence $\dot{\Sigma}(\rho)=\infty$) whenever $\mathrm{supp}\,\mathcal{L}(\rho)\not\subseteq\mathrm{supp}\,\rho$,
as shown in the proof of \ref{thm:eprateDeriv}. Since the minimizer
$\varphi$ must have $\dot{\Sigma}(\varphi)<\infty$, \ref{eq:fsupp-1}
implies that it cannot be that $\mathrm{supp}\,\varphi\ne\mathcal{H}_{P}$.
\end{proof}

\subsection{Auxiliary lemma}

The following lemma is used in some of the results above.
\begin{lem}
\label{lem:EPRsupp} If $\dim\mathcal{H}<\infty$ and $\mathrm{supp}\,\rho\subseteq\mathrm{supp}\,\varphi$,
then there is some $\alpha>0$ such that for all $\lambda\in(-\alpha,1)$,
\begin{align}
0 & \le(1-\lambda)\varphi+\lambda\rho\label{eq:suppapp1}\\
\mathrm{supp}\,\varphi & \subseteq\mathrm{supp}\,[(1-\lambda)\varphi+\lambda\rho]\label{eq:suppapp2}
\end{align}
\end{lem}

\begin{proof}
Let $\Pi^{\varphi}$ indicate the projection
onto the support of $\varphi$. Since $\dim\mathcal{H}<\infty$ and
$\mathrm{supp}\,\rho\subseteq\mathrm{supp}\,\varphi$,
\begin{equation}
\varphi\ge \alpha\Pi^{\varphi}\ge \alpha\rho,\label{eq:n2-3}
\end{equation}
where $\alpha>0$ is the smallest non-zero eigenvalue of $\varphi$.
Note that $0\le(1-\lambda)\varphi+\lambda\rho$
for $\lambda \in \{-\alpha,1\}$, hence also for all
$\lambda\in[-\alpha,1]$ (since the set of positive operators is convex).

Next we derive \ref{eq:suppapp2}. For any $\vert a\rangle\in\mathrm{supp}\,\varphi$
and $-\alpha < \lambda\ < 0$,
\begin{align*}
\langle a\vert(1-\lambda)\varphi+\lambda\rho\vert a\rangle & =(1-\lambda)\langle a\vert\varphi\vert a\rangle+\lambda\langle a\vert\rho\vert a\rangle\\
& > \langle a\vert\varphi\vert a\rangle - \alpha \langle a\vert\rho\vert a\rangle\\
& \ge \alpha\langle a \vert a \rangle - \alpha\langle a \vert a \rangle =0,
\end{align*}
where the strict inequality uses $\langle a\vert\varphi\vert a\rangle>0$ and $-\alpha < \lambda<0$.
Then, for any $0\le \lambda < 1$,
\begin{align*}
\langle a\vert(1-\lambda)\varphi+\lambda\rho\vert a\rangle & =(1-\lambda)\langle a\vert\varphi\vert a\rangle+\lambda\langle a\vert\rho\vert a\rangle\\
& \ge(1-\lambda)\langle a\vert\varphi\vert a\rangle>0,
\end{align*}
where the strict inequality uses $\langle a\vert\varphi\vert a\rangle>0$ and $0\le \lambda < 1$.
Combining implies that for all $\lambda\in(-\alpha,1)$, $\vert a\rangle \in \mathrm{supp} [(1-\lambda)\varphi+\lambda\rho]$
for all $\vert a\rangle\in\mathrm{supp}\,\varphi$, proving \ref{eq:suppapp2}.
\end{proof}

\section{Classical processes}

\label{app:classical}

In this appendix, we show that our expressions for mismatch cost also
apply to classical systems, as briefly discussed in \ref{sec:classical}
in the main text.

We first consider discrete-state classical systems, and show that
our quantum results immediately apply to them as a special case. After
that, we consider continuous-state classical systems, and demonstrate
how our quantum results can again be applied, once some appropriate
modifications are made.

Below we write classical entropy and entropy production in sans-serif
font, $\mathsf{S}$ and $\mathsf{\Sigma}$, so as to distinguish them
from quantum entropy $S$ and entropy production $\Sigma$. We will
also make use of classical relative entropy, also called Kullback-Leibler
(KL) divergence. The KL divergence between two probability density
functions $p$ and $r$ can be written as
\begin{align}
D(p\Vert r)= & \begin{cases}
\int p(x)\ln\frac{p(x)}{r(x)}dx & \text{if \ensuremath{\mathrm{supp}\,p\subseteq\mathrm{supp}\,r}}\\
\infty & \text{otherwise},
\end{cases}\label{eq:relentDef-1}
\end{align}
where $\mathrm{supp}\,p:=\{x\in X:p(x)>0\}$ indicates the support
of $p$ (and similarly for $r$). The same definition applies to discrete-state
probability mass functions, as long as the integral is replaced with
summation.

In this appendix we focus on converting results concerning EP in quantum
systems into results concerning EP in classical systems. We note though
that the same kind of reasoning we use below can also be used to convert
our results concerning the quantum ``EP-type'' functions discussed
in \ref{sec:EP-Type-functionals} into results concerning the associated
classical EP-type functions (e.g., classical non-adiabatic EP, entropy
gain, etc). All that's needed for our reasoning to apply is that the
classical EP-type function can be written in the form of classical
EP (\ref{eq:classicEP}, \ref{eq:EPclassicalKL}, or \ref{eq:contm2}
below), where $G$ is an arbitrary linear functional of the initial distribution $p$.

\subsection{Classical processes in discrete state-space}

\label{app:classical-discrete}

\subsubsection{Integrated EP}

We first discuss how our analysis of quantum mismatch cost for integrated
EP applies to discrete-state classical systems. Consider a classical
system with a discrete state space $X$ which undergoes a driving
protocol over some time interval $t\in[0,\tau]$ while coupled to
some thermodynamic reservoirs. As mentioned in \ref{subsec:classical-Integrated-EP},
we use $\mathrm{P}(\bm{x}\vert x_{0})$ to indicate the conditional
probability of the system undergoing the trajectory $\bm{x}=\{x_{t}:t\in[0,\tau]\}$
under the regular (``forward'') protocol, given initial microstate
$x_{0}$. We will also sometimes write the conditional probability
of final microstates $j$ given initial microstate $i$ in terms of
the transition matrix $T(j\vert i)=\mathrm{P}(x_{\tau}=j\vert x_{0}=i)$,
so that the map from initial to final distributions can be expressed
in matrix notation as $p'=Tp$. In addition, it will sometimes be
useful to consider the conditional probability $\tilde{\mathrm{P}}(\tilde{\bm{x}}\vert\tilde{x}_{\tau})$
of observing the \emph{time-reversed} trajectory $\tilde{\bm{x}}=\{\tilde{x}_{\tau-t}:t\in[0,\tau]\}$
under the \emph{time-reversed} driving protocol given initial microstate
$\tilde{x}_{\tau}$ (tilde notation like $\tilde{x}$ indicates conjugation
of odd variables such as momentum \citep{ford_entropy_2012,spinneyNonequilibriumThermodynamicsStochastic2012}).

Let the elements of the state space $X$ index a set of pure quantum
states in some complete orthonormal reference basis $\{\vert i\rangle:i\in X\}$.
One can then choose $P=\{\vert i\rangle\langle i\vert\}_{i\in X}$
and define $\mathcal{D}_{P}$ as in \ref{eq:densBdef} (i.e., as the
set of density operators diagonal in the reference basis). Any probability
distribution $p$ over $X$ now corresponds to the mixed quantum state
\begin{equation}
\rho^{p}=\sum_{i}p_{i}\vert i\rangle\langle i\vert\in\mathcal{D}_{P}.\label{eq:XpDef}
\end{equation}
Note that the quantum and classical relative entropy are identical
when applied to elements of $\mathcal{D}_{P}$:
\begin{equation}
S({\rho^{p}}\Vert{\rho^{r}})=D(p\Vert r).\label{eq:relEntAgree}
\end{equation}
Conversely to \ref{eq:XpDef}, any quantum state $\rho$ can be turned
into a distribution over $X$ via
\begin{equation}
p_{i}^{\rho}=\langle i\vert\rho\vert i\rangle.\label{eq:PXpdef}
\end{equation}
Note that the map $\rho\mapsto p^{\rho}$ is many-to-one, as it ignores
all off-diagonal elements of $\rho$ relative to the reference basis
(i.e., it ignores any coherence in $\rho$).

Now consider the quantum channel, which is defined in terms of $T$
as
\begin{align}
\Phi(\rho) & :=\sum_{i,j}T(j\vert i)\langle i\vert\rho\vert i\rangle\vert j\rangle\langle j\vert.\label{eq:clch0}
\end{align}
Applying the classical transition matrix $T$ to the classical distribution
$p$ and then converting it into a density matrix via \ref{eq:XpDef}
is equivalent to applying $\Phi$ to the associated quantum mixed
state $\rho^{p}$:
\begin{align}
\Phi(\rho^{p})=\sum_{j}\Big(\sum_{i}T(j\vert i)p_{i}\Big)\vert j\rangle\langle j\vert=\rho^{Tp}.\label{eq:clch}
\end{align}
In this sense, maps between the classical and quantum pictures commute
with the associated dynamic operators.

The expected classical entropy flow can also be written in terms of
a quantum functional, which is defined in terms of $G$ as
\begin{align}
Q(\rho):=G(p^{\rho}).\label{eq:chef}
\end{align}
$Q$ is a linear functional (since we assumed $G$ is linear).
In addition, for any ``classical'' mixed state
$\rho^{p}\in\mathcal{D}_{P}$, $Q(\rho^{p})=G(p)$ as expected.

Note that although $Q$ and $\Phi$ are defined in a quantum manner,
they behave classically. In particular, they are both invariant to
coherence relative to the reference basis $\{\vert i\rangle\}$,
\begin{align}
\Phi(\rho)=\Phi(\mathcal{P}_{P}(\rho)),\;Q(\rho)=Q(\mathcal{P}_{P}(\rho))\quad\forall\rho\in\mathcal{D},\label{eq:app83}
\end{align}
where $\mathcal{P}_{P}(\rho)=\sum_{i}\vert i\rangle\langle i\vert\rho\vert i\rangle\langle i\vert$
is the ``pinching map'' for the reference basis~\citep{tomamichelQuantumInformationProcessing2016}.
In addition, the output of $\Phi$ is always diagonal in the reference
basis, so its outputs always commute,
\begin{align}
[\Phi(\rho),\Phi(\varphi)]=0\quad\forall\rho,\varphi\in\mathcal{D}.\label{eq:app84}
\end{align}

With these definitions, the standard definition of integrated EP in
classical stochastic thermodynamics, \ref{eq:classicEP} (or equivalently
\ref{eq:EPclassicalKL}), can be seen as a special case of quantum
integrated EP, as defined in \ref{eq:functype1}, i.e., $\mathsf{\Sigma}(p)=\Sigma(\rho^{p})$.
Therefore one can analyze classical mismatch cost using the results
in the main text, such as \ref{eq:equalityBasis,eq:ineqConvex}, by
considering the quantum channel $\Phi$ and entropy flow functional
$Q$ defined above, and by restricting attention to the set of mixed
states in $\mathcal{D}_{P}$.

It is also possible to analyze classical mismatch cost within
the subset of probability distributions whose support is restricted
to some subset of microstates $S\subseteq X$. This can be done by
choosing $P$ to be the corresponding subset of pure states, $P=\{\vert i\rangle\langle i\vert\}_{i\in S}$,
and then analyzing mismatch cost within the resulting set of diagonal
mixed states $\mathcal{D}_{P}$.

\subsubsection{Fluctuating EP}

\label{app:classical-fluct-discrete}

Consider a quantum channel that has the form given in \ref{eq:clch}
and an entropy flow function that has the form given in \ref{eq:chef},
as might represent entropy flow in a classical system. We consider
two mixed states $\rho^{p}=\sum_{i}p_{i}\vert i\rangle\langle i\vert\in\mathcal{D}_{P}$
and $\rho^{r}=\sum_{i}r_{i}\vert i\rangle\langle i\vert\in\mathcal{D}_{P}$
that correspond to two classical probability distributions $p$ and
$r$, and we will use the shorthand $p'=Tp$ and $r'=Tr$. As in the
main text, we assume that $D(p\Vert r)<\infty$ and
\[
\Sigma(\rho^{p})-\Sigma(\rho^{r})=-\Delta S({\rho^{p}}\Vert{\rho^{r}}).
\]
(In particular, this might be because $\rho^{r}$ is a minimizer of
EP in some convex set.) It is clear that $\rho^{p}$ and $\rho^{r}$
commute since they are both diagonal in the reference basis. In addition,
$\Phi(\rho^{p})=\rho^{p'}$ and $\Phi_{T}(\rho^{r})=\rho^{r'}$ must
also commute, given \ref{eq:app84}. Therefore the simple commuting
case of fluctuating mismatch cost which is analyzed in the main text,
and in more detail in \ref{app:fluctuating}, applies to all classical
processes. In particular, the fluctuating mismatch cost in \ref{eq:mDef0}
can be written as in terms of probability values in $p$ and $r$
as
\begin{multline}
\sigma_{{\rho^{p}}}(i\to j,q)-\sigma_{{\rho^{r}}}(i\to j,q)=\\
(-\ln p'_{j}+\ln p_{i})-(-\ln r'_{j}+\ln r_{i}).
\end{multline}
This classical special case of fluctuating mismatch cost obeys the
fluctuating mismatch cost results described in the main text. In particular,
it agrees with average mismatch cost in expectation,
\begin{align*}
& \big\langle\sigma_{{\rho^{p}}}-\sigma_{{\rho^{r}}}\big\rangle_{\mathrm{P}(\bm{x}\vert x_{0})p(x_{0})}\\
& \qquad=-\Delta S({\rho^{p}}\Vert{\rho^{r}})=\Sigma({\rho^{p}})-\Sigma({\rho^{r}})\\
& \qquad=-\Delta D(p\Vert r)=\mathsf{\Sigma}(p)-\Sigma(r).
\end{align*}
In addition, it obeys an integral fluctuation theorem,
\begin{equation}
\big\langle e^{\sigma_{{\rho^{p}}}-\sigma_{{\rho^{r}}}}\big\rangle_{\mathrm{P}(\bm{x}\vert x_{0})p(x_{0})}=\gamma,\label{eq:iftcl}
\end{equation}
where
\begin{equation}
\gamma=\sum_{j}p'_{j}\frac{\sum_{i}T(j\vert i)r_{i}\mathbf{1}_{\mathrm{supp}\,p}(i)}{r'_{j}}\in(0,1],\label{eq:classicalGamma}
\end{equation}
where $\mathbf{1}$ is the indicator function.
\ref{eq:iftcl} is the classical analogue of \ref{eq:IFT}. It implies
that negative values of classical fluctuating mismatch cost are exponentially
unlikely: $\mathrm{Pr}\big[(\sigma_{{\rho^{p}}}-\sigma_{{\rho^{r}}})\le-\xi\big]\le\gamma e^{-\xi}$
(see \ref{app:fluctuating}).

For this classical channel, the Petz recovery map is simply
the Bayesian inverse of the transition matrix with respect to the
reference probability distribution \citep{leiferFormulationQuantumTheory2013,wildeQuantumInformationTheory2017}.
In other words, plugging $\Phi$ from \ref{eq:clch0} and $\varphi={\rho^{r}}$
into \ref{eq:recov} gives
\begin{equation}
T_{\mathcal{R}_{\Phi}^{\varphi}}(i\vert j)=\frac{T(j\vert i)r_{i}}{\sum_{i'}T(j\vert i')r_{i'}}.\label{eq:trm2}
\end{equation}
Thus the classical analogue of \ref{eq:ldb1} holds, which allows
us to write the classical mismatch cost as
\begin{multline}
\sigma_{{\rho^{p}}}(i\to j,q)-\sigma_{{\rho^{r}}}(i\to j,q)=\\
(-\ln p_{j}'+\ln p_{i})+\ln\frac{T(j\vert i)}{T_{\mathcal{R}_{\Phi}^{\varphi}}(i\vert j)}=\ln\frac{T(j\vert i)p_{i}}{T_{\mathcal{R}_{\Phi}^{\varphi}}(i\vert j)p'_{j}}.\label{eq:clPetz}
\end{multline}
In this sense, the classical fluctuating mismatch cost of $p$ quantifies
the time-asymmetry between the forward process and the reverse process,
as defined by the Bayesian inverse of the forward process run on the
optimal distribution $r$.

\subsubsection{EP rate}

\label{app:classical-eprate-disc}

Consider a discrete-state classical system which evolves according
to a Markovian master equation,
\[
{\textstyle \frac{d}{dt}}p_{j}(t)=\sum_{i}p_{i}(t)W_{ji}.
\]
In general, the classical EP rate can be written as \citep{esposito2010three}
\begin{align}
\dot{\mathsf{\Sigma}}(p)={\textstyle \frac{d}{dt}}\mathsf{S}(p(t))+\dot{G}(p),\label{eq:classicEPrate}
\end{align}
where $\dot{G}(p)$ is the rate of entropy flow to environment. As
always, the form of $\dot{G}(p)$ will depend on the specifics of
the physical process, but it can generally be written as an expectation
over the microstates. For instance, imagine a system coupled to some
number of thermodynamic reservoirs $\{\nu\}$ which contribute additively
to the overall rate matrix $W$ as $W=\sum_{\nu}W^{\nu}$. Then, the
expression for the rate of entropy flow is
\[
\dot{G}(p)=\sum_{i}p_{i}\sum_{\nu,j}W_{ji}^{\nu}\ln\frac{W_{ji}^{\nu}}{W_{ij}^{\nu}},
\]
where $W_{ji}^{\nu}$ is the transition rate from microstate $i$
to microstate $j$ due to transitions mediated by reservoir $\nu$
(for details, see~\citep{esposito2010three}).

We now show how mismatch cost for classical EP rate can be expressed
in the quantum formalism used in the main text. Define the
following Lindbladian in terms of $W$.
\begin{align}
\mathcal{L}(\rho):=\sum_{i,j}W_{ji}\langle i\vert\rho\vert i\rangle\vert j\rangle\langle j\vert,
\end{align}
Next, define a quantum functional corresponding to the entropy flow
rate in terms of $\dot{G}$,
\begin{align}
\dot{Q}(\rho):=\dot{G}(p^{\rho}),\label{eq:efrapp}
\end{align}
where $p^{\rho}$ is defined as in \ref{eq:PXpdef}.

Given these definitions, consider a mixed state $\rho^{p}=\sum_{i}p_{i}\vert i\rangle\langle i\vert\in\mathcal{D}_{P}$
that represents a classical distribution $p$. Applying the Lindbladian $\mathcal{L}$ to $\rho^{p}$ is equivalent to evolving $p$
under the classical rate matrix,
\begin{align}
\mathcal{L}(\rho^{p})=\sum_{i,j}W_{ji}p_{i}\vert j\rangle\langle j\vert=\sum_{j}\big({\textstyle \frac{d}{dt}}p_{j}(t)\big)\vert j\rangle\langle j\vert.
\end{align}
Similarly, the quantum entropy flow rate obeys $\dot{Q}(\rho^{p})=\dot{G}(p)$,
as expected, and is a linear functional since $\dot{G}$ is an expectation.
Therefore, one can analyze classical instantaneous mismatch cost using
\ref{eq:EPr-equality,eq:ineqConvex-1}, by defining the Lindbladian
$\mathcal{L}$ and entropy flow rate functional $\dot{Q}$ as above,
and by restricting attention to the set of states in $\mathcal{D}_{P}$.

Note that it is also possible to consider instantaneous mismatch cost
within the subset of probability distributions with support restricted
to some subset of microstates $S\subseteq X$. This can be done by
choosing $P=\{\vert i\rangle\langle i\vert\}_{i\in S}$ to be the
corresponding subset of pure states, and then analyzing instantaneous mismatch cost
within the resulting set of diagonal states $\mathcal{D}_{P}$.

\subsection{Classical processes in continuous phase space}

\label{app:classical-continuous}

Above we showed that mismatch cost for discrete-state classical systems
follows as a special case of our quantum analysis. However, the mapping
between quantum and continuous-state classical system is not as straightforward,
because it is not generally possible to represent a continuous probability
distribution in terms of a density operator over a separable Hilbert
space. Nonetheless, as we show in this appendix, the same proof techniques
used to derive our quantum results can also be used to derive mismatch
cost for continuous-state classical processes, as long as an appropriate
``translation'' is carried out.

We start with some definitions. Let $X\subseteq\mathbb{R}^{n}$ indicate
the continuous-state space of a classical system. This state space
can represent the configuration space of the system (only position d.o.f.s),
as might be appropriate for a system with overdamped
dynamics, or the full phase space of the system (both position and
momentum d.o.f.s), as might be appropriate for a system with underdamped dynamics.
In this subsection, we use the
term ``probability distribution'' to refer to a probability density
function.

\subsubsection{Integrated EP}

\label{app:contIntEP}

Consider a continuous-state system that undergoes a driving protocol
over some time interval $t\in[0,\tau]$, while coupled to some thermal
reservoir(s). As above, we use $\mathrm{P}(\bm{x}\vert x_{0})$ and
$\tilde{\mathrm{P}}(\tilde{\bm{x}}\vert\tilde{x}_{\tau})$ to indicate
the conditional trajectory distributions under the forward and backward
protocols, respectively. We will sometimes write the map from initial
to final probability distributions in operator notation as $p'=Tp$,
where the transition operator $T$ is defined in terms of the conditional
probability density as $[Tp](x_{\tau})=\int\mathrm{P}(x_{\tau}\vert x_{0})p(x_{0})\,dx_{0}$.

We will consider the following two classical EP-type functions. The
first is a slightly generalized form of \ref{eq:EPclassicalKL},
\begin{align}
\mathsf{\Sigma}(p) & =D\big(\mathrm{P}(\bm{X}\vert X_{0})p(X_{0})\Vert\tilde{\mathrm{P}}(\tilde{\bm{X}}\vert\tilde{X}_{\tau})p'(X_{\tau})\big)+G(p),\label{eq:contm1}
\end{align}
where $G'$ is any lower-semicontinuous linear functional (lower-semicontinuity
is taken to be in the topology of total variation).

The second is
a slightly generalized form of \ref{eq:contm2PRE},
\begin{align}
\mathsf{\Sigma}(p) & =D({p'(X_{\tau},Y_{\tau})}\Vert{p'(X_{\tau})q({Y_{\tau}\vert X_{\tau}})})+G'(p),\label{eq:contm2}
\end{align}
where $q(y_{\tau}\vert x_{\tau})$ is any conditional distribution
of bath states given system states and $G'$ is any lower-semicontinuous
linear functional.
As discussed near \ref{eq:contm2PRE}, this definition applies only when the system and environment evolve together in a Hamiltonian manner in the full phase space, so that the map from the initial to the final distribution is volume-preserving.

Note that in principle this conditional distribution
may be independent of $X$, in which case the right hand side of \ref{eq:contm2}
would have the form of $D({p'(X_{\tau},Y_{\tau})}\Vert{p'(X_{\tau})q(Y_{\tau})})+G'(p)$,
in complete analogy to \ref{eq:functype2}. (As in the other setting
we consider in this paper, the generalization to any such linear functional
allows us to consider various ``EP-type'' functions in the setting
of continuous-state classical systems, including not only EP but also
nonadiabatic EP, entropy gain, etc., see discussion in \ref{sec:EP-Type-functionals}).

Our results below apply to both forms of classical EP, \ref{eq:contm1}
and \ref{eq:contm2}. This is not surprising, as for Hamiltonian systems the two forms can
be shown to be mathematically equivalent up to the choice of the arbitrary
linear functions $G$and $G'$. This is proved in \ref{prop:formequiv-1}
below, which is the classical equivalent of \ref{prop:formequiv}.

Using these definitions, we show that our results for mismatch cost
for integrated EP apply to continuous-state classical systems. We
do so by using the exact same proofs as for the quantum case, as found
in \ref{app:proofs-integrated}, with the following replacements:
\begin{enumerate}[wide,labelindent=0pt,labelwidth=!]
\item The quantum EP $\Sigma$ should be re-interpreted as the classical
EP $\mathsf{\Sigma}$ (in particular, \ref{eq:functype1}
can be re-interpreted as \ref{eq:contm1}, while \ref{eq:functype2}
can be reinterpreted as \ref{eq:contm2}).
\item The quantum relative entropy $S(\cdot\Vert\cdot)$ should be re-interpreted
as the classical relative entropy, $D(\cdot\Vert\cdot)$. Similarly,
the change of quantum relative entropy under the quantum channel $\Phi$,
$\Delta S({\rho}\Vert{\varphi})=S(\Phi(\rho)\Vert\Phi(\varphi))-S(\rho\Vert\varphi)$,
should be re-interpreted as the change of KL divergence under the
conditional probability density $T$,
\[
\Delta D(p\Vert r)=D(Tp\Vert Tr)-D(p\Vert r).
\]
\item The set of quantum states $\mathcal{D}$ should be re-interpreted
as the set of probability density functions over $X$. $\mathcal{D}_{P}$
should be re-interpreted as the set of probability density functions
with support limited to some measurable subset $P\subseteq X$.
\item The quantum operator notation $p\ge\alpha r$ should be re-interpreted
to mean $p(x)\ge\alpha r(x)$ for all $x\in X$.
\item References to three propositions, which concern properties quantum
relative entropy and quantum EP, should be replaced by references
to the following propositions (proved below in \ref{app:auxClInt})
which prove analogous properties of KL divergence and classical EP
for continuous state spaces:
\begin{enumerate}
\item \ref{prop:KLpropCl} replaces \ref{prop:KLprop},
\item \ref{prop:epConvMix-1} replaces \ref{prop:epConvMix},
\item \ref{prop:epLSC-1} replaces \ref{prop:epLSC}.
\end{enumerate}
\end{enumerate}
By making these replacement, one can re-use the proofs
of \ref{thm:genf,thm:ineqs,thm:eqresult} to derive expressions
of mismatch cost for continuous-state classical systems rather than
quantum systems. First, consider any pair of distribution $p,r$ such
that $\mathsf{\Sigma}(p),\mathsf{\Sigma}(r),D(p\Vert r)<\infty$.
Then, by \ref{thm:genf}, the directional derivative of $\mathsf{\Sigma}$
at $p$ in the direction of $r$ obeys
\begin{equation}
{\textstyle \partial_{\lambda}^{+}}\mathsf{\Sigma}(r(\lambda))\vert_{\lambda=0}=\mathsf{\Sigma}(p)-\mathsf{\Sigma}(r)+\Delta D(p\Vert r),\label{eq:ddCLAPP}
\end{equation}

This equation is the starting point for deriving various expressions
for mismatch cost. Let $\mathcal{D}_{P}$ indicate the set of distributions
with support limited to some arbitrary measurable subset $P\subseteq X$,
and consider any $p\in\mathcal{D}_{P}$ and $r\in\mathop{\arg\min}_{w\in\mathcal{D}_{P}}\mathsf{\Sigma}(w)$
such that $D(p\Vert r)<\infty$. \ref{thm:eqresult} then shows that
\begin{equation}
\mathsf{\Sigma}(p)-\mathsf{\Sigma}(r)=-\Delta D(p\Vert r),\label{eq:appClMEP}
\end{equation}
which is the classical analogue of \ref{eq:equalityBasis}. More generally,
let $\mathcal{S}\subseteq\mathcal{D}$ be any convex subset of distributions.
Then, by \ref{thm:ineqs}, for any $p\in\mathcal{S}$ and $r_{\mathcal{S}}\in\mathop{\arg\min}_{w\in\mathcal{S}}\mathsf{\Sigma}(w)$
such that $D(p\Vert r_{\mathcal{S}})<\infty$,
\begin{equation}
\mathsf{\Sigma}(p)-\mathsf{\Sigma}(r_{\mathcal{S}})\ge-\Delta D(p\Vert r_{\mathcal{S}}),\label{eq:ineqConvex-2}
\end{equation}
with equality if $(1-\lambda)r_{\mathcal{S}}+\lambda p\in\mathcal{S}$
for some $\lambda<0$. Since $\mathsf{\Sigma}(r_{\mathcal{S}})\ge0$
by the second law, \ref{eq:ineqConvex-2} implies the EP bound
\begin{equation}
\mathsf{\Sigma}(p)\ge-\Delta D(p\Vert r_{\mathcal{S}}).
\end{equation}

We do not prove any result about the support of the optimizer $r\in\mathop{\arg\min}_{w}\mathsf{\Sigma}(w)$
for continuous-state classical systems (as we did for quantum systems
in \ref{thm:connectedmeansfullsupport}), instead leaving this for
future work.

\subsubsection{Fluctuating EP}

\label{app:classical-fluct-cont}

Here we show that our results for fluctuating mismatch cost also apply
to continuous-state classical systems. The underlying logic of the
derivation is the same as for the quantum case, though we slightly
modify our notation.

Consider a continuous-state classical system that undergoes a physical
process, which starts from the initial distribution $p$ and ends
on the final distribution $p'=Tp$. In general, the fluctuating EP
incurred by a continuous-state trajectory $\bm{x}$ can be expressed
as \citep{seifert2012stochastic}
\[
\sigma_{p}(\bm{x})=\ln p(x_{0})-\ln p'(x_{\tau})+q(\bm{x}),
\]
where $q(\bm{x})$ is the entropy flow in coupled reservoirs incurred
by trajectory $\bm{x}(t)$.

Now let $r$ indicate the initial probability distribution that minimizes
EP, so that the following mismatch cost relationship holds:
\begin{equation}
\mathsf{\Sigma}(p)-\mathsf{\Sigma}(r)=-\Delta D(p\Vert r).\label{eq:mnnn2}
\end{equation}
As in the main text, we define fluctuating mismatch cost as the difference
between the fluctuating EP incurred by the trajectory $\bm{x}$ under
the actual initial distribution $p$ and the optimal initial distribution
$r$,
\begin{multline}
\sigma_{p}(\bm{x})-\sigma_{r}(\bm{x})=[-\ln p'(x_{\tau})+\ln p(x_{0})]\\
-[-\ln r'(x_{\tau})+\ln r(x_{0}))],\label{eq:clfluctm}
\end{multline}
where $r'=Tr$, which is the classical analogue of \ref{eq:mDef0}.
It is easy to verify that \ref{eq:clfluctm} is the proper trajectory-level
expression of mismatch cost,
\[
\langle\sigma_{p}-\sigma_{r}\rangle_{\mathrm{P}(\bm{x}\vert x_{0})p(x_{0})}=-\Delta D(p\Vert r)=\mathsf{\Sigma}(p)-\mathsf{\Sigma}(r).
\]
Using a derivation similar to the one in \ref{app:fluctuating}, it
can also be shown that \ref{eq:clfluctm} obeys an integral fluctuation
theorem (IFT),
\begin{equation}
\langle e^{-(\sigma_{p}-\sigma_{r})}\rangle_{\mathrm{P}(\bm{x}\vert x_{0})p(x_{0})}=\gamma,\label{eq:clift}
\end{equation}
where the $\gamma$ correction factor is given by the formula in \ref{eq:classicalGamma}
(with summation replaced by integrals).

Finally, some simple algebra shows that fluctuating mismatch cost
can also be written in terms of the time-asymmetry between the forward
conditional probability distribution $\mathrm{P}(x_{\tau}\vert x_{0})$
and its Bayesian inverse $\mathrm{P}(x_{\tau}\vert x_{0})\frac{r(x_{0})}{r'(x_{\tau})}$,
as in \ref{eq:trm2} and \ref{eq:clPetz}.

\subsubsection{EP rate}

\label{app:classical-EPrate-cont}

Consider a system that evolves in continuous-time according to a Markovian
dynamical generator $L$, which we write generically as
\begin{equation}
\dot{p}(x):=\partial_{t}p(x,t)=Lp.\label{eq:cld}
\end{equation}
For example, this generator may represent an (underdamped or overdamped)
Fokker-Planck operator.

In classical stochastic thermodynamics, the EP rate incurred by distribution
$p$ is then given by \citep{Seifert2005,van2010three}
\begin{equation}
\dot{\mathsf{\Sigma}}(p)={\textstyle \frac{d}{dt}}\mathsf{S}(p(t))+\dot{G}(p)\label{eq:EPrAppCl}
\end{equation}
where the first term indicates rate of the increase of the (continuous)
entropy,
\[
\mathsf{S}(p):=-\int p(x)\ln p(x)\,dx,
\]
while the second term $\dot{G}$ reflects the rate of entropy flow.
While the particular form of $\dot{G}(p)$ will depend on the specific
setup, it has the general form of an expectation over some function
defined over the microstates, which is a linear functional of $p$.

Our results for instantaneous mismatch cost apply to continuous-state
classical systems. In fact, one can use the same proofs as for the
quantum case, as found in \ref{app:EPrate}, while making the quantum-to-classical
substitutions \emph{1-5} described in \ref{app:contIntEP}. We will
also need to make the same technical assumptions regarding the EP
rate as we made in \ref{app:EPrate}: the symmetry of partial derivatives
as in \ref{cond:eprate1}, and the finiteness and continuous differentiability
as in \ref{cond:eprate2}.

Consider any pair of distribution $p,r$ such that $\dot{\mathsf{\Sigma}}(p)<\infty,\dot{\mathsf{\Sigma}}(r)<\infty,D(p\Vert r)<\infty$.
Using the same derivation as in \ref{eq:EPratedd}, the directional
derivative of $\dot{\mathsf{\Sigma}}$ at $p$ in the direction of
$r$ obeys
\begin{equation}
{\textstyle \partial_{\lambda}^{+}}\dot{\mathsf{\Sigma}}(r(\lambda))\vert_{\lambda=0}=\dot{\mathsf{\Sigma}}(p)-\dot{\mathsf{\Sigma}}(r)+{\textstyle \frac{d}{dt}}D(p(t)\Vert r(t)),\label{eq:ddCLAPP-1}
\end{equation}
which allows us to derive various expressions for mismatch cost. In
particular, let $\mathcal{D}_{P}$ indicate set of distributions with
support limited to some arbitrary measurable subset $P\subseteq X$.
Consider $r\in\mathop{\arg\min}_{w\in\mathcal{D}_{P}}\dot{\mathsf{\Sigma}}(w)$
and any $p\in\mathcal{D}_{P}$ such that $p\ge\alpha r$ for some
$\alpha>0$. \ref{thm:eprate} then shows that
\begin{equation}
\dot{\mathsf{\Sigma}}(p)-\dot{\mathsf{\Sigma}}(r)=-{\textstyle \frac{d}{dt}}D(p(t)\Vert r(t)),\label{eq:appClMEP-1}
\end{equation}
which is the classical analogue of \ref{eq:EPr-equality}. More generally,
let $\mathcal{S}\subseteq\mathcal{D}$ be any convex subset of distributions.
By \ref{thm:ineqs}, for any $p\in\mathcal{S}$ and $r_{\mathcal{S}}\in\mathop{\arg\min}_{w\in\mathcal{S}}\dot{\mathsf{\Sigma}}(w)$
such that $D(p\Vert r_{\mathcal{S}})<\infty$,
\begin{equation}
\dot{\mathsf{\Sigma}}(p)-\dot{\mathsf{\Sigma}}(r_{\mathcal{S}})\ge-{\textstyle \frac{d}{dt}}D(p(t)\Vert r_{\mathcal{S}}(t))\label{eq:ineqConvex-2-2}
\end{equation}
with equality if $(1-\lambda)r_{\mathcal{S}}+\lambda p\in\mathcal{S}$
for some $\lambda<0$. Since $\dot{\mathsf{\Sigma}}(r_{\mathcal{S}})\ge0$
by the second law, \ref{eq:ineqConvex-2-2} implies the EP rate bound
\[
\dot{\mathsf{\Sigma}}(p)\ge-{\textstyle \frac{d}{dt}}D(p(t)\Vert r_{\mathcal{S}}(t)).
\]

We do not prove any results about the support of the optimizer $r\in\mathop{\arg\min}_{w}\dot{\mathsf{\Sigma}}(w)$
for continuous-state classical systems (as we did for quantum systems
in \ref{prop:epr-fsupp}), instead leaving this for future work.

\subsubsection{Properties of KL divergence and classical EP for continuous-state
systems}

\label{app:auxClInt}

We now state several (mostly well-known) results about classical EP
and relative entropy in continuous-state spaces. These results serve the role of
\ref{prop:KLprop,prop:epConvMix,prop:epLSC} for continuous-state
classical systems.
\begin{prop}
\label{prop:KLpropCl}For any $p,r\in\Omega$ and conditional probability
density $T(x'\vert x)$, the classical relative entropy $D(p\Vert r)$
obeys the following properties:
\end{prop}

\begin{enumerate}
\item[\klconvexprop] $D(p\Vert r)$ is jointly convex in both arguments.
\item[\kllimAaProp] $\lim_{{\lambda\to0^{+}}}D(p\Vert(1-\lambda)r+\lambda p)=D(p\Vert r).$
\item[\kllimBaProp] If $D(p\Vert r)<\infty$, then
\begin{equation}
\lim_{{\lambda\to0^{+}}}\frac{1-\lambda}{\lambda}D(r\Vert(1-\lambda)r+\lambda p)=0.\label{eq:lm9A-1}
\end{equation}
\item[\klupperboundprop] If $\ensuremath{r\ge\alpha p}$ and some $\ensuremath{\alpha>0}$,
\begin{align}
D(p\Vert r)\le-\ln\alpha<\infty\text{.}\label{eq:boundKL-1}
\end{align}
\item[\klmonoprop] \emph{Monotonicity:} if $D(p\Vert r)<\infty$, then \emph{
\[
\Delta D(p\Vert r):=D(Tp\Vert Tr)-D(p\Vert r)\le0.
\]
}
\end{enumerate}
\begin{proof}
\klconvexprop Proved in \citep{polyanskiyLectureNotesInformation2019}.\\

\noindent \noindent \kllimAaProp It is clear that $\lim_{\lambda\to0}(1-\lambda)r+\lambda p=r$
in the topology of total variation distance. Note that KL divergence
obeys monotonicity, convexity in both arguments \citep{cover_elements_2006}
and lower-semicontinuity in the topology of weak convergence \citep{posnerRandomCodingStrategies1975}
(thus also in the topology of total variation distance, which is stronger).
The result then follows from \mycitep[Corollary~7.5.1]{rockafellarConvexAnalysis1970}.\\

\noindent \noindent \kllimBaProp Define $f(\lambda):=-\frac{1-\lambda}{\lambda}\ln(1-\lambda)$
and then write
\begin{align}
& \lim_{{\lambda\to0^{+}}}\frac{1-\lambda}{\lambda}D(r\Vert(1-\lambda)r+\lambda p)\nonumber \\
& =\lim_{{\lambda\to0^{+}}}f(\lambda)\lim_{{\lambda\to0^{+}}}\frac{D(r\Vert(1-\lambda)r+\lambda p)}{-\ln(1-\lambda)}\\
& =\lim_{{\lambda\to0^{+}}}\frac{D(r\Vert(1-\lambda)r+\lambda p)}{-\ln(1-\lambda)},\label{eq:mz0a-1}
\end{align}
where we used that $\lim_{{\lambda\to0^{+}}}f(\lambda)=1$ from L'H\^{o}pital's
rule. A bit of rearranging then gives
\begin{align*}
& \frac{D(r\Vert(1-\lambda)r+\lambda p)}{-\ln(1-\lambda)}\\
& =\int r(x)\frac{\ln\Big((1-\lambda)+\lambda p(x)/r(x)\big)}{\ln(1-\lambda)}dx
\end{align*}
Note that $\left|\ln[(1-\lambda)+\lambda z]\right|\le\left|1-z\right|$for
$\lambda\in[0,1-1/e)$ and $z>0$. That implies that for $\lambda\in[0,1-1/e)$,
\begin{align*}
r(x)\left|\ln[(1-\lambda)+\lambda p(x)/r(x)]\right| & \le r(x)\left|1-p(x)/r(x)\right|\\
& =\left|r(x)-p(x)\right|.
\end{align*}
Then, by the dominated convergence theorem, one can move the limit
inside the integral:
\begin{align*}
& \lim_{{\lambda\to0^{+}}}\frac{D(r\Vert(1-\lambda)r+\lambda p)}{-\ln(1-\lambda)}\\
& =\int r(x)\lim_{{\lambda\to0^{+}}}\frac{\ln\Big((1-\lambda)+\lambda p(x)/r(x)\big)}{\ln(1-\lambda)}dx\\
& =\int r(x)\left[1-\frac{p(x)}{r(x)}\right]dx\\
& =1-\int\mathbf{1}_{\mathrm{supp}\,r}(x)p(x)\,dx\\
& =0,
\end{align*}
where in the last line we used that $D(p\Vert r)<\infty$ implies
that $\ensuremath{\mathrm{supp}\,p\subseteq\mathrm{supp}\,r}$ (by
the definition of KL divergence in \ref{eq:relentDef-1}). Plugging
into \ref{eq:mz0a-1} gives \ref{eq:lm9A-1}.\\

\noindent \noindent \klupperboundprop Follows from a simple manipulation
of \ref{eq:relentDef-1}.\\

\noindent \noindent \klmonoprop Follows from the monotonicity property
of KL divergence, i.e., the ``data processing inequality'' \citep{polyanskiyLectureNotesInformation2019}.
\end{proof}
The next result shows that the definitions in \ref{eq:contm1} and
\ref{eq:contm2} are equivalent. We note that this result applies under the assumption that the
system and environment jointly evolve in a Hamiltonian manner, so that \ref{eq:contm2} is a valid definition of integrated EP.
We will use $g:X\times Y\to X\times Y$
to indicate the invertible volume-preserving evolution function specified by the Hamiltonian
dynamics over system and environment from time $t=0$ to time $t=\tau$.

\begin{prop}
\label{prop:formequiv-1}Given the definitions of the terms in \ref{eq:contm1}
and \ref{eq:contm2}, and assuming all relevant terms are finite,
\begin{align}
& D\big(\mathrm{P}(\bm{X}\vert X_{0})p(X_{0})\Vert\tilde{\mathrm{P}}(\tilde{\bm{X}}\vert\tilde{X}_{\tau})p'(X_{\tau})\big)+G(p)\nonumber \\
& \quad=D({p'(X_{\tau},Y_{\tau})}\Vert{p'(X_{\tau})q({Y_{\tau}\vert X_{\tau}})})+G'(p),\label{eq:gd2}
\end{align}
where $G'(p):=G(p)+\int p(x_{0})f(x_{0})\,dx_{0}$ and $f:X\to\mathbb{R}$
is defined as
\[
f(x_{0})\!:=\!\Big\langle\!\ln\frac{\mathrm{P}(\bm{x}\vert x_{0})}{\tilde{\mathrm{P}}(\tilde{\bm{x}}\vert\tilde{x}_{\tau})}\!\Big\rangle_{\mathrm{P}(\bm{x}\vert x_{0})}\!+\!\Big\langle\!\ln\frac{q(y\vert x)\vert_{g(x_{0},y_{0})}}{q(y_{0}\vert x_{0})}\!\Big\rangle_{q(y_{0}\vert x_{0})}.
\]
\end{prop}

\begin{proof}
Rewrite the KL divergence in \ref{eq:contm1} as
\begin{multline}
D\big(\mathrm{P}(\bm{X}\vert X_{0})p(X_{0})\Vert\tilde{\mathrm{P}}(\tilde{\bm{X}}\vert\tilde{X}_{\tau})p'(X_{\tau})\big)=\\
S(p'(X_{\tau}))-S(p(X_{0}))+\Big\langle\ln\frac{\mathrm{P}(\bm{x}\vert x_{0})}{\tilde{\mathrm{P}}(\tilde{\bm{x}}\vert\tilde{x}_{\tau})}\Big\rangle_{\mathrm{P}(\bm{x}\vert x_{0})p(x_{0})}.\label{eq:gda2}
\end{multline}
One can also rewrite the KL divergence in \ref{eq:contm2} as
\begin{multline}
D({p'(X_{\tau},Y_{\tau})}\Vert{p'(X_{\tau})q({Y_{\tau}\vert X_{\tau}})})=\\
S(p'(X_{\tau}))-S(p'(X_{\tau},Y_{\tau}))-\Big\langle\ln q(y_{\tau}\vert x_{\tau})\Big\rangle_{p'(x_{\tau},y_{\tau})}.\label{eq:gd}
\end{multline}
The second entropy term can be written as
\begin{align*}
& S(p'(X_{\tau},Y_{\tau}))=S(p(X_{0},Y_{0}))\\
& \qquad=S(p(X_{0}))-\int p(x_{0})q(y_{0}\vert x_{0})\ln q(y_{0}\vert x_{0})dx_{0}dy_{0},
\end{align*}
where we first used the invariance of differential entropy under volume-preserving
transformations, and in the second line the chain rule for entropy.
One can then rewrite the last term in \ref{eq:gd} as
\begin{align*}
& \Big\langle\ln q(y_{\tau}\vert x_{\tau})\Big\rangle_{p'(x_{\tau},y_{\tau})}\\
& =\int p'(x_{\tau},y_{\tau})\ln q(y_{\tau}\vert x_{\tau})\,dx_{\tau}dy_{\tau}\\
& =\int p(x_{0},y_{0})\ln q(y\vert x)\vert_{(x,y)=g(x_{0},y_{0}),}\,dx_{0}dy_{0},
\end{align*}
where we performed a change of variables and used that $p(x_{0},y_{0})=p'(x_{\tau},y_{\tau})\vert_{(x_{\tau},y_{\tau})=g(x_{0},y_{0})}$.
Combining lets us rewrite the right hand side of \ref{eq:gd} as
\[
S(p'(X_{\tau}))-S(p(X_{0}))-\Big\langle\ln\frac{q(y\vert x)\vert_{g(x_{0},y_{0})}}{q(y_{0}\vert x_{0})}\Big\rangle_{q(y_{0}\vert x_{0})p(x_{0})}.
\]
Combining with \ref{eq:gda2} and rearranging gives \ref{eq:gd2}.
\end{proof}
We now prove the classical analogues of \ref{prop:epConvMix} and
\ref{prop:epLSC}.
\begin{prop}
\label{prop:epConvMix-1}Consider a classical EP-type function $\mathsf{\Sigma}$,
as in \ref{eq:contm1} or \ref{eq:contm2}. Then, for any $p,r\in\mathcal{D}$,
$\lambda\in(0,1)$ such that $\Sigma(r(\lambda))<\infty$:
\begin{multline}
(1-\lambda)\mathsf{\Sigma}(r)+\lambda\mathsf{\Sigma}(p)-\mathsf{\Sigma}(r(\lambda))=\\
-(1-\lambda)\Delta D(r\Vert r(\lambda))-\lambda\Delta D(p\Vert r(\lambda)).\label{eq:conv1-1}
\end{multline}
\end{prop}

\begin{proof}
\emph{EP-type functions as in \ref{eq:contm1}.} For notational convenience,
define
\[
f(\bm{x})=\ln\frac{\mathrm{P}(\bm{x}\vert x_{0})}{\tilde{\mathrm{P}}(\tilde{\bm{x}}\vert\tilde{x}_{\tau})}.
\]
We will also use shorthand like $\langle\cdot\rangle_{p}$ to indicate
expectation under the distribution $\mathrm{P}(\bm{x}\vert x_{0})p(x_{0})$.
Then, write the EP incurred by initial distribution $r(\lambda)$
as
\begin{align}
& \mathsf{\Sigma}(r(\lambda))\nonumber \\
& =D\big(\mathrm{P}(\bm{X}\vert X_{0})r(\lambda)(X_{0})\Vert\tilde{\mathrm{P}}(\tilde{\bm{X}}\vert\tilde{X}_{\tau})r'(\lambda)(X_{\tau})\big)+G'(r(\lambda))\nonumber \\
& =\Big\langle\ln\frac{r(\lambda)(x_{0})}{r'(\lambda)(\tilde{x}_{\tau})}+f(\bm{x})\Big\rangle_{r(\lambda)}\!\!\!+G'(r(\lambda))\nonumber \\
& =(1-\lambda)\left[\Big\langle\ln\frac{r(\lambda)(x_{0})}{r'(\lambda)(\tilde{x}_{\tau})}+f(\bm{x})\Big\rangle_{r}\!\!\!+G'(r)\right]\label{eq:ghj2}\\
& \qquad+\lambda\left[\Big\langle\ln\frac{r(\lambda)(x_{0})}{r'(\lambda)(\tilde{x}_{\tau})}+f(\bm{x})\Big\rangle_{p}\!\!\!+G'(p)\right].\label{eq:ghj3}
\end{align}
where we used that the expectation and $G'$ are linear. Now consider
that the change of KL divergence between $r$ and $r(\lambda)$ can
be written as
\begin{align*}
\Delta D(r\Vert r(\lambda)) & =\Big\langle\ln\frac{r(\lambda)(x_{0})}{r'(\lambda)(\tilde{x}_{\tau})}-\ln\frac{r(x_{0})}{r'(\tilde{x}_{\tau})}\Big\rangle_{r}.
\end{align*}
By adding and subtracting $\Delta D(r\Vert r(\lambda))$ to the bracketed
term in \ref{eq:ghj2}, one can rewrite that term as
\begin{multline*}
\Delta D(r\Vert r(\lambda))+\Big\langle\ln\frac{r(x_{0})}{r'(\tilde{x}_{\tau})}+f(\bm{x})\Big\rangle_{r}+G'(r)\\
=\Delta D(r\Vert r(\lambda))+\mathsf{\Sigma}(r).
\end{multline*}
Performing a similar rewriting of the bracketed term in \ref{eq:ghj3},
and then combining with the above expression for $\mathsf{\Sigma}(r(\lambda))$,
gives
\begin{multline*}
\mathsf{\Sigma}(r(\lambda))=\\
(1-\lambda)[\Delta D(r\Vert r(\lambda))+\mathsf{\Sigma}(r)]+\lambda[\Delta D(p\Vert r(\lambda))+\mathsf{\Sigma}(p)].
\end{multline*}
This leads to \ref{eq:conv1-1} after some simple rearrangement.

\vspace{5pt}

\emph{EP-type functions as in \ref{eq:contm2}.} For EP-type functions
as in \ref{eq:contm2}, the derivation proceeds in exactly the same
manner as the derivation of \ref{prop:epConvMix} for quantum EP-type
functions as in \ref{eq:functype2} (up to a change of quantum notation
for classical probability notation). For this reason, we omit details
and refer the reader to the proof of \ref{prop:epConvMix}. We will
only mention the classical analogues of two quantum identities used
in that derivation: ``Donald's identity'' as stated in \ref{eq:donald}
and \mycitep[Thm.~3.12]{petzQuantumInformationTheory2008} as used
in \ref{eq:petzTr}. Donald's identity is usually called the ``compensation
identity'' in classical information theory, which can be found as
\mycitep[Lemma~7]{topsoe1979information}. For classical distributions,
the lines after \ref{eq:petzTr} can be derived using the chain rule
for KL divergence,
\begin{align*}
& D({p'(X_{\tau},Y_{\tau})}\Vert{r'(\lambda)(X_{\tau})q(Y_{\tau}\vert X_{\tau})})=\\
& \quad D({p'(X_{\tau})}\Vert{r'(\lambda)(X_{\tau})})+D({p'(Y_{\tau}\vert X_{\tau})}\Vert{q(Y_{\tau}\vert X_{\tau})})=\\
& D({p'(X_{\tau})}\Vert{r'(\lambda)(X_{\tau})})+D({p'(X_{\tau},Y_{\tau})}\Vert{p'(X_{\tau})q(Y_{\tau}\vert X_{\tau})}).
\end{align*}
\end{proof}
\begin{prop}
\label{prop:epLSC-1}Consider a classical EP-type function $\mathsf{\Sigma}$,
as in \ref{eq:contm1} and \ref{eq:contm2}. For any $p,r\in\mathcal{D}_{P}$
with $\mathsf{\Sigma}(p),\mathsf{\Sigma}(r),D(p\Vert r)<\infty$,
there is a sequence $\{p_{n}\}\subset\mathcal{D}_{P}$ such that:
\end{prop}

\begin{enumerate}
\item[I.] For all $n$, there is some $\alpha_{n}>0$ such that $p_{n}\ge\alpha_{n}r$.
\item[II.] ${\displaystyle \liminf_{n\to\infty}\mathsf{\Sigma}(p_{n})+\Delta D(p_{n}\Vert r)\ge\mathsf{\Sigma}(p)+\Delta D(p\Vert r)}$.
\end{enumerate}
\begin{proof}
Let $P$ and $R$ be two probability measures over the same measurable
space $(X,\mathcal{A})$ that correspond to the densities $p$ and
$r$. By the Gelfand-Yaglom-Perez theorem \citep{pinskerInformationInformationStability1964,polyanskiyLectureNotesInformation2019},
there is a sequence of measurable functions (i.e., ``quantizers'')
$f_{1},f_{2},\dots$ over $X$ such that each $f_{i}(X)$ is a finite
set, and
\begin{equation}
\lim_{n\to\infty}D(P(f_{n}(X))\Vert R(f_{n}(X)))=D(p\Vert r).\label{eq:n12}
\end{equation}
For each $n$, define the following probability density function:
\[
p_{n}(x):=\begin{cases}
r(x\vert f_{n}(x))P(f_{n}(x)) & \text{if \ensuremath{r(x)>0}}\\
0 & \text{otherwise}.
\end{cases}.
\]
In words, $p_{n}$ has the same distribution as $p$ over the coarse-grained
quantized bins $f_{n}(X)$, and the same conditional distribution
as $r$ within each quantized bin.
Note that $\mathrm{supp}\,p\subseteq\mathrm{supp}\,r$,
which follows from $D(p\Vert r)<\infty$. Thus, it is easy to verify that for each
$n$, $\mathrm{supp}\,p_{n}\subseteq\mathrm{supp}\,r$, therefore
$p_{n}\in\mathcal{D}_{P}$. It is also easy to verify that for each
$n$ and any $x\in\mathrm{supp}\,p_{n}$,
\[
\frac{p_{n}(x)}{r(x)}=\frac{P(f_{n}(x))}{R(f_{n}(x))}\ge\alpha_{n}:=\min_{z}\frac{P(f_{n}(X)=z)}{R(f_{n}(X)=z)}>0,
\]
where the last inequality uses that $f_{n}(X)$ is a finite set and
that $p_{n}(x)>0\implies P(f_{n}(x))>0\implies R(f_{n}(x))>0$ (the
last implication follows from $\mathrm{supp}\,p\subseteq\mathrm{supp}\,r$). This proves (I).

To prove (II), observe that
\[
D(p_{n}\Vert r)=D(P(f_{n}(X))\Vert R(f_{n}(X)))
\]
which follows from \ref{eq:relentDef-1} and some simple algebra.
Along with \ref{eq:n12}, this implies
\begin{equation}
\lim_{n\to\infty}D(p_{n}\Vert r)=D(p\Vert r).\label{eq:lk2}
\end{equation}
Next, consider the KL divergence between $Tp$ and $Tp_{n}$:
\begin{align}
& D(Tp\Vert Tp_{n})\le D(p\Vert p_{n})\nonumber \\
& \quad=D(p(X\vert f_{n}(X))\Vert r(X\vert f_{n}(X)))\nonumber \\
& \quad=D(p\Vert r)-D(P(f_{n}(X))\Vert R(f_{n}(X))),\label{eq:kjs1}
\end{align}
where in the first line we used monotonicity, and in the third line
we used the chain rule for KL divergence \citep{polyanskiyLectureNotesInformation2019}.
Given \ref{eq:n12}, the expression in \ref{eq:kjs1} vanishes in
the $n\to\infty$ limit, so
\begin{equation}
\lim_{n\to\infty}D(p\Vert p_{n})=\lim_{n\to\infty}D(Tp\Vert Tp_{n})=0.\label{eq:nmmm2}
\end{equation}
Note that convergence in KL divergence \citep{harremoesInformationTopologiesApplications2007}
implies convergence in total variation distance (by Pinsker's inequality),
which in turns implies weak convergence. Since KL divergence is lower-semicontinuous
in the topology of weak converge \mycitep[Theorem~1]{posnerRandomCodingStrategies1975},
\begin{equation}
\liminf_{n\to\infty}D(Tp_{n}\Vert Tr)\ge D(Tp\Vert Tr).\label{eq:lk3}
\end{equation}
Finally, in \ref{prop:lscappcl} below we show that classical EP-type
functions, as in \ref{eq:contm1} and \ref{eq:contm2}, obey
\begin{equation}
\liminf_{n\to\infty}\mathsf{\Sigma}(p_{n})\ge\mathsf{\Sigma}(p).\label{eq:appEPLim0-1}
\end{equation}
(II) follows by combining \ref{eq:lk2}, \ref{eq:lk3},
and \ref{eq:appEPLim0-1}.
\end{proof}
\begin{lem}
\label{prop:lscappcl}For any $p,r\in\mathcal{D}_{P}$ with $\mathsf{\Sigma}(p),\mathsf{\Sigma}(r),D(p\Vert r)<\infty$,
let the sequence of distribution $\{p_{n}\}_{n}$ be defined as in
the proof of \ref{prop:epLSC-1}. Then, EP-type functions as in \ref{eq:contm1}
and \ref{eq:contm2} obey $\liminf_{n\to\infty}\mathsf{\Sigma}(p_{n})\ge\mathsf{\Sigma}(p)$.
\end{lem}

\begin{proof}
\emph{EP-type functions as in} \ref{eq:contm1}. Consider the following
limit of KL divergences,
\begin{align*}
& \lim_{n\to\infty}D\big(\mathrm{P}(\bm{X}\vert X_{0})p(X_{0})\Vert\mathrm{P}(\bm{X}\vert X_{0})p_{n}(X_{0})\big)\\
& =\lim_{n\to\infty}D(p\Vert p_{n})=0,
\end{align*}
where we used the chain rule and then \ref{eq:nmmm2}. A similar derivation
shows that
\begin{multline*}
\lim_{n\to\infty}D\big(\tilde{\mathrm{P}}(\tilde{\bm{X}}\vert\tilde{X}_{\tau})p'(X_{\tau})\Vert\tilde{\mathrm{P}}(\tilde{\bm{X}}\vert\tilde{X}_{\tau})p_{n}'(X_{\tau})\big)\\
=\lim_{n\to\infty}D(p'\Vert p_{n}')=0.
\end{multline*}
This shows that $\mathrm{P}(\bm{x}\vert x_{0})p_{n}(x_{0})\to\mathrm{P}(\bm{x}\vert x_{0})p(x_{0})$
and $\tilde{\mathrm{P}}(\tilde{\bm{x}}\vert\tilde{x}_{\tau})p_{n}'(\tilde{x}_{\tau})\to\tilde{\mathrm{P}}(\tilde{\bm{x}}\vert\tilde{x}_{\tau})p'(\tilde{x}_{\tau})$
in KL divergence, thus also in total variation. Then, by lower-semicontinuity
of KL and $G$,
\begin{align*}
& \liminf_{n\to\infty}\mathsf{\Sigma}(p_{n})\\
& =\liminf_{n\to\infty}D\big(\mathrm{P}(\bm{X}\vert X_{0})p_{n}(X_{0})\Vert\tilde{\mathrm{P}}(\tilde{\bm{X}}\vert\tilde{X}_{\tau})p_{n}'(X_{\tau})\big)+G(p_{n})\\
& \ge D\big(\mathrm{P}(\bm{X}\vert X_{0})p(X_{0})\Vert\tilde{\mathrm{P}}(\tilde{\bm{X}}\vert\tilde{X}_{\tau})p'(X_{\tau})\big)+G(p)=\mathsf{\Sigma}(p).
\end{align*}

\vspace{5pt}
\emph{EP-type functions as in} \ref{eq:contm2}. Let $g:X\times Y\to X\times Y$
be the invertible volume-preserving evolution function specified by the Hamiltonian
dynamics over system and environment from time $t=0$ to time $t=\tau$.
Let $p_{n}(x_{0},y_{0})=p_{n}(x_{0})q(y_{0}\vert x_{0})$ and $p_{n}'(x_{\tau},y_{\tau})=p_{n}(x_{0},y_{0})\vert_{(x_{0},y_{0})=g^{-1}(x_{\tau},y_{\tau})}$,
and similarly $p(x_{0},y_{0})=p(x_{0})q(y_{0}\vert x_{0})$ and $p'(x_{\tau},y_{\tau})=p(x_{0},y_{0})\vert_{(x_{0},y_{0})=g^{-1}(x_{\tau},y_{\tau})}$.
Then, consider the following limit of KL divergences:
\begin{align*}
& \lim_{n\to\infty}D({p'(X_{\tau},Y_{\tau})}\Vert p_{n}'({X_{\tau}},{Y_{\tau}}))\\
& =\lim_{n\to\infty}D({p(X_{0},Y_{0})}\Vert p_{n}({X_{0}},{Y_{0}}))\\
& =\lim_{n\to\infty}D(p({X_{0}})q(Y_{0}\vert{X_{0}})\Vert p_{n}({X_{0}})q(Y_{0}\vert{X_{0}}))\\
& =\lim_{n\to\infty}D(p\Vert p_{n})=0,
\end{align*}
where we first used the invariance of KL under invertible transformations,
and in the last line we used the chain rule and then \ref{eq:nmmm2}.
Similarly,
\begin{multline*}
\lim_{n\to\infty}D({p'(X_{\tau})q(Y_{\tau}\vert X_{\tau})}\Vert{p'_{n}(X_{\tau})q(Y_{\tau}\vert X_{\tau})})\\
=\lim_{n\to\infty}D(p'\Vert p_{n}')=0,
\end{multline*}
where we've used the chain rule and \ref{eq:nmmm2}. This shows that
$p_{n}'(x_{\tau},y_{\tau})\to p'(x_{\tau},y_{\tau})$ and $p_{n}'(x_{\tau})q(y_{\tau}\vert x_{\tau})\to p'(x_{\tau})q(y_{\tau}\vert x_{\tau})$
in KL divergence, thus also in total variation. In addition, we know
that $p_{n}\to p$ by \ref{eq:nmmm2}. Then, by lower-semicontinuity
of KL and $G'$,
\begin{align*}
& \liminf_{n\to\infty}\mathsf{\Sigma}(p_{n})\\
& =\liminf_{n\to\infty}[D({p'_{n}(X_{\tau},Y_{\tau})}\Vert{p'_{n}(X_{\tau})q(Y_{\tau}\vert X_{\tau})})+G'(p_{n})]\\
& \ge D({p'(X_{\tau},Y_{\tau})}\Vert{p'(X_{\tau})q({Y_{\tau}\vert X_{\tau}})})+G'(p)=\mathsf{\Sigma}(p).
\end{align*}
\end{proof}
\clearpage
\end{document}